%% file: DL_TITB.tex
\documentclass[journal, 11pt]{IEEEtran}
\onecolumn
\usepackage{cite}
\usepackage{amsmath}
\usepackage{graphicx,amsfonts,amsthm, amssymb,bm,url,color,latexsym,mathtools}
\usepackage{subcaption}
\usepackage{bbm}
\usepackage{microtype}
\usepackage{algorithm}
\usepackage{algorithmicx}
\usepackage{algpseudocode} 
\usepackage{hyperref}
\hypersetup{
    colorlinks=true,%
    citecolor=blue,%
    filecolor=blue,%
    linkcolor=blue,%
    urlcolor=blue
}
\usepackage{verbatim}
\usepackage{framed}

\newtheorem{theorem}{Theorem}[section]
\newtheorem{lemma}[theorem]{Lemma}

\newtheorem{proposition}[theorem]{Proposition}

\usepackage{tikz}
\usepackage{fge}

\renewcommand{\mathbf}{\boldsymbol}

\newcommand{\mb}{\mathbf}
\newcommand{\mc}{\mathcal}

\newcommand{\bb}{\mathbb}
\newcommand{\magnitude}[1]{ \left| #1 \right| } 
\newcommand{\set}[1]{\left\{ #1 \right\}}

\newcommand{\reals}{\bb R}

\newcommand{\eps}{\varepsilon}
\newcommand{\R}{\reals}

\newcommand{\N}{\bb N}

\newcommand{\indicator}[1]{\mathbbm 1_{#1}}

\newcommand{ \Brac }[1]{\left\lbrace #1 \right\rbrace}
\newcommand{ \brac }[1]{\left[ #1 \right]}
\newcommand{ \paren }[1]{ \left( #1 \right) }

\DeclareMathOperator{\trace}{tr}
\DeclareMathOperator{\supp}{supp}

\DeclareMathOperator{\sign}{sign}
\DeclareMathOperator{\grad}{grad}
\DeclareMathOperator{\Hess}{Hess}
\DeclareMathOperator{\mini}{minimize}

\DeclareMathOperator{\st}{subject\; to}

\numberwithin{equation}{section}

\newcommand{\event}{\mc E}

\newcommand{\rconcave}{r_\fgecap}
\newcommand{\Lconcave}{L_\fgecap}
\newcommand{\rconvex}{r_\fgecup}

\newcommand{\Lconvex}{L_\fgecup}

\newcommand{\wh}{\widehat}
\newcommand{\wt}{\widetilde}
\newcommand{\ol}{\overline}

\newcommand{\betaconcave}{\beta_\fgecap}
\newcommand{\rI}{R_{\mathtt{I}}}
\newcommand{\rII}{R_{\mathtt{II}}}
\newcommand{\rIII}{R_{\mathtt{III}}}
\newcommand{\dI}{d_{\mathtt{I}}}
\newcommand{\dII}{d_{\mathtt{II}}}
\newcommand{\dIII}{d_{\mathtt{III}}}
\newcommand{\betagrad}{\beta_{\mathrm{grad}}}

\newcommand{\norm}[2]{\left\| #1 \right\|_{#2}}
\newcommand{\abs}[1]{\left| #1 \right|}
\newcommand{\innerprod}[2]{\left\langle #1,  #2 \right\rangle}
\newcommand{\prob}[1]{\bb P\left[ #1 \right]}
\newcommand{\expect}[1]{\bb E\left[ #1 \right]}

\newcommand{\js}[1]{#1}

\begin{document}
\title{Complete Dictionary Recovery over the Sphere \\ II: Recovery by Riemannian Trust-Region Method}
\author{Ju~Sun,~\IEEEmembership{Student Member,~IEEE,}
        Qing~Qu,~\IEEEmembership{Student Member,~IEEE,}
        and~John~Wright,~\IEEEmembership{Member,~IEEE}
\thanks{JS, QQ, and JW are all with Electrical Engineering, Columbia University, New York, NY 10027, USA. Email: \{js4038, qq2105, jw2966\}@columbia.edu. An extended abstract of the current work has been published in~\cite{sun2015complete_conf}. Proofs of some secondary results are contained in the combined technical report~\cite{sun2015complete_tr}.}
\thanks{Manuscript received xxx; revised xxx.}}

\markboth{IEEE Transaction on Information Theory,~Vol.~xx, No.~xx, xxxx~2016}%
{Sun \MakeLowercase{\textit{et al.}}: Complete Dictionary Recovery over the Sphere}



\maketitle

\begin{abstract}
We consider the problem of recovering a complete (i.e., square and invertible) matrix $\mb A_0$, from $\mb Y \in \R^{n \times p}$ with $\mb Y = \mb A_0 \mb X_0$, provided $\mb X_0$ is sufficiently sparse. This recovery problem is central to theoretical understanding of dictionary learning, which seeks a sparse representation for a collection of input signals and finds numerous applications in modern signal processing and machine learning. We give the first efficient algorithm that provably recovers $\mb A_0$ when $\mb X_0$ has $O\paren{n}$ nonzeros per column, under suitable probability model for $\mb X_0$.

Our algorithmic pipeline centers around solving a certain nonconvex optimization problem with a spherical constraint, and hence is naturally phrased in the language of manifold optimization. In a companion paper~\cite{sun2015complete_a}, we have showed that with high probability our nonconvex formulation has no ``spurious'' local minimizers and around any saddle point the objective function has a negative directional curvature. In this paper, we take advantage of the particular geometric structure, and describe a Riemannian trust region algorithm that provably converges to a local minimizer with from arbitrary initializations. Such minimizers give excellent approximations to rows of $\mb X_0$. The rows are then recovered by linear programming rounding and deflation. 
\end{abstract}


%
\IEEEpeerreviewmaketitle


\input{sec/intro}
\input{sec/algorithm}

\input{sec/main_result}

\input{sec/exp}

\input{sec/discuss}

\input{sec/proof_algorithm}

\input{sec/proof_main}

\input{sec/appendix}

\section*{Acknowledgment}
We thank Dr. Boaz Barak for pointing out an inaccurate comment made on overcomplete dictionary learning using SOS. We thank Cun Mu and Henry Kuo of Columbia University for discussions related to this project. We also thank the anonymous reviewers for their careful reading of the paper, and for comments which have helped us to substantially improve the presentation. JS thanks the Wei Family Private Foundation for their generous support. This work was partially supported by grants ONR N00014-13-1-0492, NSF 1343282, NSF CCF 1527809, NSF IIS 1546411, and funding from the Moore and Sloan Foundations.

\bibliographystyle{IEEEtran}
\bibliography{IEEEabrv,../dl_tit,../ncvx}








\end{document}

%% file: sec/intro.tex
\section{Introduction}
Recently, there is a surge of research studying nonconvex formulations and provable algorithms for a number of central problems in signal processing and machine learning, including, e.g., low-rank matrix completion/recovery~\cite{keshavan2010matrix, jain2013low, hardt2014understanding, hardt2014fast, netrapalli2014non, jain2014fast, sun2014guaranteed, zheng2015convergent, tu2015low, chen2015fast,wei2015guarantees,li2016recovery,gamarnik2016note,wei2016guarantees,zheng2016convergence,yi2016fast,jin2016provable,park2016finding,park2016provable,cherapanamjeri2016nearly,ge2016matrix,bhojanapalli2016global}, phase retreival~\cite{netrapalli2013phase, candes2015phase, chen2015solving, white2015local,sun2016geometric,zhang2016provable,zhang2016reshaped,wang2016solving,kolte2016phase,gao2016gauss,bendory2016non}, tensor recovery~\cite{jain2014provable, anandkumar2014guaranteed, anandkumar2014analyzing, anandkumar2015tensor,hopkins2015fast}, mixed regression~\cite{yi2013alternating, sedghi2014provable}, structured element pursuit~\cite{qu2014finding,hopkins2015fast}, blind deconvolution~\cite{lee2013near,lee2015blind,lee2015rip,cambareri2016non,li2016rapid}, noisy phase synchronization and community detection~\cite{boumal2016non,boumal2016nonconvex,bandeira2016low}, deep learning~\cite{kawaguchi2016deep,soudry2016no}, numerical linear algebra and optimization~\cite{jain2015computing, bhojanapalli2015dropping}. The research efforts are fruitful in producing more practical and scalable algorithms and even significantly better performance guarantees than known convex methods. 

In a companion paper~\cite{sun2015complete_a}, we set out to understand the surprising effectiveness of nonconvex heuristics on the dictionary learning (DL) problem. In particular, we have focused on the complete dictionary recovery (DR) setting: given $\mb Y = \mb A_0 \mb X_0$, with $\mb A_0 \in \R^{n \times n}$ complete (i.e., square and invertible), and $\mb X_0 \in \R^{n \times p}$ obeying an i.i.d. Bernoulli-Gaussian (BG) model with rate $\theta$ (i.e., $[X_0]_{ij} = \Omega_{ij} Z_{ij}$ with $\Omega_{ij} \sim \mathrm{Ber}(\theta)$ and $Z_{ij} \sim \mc N(0, 1)$), recover $\mb A_0$ and $\mb X_0$. In this setting, $\mathrm{row}(\mb Y) = \mathrm{row}(\mb X_0)$, where $\mathrm{row}(\cdot)$ denotes the row space. To first recover rows of $\mb X_0$, we have tried to find the sparsest vectors in $\mathrm{row}(\mb Y)$, and proposed solving the nonconvex formulation
\begin{align} \label{eq:main_obj}
\mini\; \quad f(\mb q; \widehat{\mb Y}) \doteq \frac{1}{p}\sum_{k=1}^p h_{\mu}( \mb q^* \widehat{\mb y}_k) \quad \st \quad \mb q \in \bb S^{n-1}, 
\end{align}
where $\wh{\mb Y}$ is a proxy of $\mb Y$ (i.e., after appropriate processing), $\wh{\mb y}_k$ is the $k$-th column of $\wh{\mb Y}$, and $h_{\mu}(z) \doteq \mu \log \cosh(z/\mu)$ is a (convex) smooth approximation to the absolute-value function. The spherical constraint renders the problem nonconvex. 

Despite the apparent nonconvexity, our prior analysis in~\cite{sun2015complete_a} has showed that all local minimizers of~\eqref{eq:main_obj} are qualitatively equally good, because each of them produces a close approximation to certain row of $\mb X_0$ (Corollary II.4 in~\cite{sun2015complete_a}). So the central issue is how to escape from saddle points. Fortunately, our previous results (Theorem II.3 in~\cite{sun2015complete_a}) imply that \js{all saddle points under consideration are \emph{ridable}, i.e., the associated Hessians have both strictly positive and strictly negative values (see the recapitulation in Section~\ref{sec:geo_reproduce}). Particularly, eigenvectors of the negative eigenvalues are direction of negative curvature, which intuitively serve as directions of local descent. }

\js{Second-order methods can naturally exploit the curvature information to escape from ridable saddle points}. To gain some intuition, consider an unconstrained optimization problem
\begin{align*}
\mini_{\mb x \in \R^n} \phi(\mb x). 
\end{align*} 
The second-order Taylor expansion of $\phi$ at a saddle point $\mb x_0$ is
\begin{align*}
\wh{\phi}(\mb \delta; \mb x_0) = \phi(\mb x_0) + \tfrac{1}{2} \mb \delta^* \nabla^2 \phi(\mb x_0) \mb \delta. 
\end{align*}
When $\mb \delta$ is chosen to align with an eigenvector of a negative eigenvalue $\lambda_{\mathrm{neg}}[\nabla^2 \phi(\mb x_0)] < 0$, it holds that 
\begin{align*}
\wh{\phi}(\mb \delta; \mb x_0) - \phi(\mb x_0) \le -\abs{\lambda_{\mathrm{neg}}[\nabla^2 \phi(\mb x)]} \norm{\mb \delta}{}^2. 
\end{align*}
Thus, minimizing $\wh{\phi}(\mb \delta; \mb x_0)$ returns a direction $\mb \delta_\star$ that tends to decrease the objective $\phi$, provided local approximation of $\wh{\phi}$ to $\phi$ is reasonably accurate. Based on this intuition, we derive a (second-order) Riemannian trust-region algorithm that exploits the second-order information to escape from saddle points and provably returns a local minimizer to~\eqref{eq:main_obj}, from arbitrary initializations. We provide rigorous guarantees for recovering a local minimizer in Section~\ref{sec:algorithm}. 

Obtaining a local minimizer only helps approximate one row of $\mb X_0$. To recover the row, we derive a simple linear programming rounding procedure that provably works. To recover all rows of $\mb X_0$, one repeats the above process based on a carefully designed deflation process. The whole algorithmic pipeline and the related recovery guarantees are provided in Section~\ref{sec:main_result}. Particularly, we show that when $p$ is reasonably large, with high probability (w.h.p.), our pipeline efficiently recovers $\mb A_0$ and $\mb X_0$, even when each column of $\mb X_0$ contains $O(n)$ nonzeros. 

\subsection{Prior Arts and Connections}
In Section II.E of the companion paper~\cite{sun2015complete_a}, we provide detailed comparisons of our results with prior theoretical results on DR; we conclude that this is the first algorithmic framework that guarantees efficient recovery of complete dictionaries when the coefficients have up to constant fraction of nonzeros. We also draw methodological connections to work on understanding nonconvex heuristics, and other nonconvex problems with similar geometric structures. Here we focus on drawing detailed connections to the optimization literature. 

Trust-region method (TRM) has a rich history dating back to 40's; see the monograph~\cite{conn2000trust} for accounts of the history and developments. \js{The main motivation for early developments was to address limitations of the classic Newton's method (see, e.g., Section 3 of~\cite{sorensen1982newton}). The limitations include the technical subtleties to establish local and global convergence results.  Moreover, when the Hessian is singular or indefinite, the movement direction is either not well-defined, or does not improve the objective function.} \cite{fletcher1971modified, gander1978linear, goldfarb1980curvilinear, mccormick1977modification, mukai1978second, more1979use} initialized the line of work that addresses the limitations. Particularly, \cite{sorensen1982newton, more1983computing} proposed using local second-order Taylor approximation as model function in the trust-region framework for unconstrained optimization.  They showed that under mild conditions, the trust-region iterate sequence has a limit point that is critical and has positive semidefinite Hessian; see also Section 6.5-6.6 of~\cite{conn2000trust}. Upon inspecting the relevant proofs, it seems not hard to strengthen the results to sequence convergence to local minimizers, under a ridable saddle condition as ours, for unconstrained optimization.  

Research activities to port theories and algorithms of optimization in Euclidean space to Riemannian manifolds are best summarized by three monographs: \cite{helmke1994optimization, udriste1994convex, absil2009}. \cite{edelman1998geometry} developed Newton and conjugate-gradient methods for the Stiefel manifolds, of which the sphere is a special case; \cite{absil2009} presents a complete set of first- and second-order Riemannian algorithms and convergence analyses; see also the excellent associated optimization software toolbox~\cite{boumal2014manopt}. Among these, trust-region method was first ported to the Riemannian setting in~\cite{absil2007trust}, with emphasis on efficient implementation which only approximately solves the trust-region subproblem according to the Cauchy point scheme. The Cauchy point definition adopted there was the usual form based on the gradient, not strong enough to ensure the algorithm escape from ridable saddle points even if the true Hessian is in use in local approximation. In comparison, in this work we assume that the trust-region subproblem is \emph{exactly} solved, such that ridable saddles (the only possible saddles for our problem) are properly skipped. By this, we obtain the strong guarantee that the iterate sequence converges to a local minimizer, in contrast to the weak global convergence (gradient sequence converging to zero) or local convergence (sequence converging to a local minimizer within a small radius) established in~\cite{absil2007trust}. To the best of our knowledge, our convergence result is first of its kind for a specific problem on sphere. \js{After our initial submission, \cite{boumal2016global} has recently established worst-case iteration complexity of Riemannian TRM to converge to second-order critical points (i.e., critical points with positive semidefinite Hessians), echoing the results in the Euclidean case~\cite{cartis2012complexity}. Their results are under mild Lipschitz-type assumptions and allow inexact subproblem solvers, and hence are very practical and general. However, on our particular problem, their result is considerably pessimistic, compared to our convergence result obtained from a specialized analysis. } 

Solving the trust-region subproblem exactly is expensive. Practically, often a reasonable approximate solution with controlled quality is adequate to guarantee convergence. In this regard, the truncated conjugate gradient (tCG) solver with a good initial search direction is commonly employed in practice (see, e.g., Section 7.5 in~\cite{conn2000trust}). To ensure ridable saddle points be properly escaped from, the eigenpoint idea (see, e.g., Section 6.6 of~\cite{conn2000trust}) is particularly relevant; see also Algorithm 3 and Lemma 10 in~\cite{boumal2016global}. 

\js{The benign function landscape we characterized in the first paper allows any reasonable iterative method that is capable of escaping from ridable saddles to find a local minimizer, with possibly different performance guarantees. The trust-region method we focus on here, and the curviliear search method~\cite{goldfarb1980curvilinear} are second-order methods that guarantee global optimization from arbitrary initializations. Typical first-order methods such as the vanilla gradient descent can only guarantee convergence to a critical point. Nonetheless, for our particular function, noisy/stochastic gradient method guarantees to find a local minimizer from an arbitrary initialization with high probability~\cite{ge2015escaping}. 
} 

\subsection{Notations, and Reproducible Research}
We use bold capital and small letters such as $\mb X$ and $\mb x$ to denote matrices and vectors, respectively. Small letters are reserved for scalars. Several specific mathematical objects we will frequently work with: $O_k$ for the orthogonal group of order $k$, $\bb S^{n-1}$ for the unit sphere in $\R^n$, $\bb B^n$ for the unit ball in $\R^n$, and $[m] \doteq \set{1, \dots, m}$ for positive integers $m$. We use $\paren{\cdot}^*$ for matrix transposition, causing no confusion as we will work entirely on the real field. We use superscript to index rows of a matrix, such as $\mb x^i$ for the $i$-th row of the matrix $\mb X$,  and subscript to index columns, such as $\mb x_j$. All vectors are defaulted to column vectors. So the $i$-th row of $\mb X$ as a row vector will be written as $\paren{\mb x^i}^*$. For norms, $\norm{\cdot}{}$ is the usual $\ell^2$ norm for a vector and the operator norm (i.e., $\ell^2 \to \ell^2$) for a matrix; all other norms will be indexed by subscript, for example the Frobenius norm $\norm{\cdot}{F}$ for matrices and the element-wise max-norm $\norm{\cdot}{\infty}$. We use $\mb x \sim \mc L$ to mean that the random variable $\mb x$ is distributed according to the law $\mc L$. Let $\mc N$ denote the Gaussian law. Then $\mb x \sim \mc N\paren{\mb 0, \mb I}$ means that $\mb x$ is a standard Gaussian vector. Similarly, we use $\mb x \sim_{i.i.d.} \mc L$ to mean elements of $\mb x$ are independently and identically distributed according to the law $\mc L$. So the fact $\mb x \sim \mc N\paren{\mb 0, \mb I}$ is equivalent to that $\mb x \sim_{i.i.d.} \mc N\paren{0, 1}$. One particular distribution of interest for this paper is the Bernoulli-Gaussian with rate $\theta$: $Z \sim B \cdot G$, with $G \sim \mc N\paren{0, 1}$ and $B \sim \mathrm{Ber}\paren{\theta}$. We also write this compactly as $Z \sim \mathrm{BG}\paren{\theta}$. We frequently use indexed $C$ and $c$ for numerical constants when stating and proving technical results. The scopes of such constants are local unless otherwise noted. We use standard notations for most other cases, with exceptions clarified locally. 

The codes to reproduce all the figures and experimental results are available online: 
\begin{quote}
\centering
\url{https://github.com/sunju/dl_focm} . 
\end{quote}

%% file: sec/algorithm.tex
\section{Finding One Local Minimizer via the Riemannian Trust-Region Method} \label{sec:algorithm}
We are interested to seek a local minimizer of~\eqref{eq:main_obj}. The presence of saddle points have motivated us to develop a second-order Riemannian trust-region algorithm over the sphere; the existence of descent directions at nonoptimal points drives the trust-region iteration sequence towards one of the minimizers asymptotically. We will prove that under our modeling assumptions, this algorithm with an arbitrary initialization efficiently produces an accurate approximation\footnote{By ``accurate'' we mean one can achieve an arbitrary numerical accuracy $\eps > 0$ with a reasonable amount of time. Here the running time of the algorithm is on the order of $\log \log ( 1/\eps )$ in the target accuracy $\eps$, and polynomial in other problem parameters. } to one of the minimizers. Throughout the exposition, basic knowledge of Riemannian geometry is assumed. We will try to keep the technical requirement minimal possible; the reader can consult the excellent monograph~\cite{absil2009} for relevant background and details. 

\subsection{Some Basic Facts about the Sphere and $f$} \label{sec:review_facts_sphere}
For any point $\mb q \in \bb S^{n-1}$, the tangent space $T_{\mb q} \bb S^{n -1}$ and the orthoprojector $\mc P_{T_{\mb q} \bb S^{n -1}}$ onto $T_{\mb q} \bb S^{n -1}$ are given by
\begin{align*}
	T_{\mb q} \bb S^{n -1 } &= \set{ \mb \delta\in \bb R^n:  \mb q^* \mb \delta = 0 }, \\
	\mc P_{T_{\mb q} \bb S^{n -1}} &= \mb I - \mb q \mb q^*  = \mb U \mb U^*,
\end{align*}
where $\mb U \in \bb R^{n\times (n-1)}$ is an arbitrary orthonormal basis for $T_{\mb q}\bb S^{n-1}$ (note that the orthoprojector is independent of the basis $\mb U$ we choose). \js{Consider any $\mb \delta \in T_{\mb q} \bb S^{n-1}$. The map 
\begin{align*}
\gamma(t): t \mapsto \mb q \cos\paren{t \norm{\mb \delta}{}} + \frac{\mb \delta}{\norm{\mb \delta}{}} \sin \paren{t \norm{\mb \delta}{}}
\end{align*}
defines a smooth curve on the sphere that satisfies $\gamma(0) = \mb q$ and $\dot{\gamma}(0) = \mb \delta$. Geometrically, $\gamma(t)$ is a segment of the great circle that passes $\mb q$ and has $\mb \delta$ as its tangent vector at $\mb q$. The exponential map for $\mb \delta$ is defined as 
\begin{align*}
	\exp_{\mb q}(\mb \delta) \doteq \gamma(1) = \mb q \cos\norm{\mb \delta}{} + \frac{\mb \delta}{\norm{\mb \delta}{}} \sin\norm{\mb \delta}{}.
\end{align*}
It is a canonical way of pulling $\mb \delta$ to the sphere. 
\begin{figure}[!htbp]
\centering
\includegraphics[width=0.3\linewidth]{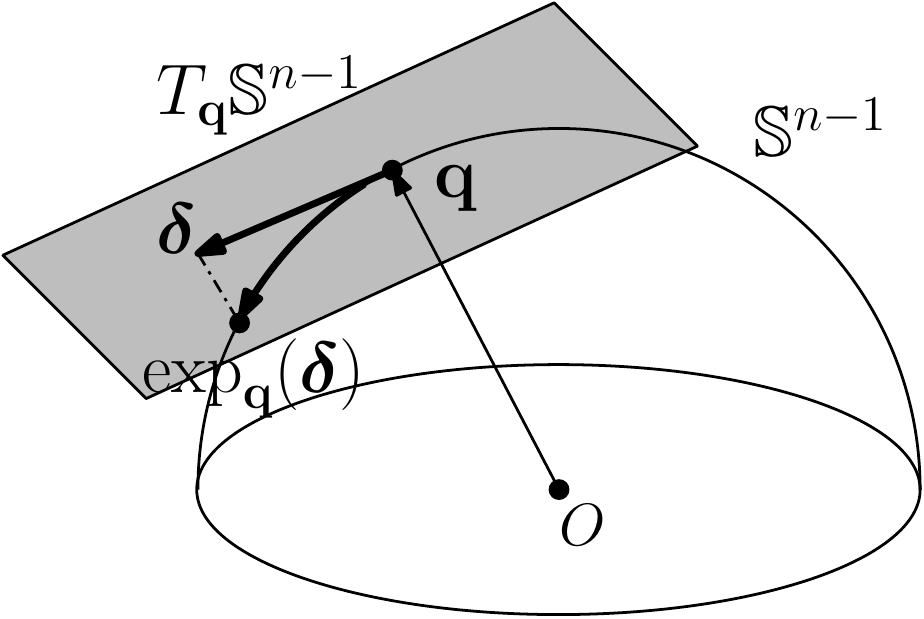}
\caption{Illustrations of the tangent space $T_{\mb q}\bb S^{n-1}$ and exponential map $\exp_{\mb q}\paren{\mb \delta}$ defined on the sphere $\bb S^{n-1}$.} \label{fig:exp-map}
\end{figure}

In this paper we are interested in the restriction of $f$ to the unit sphere $\bb S^{n-1}$. For the sake of performing optimization, we need local approximations of $f$. Instead of directly approximating the function in $\bb R^n$, we form quadratic approximations of $f$ in the tangent spaces of $\bb S^{n-1}$. We consider the smooth function $f \circ \exp_{\mb q}(\mb \delta): T_{\mb q} \bb S^{n-1} \mapsto \R$, where $\circ$ is the usual function composition operator. An applications of vector space Taylor's theorem gives 
\begin{align*}
f \circ \exp_{\mb q}(\mb \delta) \approx f(\mb q; \wh{\mb Y}) + \innerprod{ \nabla f(\mb q; \widehat{\mb Y}) }{\mb \delta } + \frac{1}{2} \mb \delta^* \left( \nabla^2 f( \mb q; \widehat{\mb Y}) - \innerprod{ \nabla f(\mb q; \widehat{\mb Y}) }{ \mb q } \mb I \right) \mb \delta 
\end{align*}
when $\norm{\mb \delta}{}$ is small. Thus, we form a quadratic approximation $\wh{f}(\mb \delta; \mb q): T_{\mb q} \bb S^{n-1} \mapsto \R$ as 
\begin{align} \label{eqn:f-appx}
\widehat{f}(\mb \delta; \mb q, \widehat{\mb Y}) \;\doteq\; f(\mb q; \wh{\mb Y}) + \innerprod{ \underbracket{\nabla f(\mb q; \widehat{\mb Y})}_{} }{\mb \delta } + \frac{1}{2} \mb \delta^* \left( \underbracket{\nabla^2 f( \mb q; \widehat{\mb Y}) - \innerprod{ \nabla f(\mb q; \widehat{\mb Y}) }{ \mb q } \mb I}_{} \right) \mb \delta. 
\end{align}
}
Here $\nabla f(\mb q)$ and $\nabla^2 f(\mb q)$ denote the usual (Euclidean) gradient and Hessian of $f$ w.r.t. $\mb q$ in $\R^n$. For our specific $f$ defined in~\eqref{eq:main_obj}, it is easy to check that
\begin{align}
\nabla f(\mb q; \widehat{\mb Y}) & = \frac{1}{p}\sum_{k=1}^p \tanh\paren{\frac{\mb q^*\widehat{\mb y}_k}{\mu}} \widehat{\mb y}_k, \label{eq:fq_grad}\\
\nabla^2 f(\mb q; \widehat{\mb Y}) & = \frac{1}{p} \sum_{k=1}^p \frac{1}{\mu}\brac{1-\tanh^2\paren{\frac{\mb q^* \widehat{\mb y}_k}{\mu}}} \widehat{\mb y}_k \widehat{\mb y}^*_k. \label{eq:fq_hess}
\end{align}
The quadratic approximation also naturally gives rise to the Riemannian gradient and Riemannian Hessian defined on $T_{\mb q} \bb S^{n-1}$ as 
\begin{align} 
\grad f (\mb q; \wh{\mb Y}) & = \mc P_{T_{\mb q} \bb S^{n-1}} \nabla f(\mb q; \wh{\mb Y}), \label{eq:fq_rie_grad} \\
\Hess f (\mb q; \wh{\mb Y}) &= \mc P_{T_{\mb q} \bb S^{n-1}} \left( \nabla^2 f(\mb q; \wh{\mb Y}) - \innerprod{\nabla f(\mb q; \wh{\mb Y}) }{\mb q} \mb I \right) \mc P_{T_{\mb q} \bb S^{n-1}}. \label{eq:fq_rie_hess}
 \end{align}
Thus, the above quadratic approximation can be rewritten compactly as 
\begin{align*}
	\widehat{f}\paren{\mb \delta; \mb q, \widehat{\mb Y}} = f(\mb q; \wh{\mb Y}) + \innerprod{\mb \delta}{\grad f(\mb q; \wh{\mb Y})}+ \frac{1}{2} \mb \delta^* \Hess f(\mb q; \wh{\mb Y})\mb \delta,\qquad \forall~\mb \delta \in T_{\mb q} \bb S^{n-1}.
\end{align*}
The first order necessary condition for {\em unconstrained} minimization of function $\widehat{f}$ over $T_{\mb q} \bb S^{n-1}$ is
\begin{align}\label{eqn:ts-optimal-solution-1}
	\grad f(\mb q; \wh{\mb Y}) + \Hess f(\mb q; \wh{\mb Y}) \mb \delta_\star = \mb 0. 
\end{align}
If $\Hess f(\mb q; \wh{\mb Y})$ is positive semidefinite and has ``full rank'' $n-1$ (hence ``nondegenerate"\footnote{Note that the $n \times n$ matrix $\Hess f(\mb q; \wh{\mb Y})$ has rank at most $n-1$, as the nonzero $\mb q$ obviously is in its null space. When  $\Hess f(\mb q; \wh{\mb Y})$ has rank $n-1$, it has no null direction in the tangent space. Thus, in this case it acts on the tangent space like a full-rank matrix. }),  
the unique solution $\mb \delta_\star$ is
\begin{align*}
	\mb \delta_\star = -\mb U \paren{\mb U^*\brac{\Hess f(\mb q; \wh{\mb Y})} \mb U }^{-1} \mb U^* \grad f(\mb q), 
\end{align*}
which is also invariant to the choice of basis $\mb U$. Given a tangent vector $\mb \delta \in T_{\mb q}\bb S^{n-1}$, let $\gamma(t) \doteq \exp_{\mb q}(t \mb \delta)$ denote a geodesic curve on $\bb S^{n-1}$. Following the notation of \cite{absil2009}, let 
 \begin{align*}
 	\mc P_{\gamma}^{\tau \leftarrow 0} : T_{\mb q} \bb S^{n-1} \to T_{\gamma(\tau)} \bb S^{n-1}
 \end{align*}
denotes the parallel translation operator, which translates the tangent vector $\mb \delta$ at $\mb q = \gamma(0)$ to a tangent vector at $\gamma(\tau)$, in a ``parallel'' manner. In the sequel, we identify $\mc P_{\gamma}^{\tau \leftarrow 0}$ with the following $n \times n$ matrix, whose restriction to $T_{\mb q} \bb S^{n-1}$ is the parallel translation operator (the detailed derivation can be found in Chapter 8.1 of \cite{absil2009}):
 \begin{eqnarray} \label{eq:alg_tsp_op}
 \mc P_{\gamma}^{\tau \leftarrow 0} &=& \paren{ \mb I - \frac{\mb \delta \mb \delta^*}{\norm{\mb \delta}{}^2} } - \mb q \sin\left( \tau \norm{\mb \delta}{} \right) \frac{\mb \delta^*}{\norm{\mb \delta}{}} + \frac{\mb \delta}{\norm{\mb \delta}{}} \cos\left( \tau \norm{\mb \delta }{} \right) \frac{\mb \delta^* }{\norm{\mb \delta}{}} \nonumber \\
 &=& \mb I + \left( \cos( \tau \norm{\mb \delta}{} ) - 1 \right) \frac{\mb \delta \mb \delta^*}{\norm{\mb \delta }{}^2} -  \sin\left( \tau \norm{\mb \delta }{} \right) \frac{\mb q \mb \delta^*}{\norm{\mb \delta}{}}. 
 \end{eqnarray}
Similarly, following the notation of \cite{absil2009}, we denote the inverse of this matrix by $\mc P_{\gamma}^{0 \leftarrow \tau}$, where its restriction to $T_{\gamma(\tau)} \bb S^{n-1}$ is the inverse of the parallel translation operator $\mc P_{\gamma}^{\tau \leftarrow 0}$.

\subsection{The Geometric Results from~\cite{sun2015complete_a}} \label{sec:geo_reproduce}
\begin{figure}[t]
\centering
\includegraphics[width=0.3\textwidth]{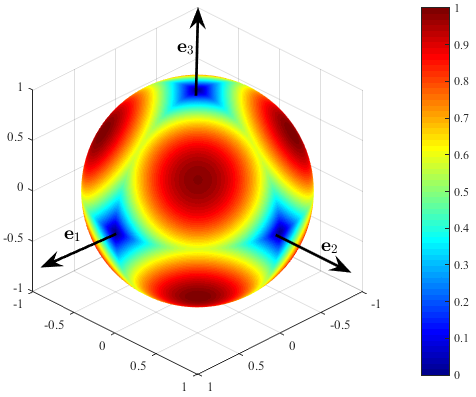}
\hspace{0.1in}
\includegraphics[width=0.3\textwidth]{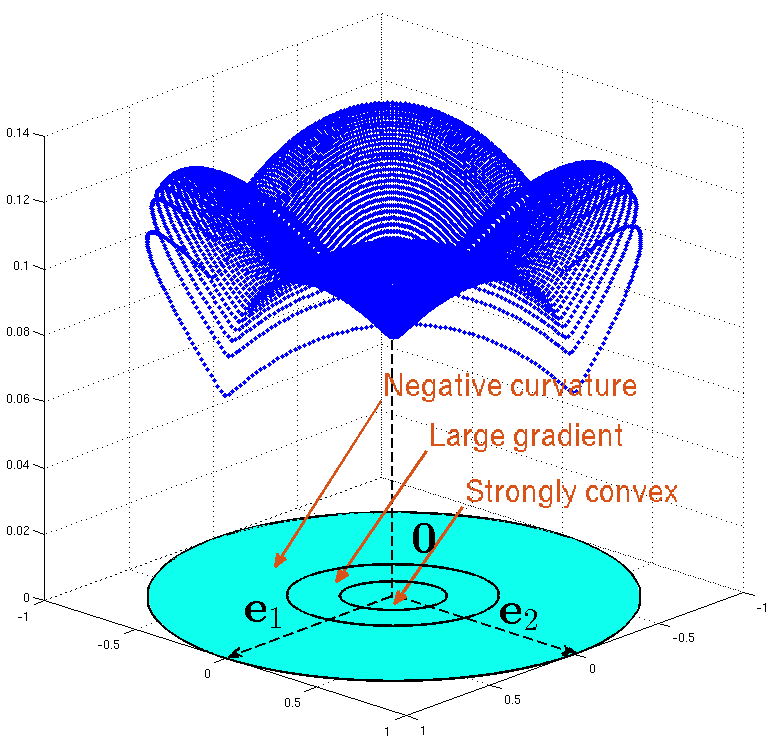} 
\caption{\textbf{Why is dictionary learning over $\bb S^{n-1}$ tractable?} \js{Assume the target dictionary $\mb A_0 = \mb I$. \textbf{Left:} Large sample objective function $\bb E_{\mb X_0}\brac{f\paren{\mb q}}$. The only local minimizers are the standard basis vectors $\mb e_i$'s and their negatives. \textbf{Right:} A visualization of the function as a height above the equatorial section $\mb e_3^\perp$, i.e., $\mathrm{span}\{\mb e_1, \mb e_2\} \cap \bb B^3$. The derived function is obtained by assigning values of points on the upper hemisphere to their corresponding projections on the equatorial section $\mb e_3^\perp$. The minimizers for the derived function are $\mb 0, \pm \mb e_1, \pm \mb e_2$. Around $\mb 0$ in $\mb e_3^\perp$, the function exhibits a small region of strong convexity, a region of large gradient, and finally a region in which the direction away from $\mb 0$ is a direction of negative curvature.}}
\label{fig:large-sample-sphere}
\end{figure}
We reproduce the main geometric theorems from~\cite{sun2015complete_a} here for the sake of completeness. To characterize the function landscape of $f\paren{\mb q; \mb X_0}$ over $\bb S^{n-1}$, we mostly work with the function 
\begin{align}\label{eqn:function-g}
g\paren{\mb w; \mb X_0} \doteq f\paren{\mb q\paren{\mb w}; \mb X_0} = \frac{1}{p} \sum_{k=1}^p h_{\mu}\paren{\mb q\paren{\mb w}^* \paren{\mb x_0}_k}, 
\end{align}
induced by the reparametrization
\begin{align}
\mb q\paren{\mb w} = \paren{\mb w, \sqrt{1-\norm{\mb w}{}^2} },  \quad \mb w \in \bb B^{n-1}. 
\end{align}
Geometrically, this corresponds to projection of the function $f$ above the equatorial section $\mb e_n^\perp$ onto $\mb e_n^\perp$ (see Fig.~\ref{fig:large-sample-sphere} (right) for illustration). In particular, we focus our attention to the smaller set of the ball: 
\begin{align}
\Gamma = \set{\mb w: \norm{\mb w}{} < \sqrt{\frac{4n-1}{4n}}} \subsetneq \bb B^{n-1},  
\end{align}
because $\mb q\paren{\Gamma}$ contains all points $\mb q \in \bb S^{n-1}$ with $n \in \mathop{\arg\max}_{i \in \pm [n]} \mb q^* \mb e_i$. We can similarly characterize other parts of $f$ on $\bb S^{n-1}$ using projection onto other equatorial sections.

\begin{theorem}[High-dimensional landscape - orthogonal dictionary]\label{thm:geometry_orth}
Suppose $\mb A_0 = \mb I$ and hence $\mb Y = \mb A_0 \mb X_0 = \mb X_0$. There exist positive constants $c_\star$ and $C$, such that for any $\theta \in (0,1/2)$ and $\mu < c_a\min\Brac{\theta n^{-1}, n^{-5/4}}$, whenever 
\begin{align}
p \ge \frac{C}{\mu^2 \theta^2} n^3 \log \frac{n}{\mu \theta},
\end{align}
the following hold simultaneously with probability at least $1 - c_b p^{-6}$: 
\begin{align}
\nabla^2 g(\mb w; \mb X_0) &\succeq \frac{c_\star \theta}{\mu} \mb I \quad &\forall \, \mb w \quad \text{s.t.}& \quad \norm{\mb w}{} \le \frac{\mu}{4\sqrt{2}},  \label{eqn:hess-zero-uni-orth} \\
\frac{\mb w^* \nabla g(\mb w; \mb X_0)}{\norm{\mb w}{}} &\ge c_\star \theta \quad &\forall \, \mb w \quad \text{s.t.}& \quad \frac{\mu}{4\sqrt{2}} \le \norm{\mb w}{} \le \frac{1}{20 \sqrt{5}} \label{eqn:grad-uni-orth} \\
\frac{\mb w^* \nabla^2 g(\mb w; \mb X_0) \mb w}{\norm{\mb w}{}^2} &\le - c_\star \theta \quad &\forall \, \mb w \quad \text{s.t.}& \quad \frac{1}{20 \sqrt{5}} \le \norm{\mb w}{} \le \sqrt{\frac{4n-1}{4n}},   \label{eqn:curvature-uni-orth}
\end{align}
{\em and} the function $g(\mb w; \mb X_0)$ has exactly one local minimizer $\mb w_\star$ over the open set $\Gamma \doteq \set{ \mb w: \norm{\mb w}{} < \sqrt{ \tfrac{4n-1}{4n} } }$, which satisfies
\begin{equation}
\norm{\mb w_\star - \mb 0 }{} \;\le\; \min\Brac{\frac{c_c\mu}{\theta} \sqrt{\frac{n \log p}{p}}, \frac{\mu}{16}}.    
\end{equation}
Here $c_a$ through $c_c$ are all positive constants. 
\end{theorem}
Recall that the reason we just need to characterize the geometry for the case $\mb A_0 = \mb I$ is that for other orthogonal $\mb A_0$, the function landscape is simply a rotated version of that of $\mb A_0 = \mb I$. 
\begin{theorem}[High-dimensional landscape - complete dictionary]\label{thm:geometry_comp}
Suppose $\mb A_0$ is complete with its condition number $\kappa\paren{\mb A_0}$. There exist positive constants $c_\star$ (particularly, the same constant as in Theorem~\ref{thm:geometry_orth}) and $C$, such that for any $\theta \in (0,1/2)$ and $\mu < c_a\min\Brac{\theta n^{-1}, n^{-5/4}}$, when 
\begin{align}
p \ge \frac{C}{c_\star^2 \theta^2} \max\set{\frac{n^4}{\mu^4}, \frac{n^5}{\mu^2}} \kappa^8\paren{\mb A_0} \log^4\paren{\frac{\kappa\paren{\mb A_0} n}{\mu \theta}}
\end{align} 
and $\overline{\mb Y} \doteq \sqrt{p\theta}\paren{\mb Y \mb Y^*}^{-1/2} \mb Y$, $\mb U \mb \Sigma \mb V^* = \mathtt{SVD}\paren{\mb A_0}$, the following hold simultaneously with probability at least $1 - c_b p^{-6}$: 
\begin{align}
\nabla^2 g(\mb w; \mb V \mb U^* \overline{\mb Y}) &\succeq \frac{c_\star \theta}{2\mu} \mb I \quad &\forall \, \mb w \quad \text{s.t.}& \quad \norm{\mb w}{} \le \frac{\mu}{4\sqrt{2}},  \label{eqn:hess-zero-uni-comp} \\
\frac{\mb w^* \nabla g(\mb w; \mb V \mb U^* \overline{\mb Y})}{\norm{\mb w}{}} &\ge \frac{1}{2}c_\star \theta \quad &\forall \, \mb w \quad \text{s.t.}& \quad \frac{\mu}{4\sqrt{2}} \le \norm{\mb w}{} \le \frac{1}{20 \sqrt{5}} \label{eqn:grad-uni-comp} \\
\frac{\mb w^* \nabla^2 g(\mb w; \mb V \mb U^* \overline{\mb Y}) \mb w}{\norm{\mb w}{}^2} &\le -\frac{1}{2} c_\star \theta \quad &\forall \, \mb w \quad \text{s.t.}& \quad \frac{1}{20 \sqrt{5}} \le \norm{\mb w}{} \le \sqrt{\frac{4n-1}{4n}},   \label{eqn:curvature-uni-comp}
\end{align}
{\em and} the function $g(\mb w; \mb V \mb U^* \overline{\mb Y})$ has exactly one local minimizer $\mb w_\star$ over the open set $\Gamma \doteq \set{ \mb w: \norm{\mb w}{} < \sqrt{ \tfrac{4n-1}{4n} } }$, which satisfies
\begin{equation}
\norm{\mb w_\star - \mb 0 }{} \;\le\; \mu /7.    
\end{equation}
Here $c_a, c_b$ are both positive constants.
\end{theorem}
From the above theorems, it is clear that for any saddle point in the $\mb w$ space, the Hessian of $g$ has at least one negative eigenvalue. Now the problem is whether all saddle points of $f$ on $\bb S^{n-1}$ are ``ridable'', because as alluded to in previous discussion, we need to perform actual optimization in the $\mb q$ space. Instead of presenting a rigorous technical statement and detailed proof, we include here just an informal argument; our actual proof for algorithmic convergence runs back and forth in $\mb w$ and $\mb q$ space and such lack will not affect our arguments there. 

It is very easy to verify the following fact (see proof of Lemma~\ref{lem:alg_gradient_lb} on page~\pageref{proof:lem_alg_gradient_lb}):
\begin{align*}
\innerprod{\grad f(\mb q)}{\mb q - \mb e_n/q_n} = \innerprod{\mb w}{\nabla g(\mb w)}. 
\end{align*} 
Thus, $\innerprod{\grad f(\mb q)}{\mb q - \mb e_n/q_n} \neq 0$ if and only if $\innerprod{\mb w}{\nabla g(\mb w)} \ne 0$, implying that $\grad f(\mb q)$ will never be zero in the spherical region corresponding to $\set{\mb w: \mu/(4\sqrt{2}) \le \|\mb w\| \le 1/(20\sqrt{5})}$. Moreover, it is shown in Lemma~\ref{lem:alg_strcvx_lb} below that the Riemannian Hessian is positive definite for the spherical region corresponding to $\set{\mb w: \|\mb w\| \le \mu/(4\sqrt{2})}$, so there is no saddle point in this region either. Over $\mb q(\Gamma) \cap \bb S^{n-1}$, potential saddle points lie only in the region corresponding to $\{\mb w: 1/(20\sqrt{5}) \le \|\mb w\| \le \sqrt{(4n-1)/(4n)} \}$. Theorem~\ref{thm:geometry_orth} and Theorem~\ref{thm:geometry_comp} imply that around each point in this region, a cross section of the function $g(\mb w)$ is \emph{strictly} concave locally. Intuitively, by the $\mb q(\mb w)$ mapping the same happens in the $\mb q$ space, i.e., the Riemannian Hessian has a strictly negative eigenvalue.

\subsection{The Riemannian Trust-Region Algorithm over the Sphere}\label{subsec:TRM-algorithm}
For a function $f$ in the Euclidean space, the typical TRM starts from some initialization $\mb q^{(0)} \in \bb R^{n}$, and produces a sequence of iterates $\mb q^{(1)}, \mb q^{(2)}, \dots$, by repeatedly minimizing a quadratic approximation $\widehat{f}$ to the objective function $f(\mb q)$, over a ball centered around the current iterate. 

For our $f$ defined over $\bb S^{n-1}$, given the previous iterate $\mb q^{(r-1)}$, the TRM produces the next movement by generating a solution $\widehat{\mb \delta}$ to 
\begin{equation} \label{eqn:subproblem-1}
\mini_{\mb \delta \in T_{\mb q^{(r-1)}} \bb S^{n-1}, \; \norm{\mb \delta}{} \le \Delta} \quad \widehat{f}\paren{\mb \delta; \mb q^{(r-1)}},
\end{equation}
where $\wh{f}\paren{\mb \delta; \mb q^{(r-1)}}$ is the local quadratic approximation defined in~\eqref{eqn:f-appx}. The solution $\widehat{\mb \delta}$ is then pulled back to $\bb S^{n-1}$ from $T_{\mb q} \bb S^{n-1}$. If we choose the exponential map to pull back the movement $\widehat{\mb \delta},$\footnote{The exponential map is only one of the many possibilities; also for general manifolds other retraction schemes may be more practical. See exposition on retraction in Chapter 4 of~\cite{absil2009}. } the next iterate then reads 
\begin{align}
\mb q^{(r)} = \mb q^{(r-1)} \cos \|\widehat{\mb \delta}\| + \frac{\widehat{\mb \delta}}{\|\widehat{\mb \delta}\|} \sin \|\widehat{\mb \delta}\|. 
\end{align}

To solve the subproblem \eqref{eqn:subproblem-1} numerically, we can take any matrix $\mb U \in \bb R^{n \times (n-1)}$ whose columns form an orthonormal basis for $T_{\mb q^{(r-1)}}\bb S^{n-1}$, and produce a solution $\widehat{\mb \xi}$ to
\begin{equation}
\mini_{\norm{\mb \xi}{} \le \Delta} \quad \widehat{f}\paren{\mb U \mb \xi; \mb q^{(r-1)}}.   \label{eqn:trsp} 
\end{equation}
Solution to~\eqref{eqn:subproblem-1} can then be recovered as $\widehat{\mb \delta} = \mb U\widehat{\mb \xi}$. 

The problem \eqref{eqn:trsp} is an instance of the classic {\em trust region subproblem}, i.e., minimizing a quadratic function subject to a single quadratic constraint. \js{Albeit potentially nonconvex, this notable subproblem can be solved in polynomial time by several numerical methods~\cite{more1983computing, conn2000trust, rendl1997semidefinite, ye2003new, fortin2004trust, hazan2014linear}. Approximate solution of the subproblem suffices to guarantee convergence in theory, and lessens the storage and computational burden in practice. We will deploy the approximate version in simulations. For simplicity, however, our subsequent analysis assumes the subproblem is solved \emph{exactly}. We next briefly describe how one can deploy the semidefinite programming (SDP) approach~\cite{rendl1997semidefinite, ye2003new, fortin2004trust, hazan2014linear} to solve the subproblem exactly. This choice is due to the well-known effectiveness and robustness of the SDP approach on this problem. } We introduce
\begin{align}
\tilde{\mb \xi} = \brac{\mb \xi^*,1}^*,~\mb \Theta = \tilde{\mb \xi} \tilde{\mb \xi} ^*,~\mb M = \brac{\begin{array}{ll}
\mb A &\mb b\\
\mb b^* &0
\end{array}},
\end{align}
where $\mb A = \mb U^* \Hess f(\mb q^{(r-1)}; \wh{\mb Y})\mb U$ and $\mb b = \mb U^* \grad \nabla f(\mb q^{(r-1)}; \widehat{\mb Y})$. The resulting SDP to solve is 
\begin{align}
\mini_{\; \mb \Theta} \innerprod{\mb M}{\mb \Theta},~\st~\trace({\mb \Theta})\le \Delta^2+ 1,~\innerprod{\mb E_{n+1}}{\mb \Theta}=1,~\mb \Theta \succeq \mb 0, \label{eqn:SDP_relaxation}
\end{align}
where $\mb E_{n+1} = \mb e_{n+1} \mb e_{n+1}^*$. Once the problem \eqref{eqn:SDP_relaxation} is solved to its optimum $\mb \Theta_\star$, one can provably recover the minimizer $\mb \xi_\star$ of \eqref{eqn:trsp} by computing the SVD of $\mb \Theta_\star= \widetilde{\mb U}\mb \Sigma \widetilde{\mb V}^*$, and extract as a subvector the first $n-1$ coordinates of the principal eigenvector $\widetilde{\mb u}_1$ (see Appendix B of \cite{boyd2004optimization}).

\subsection{Main Convergence Results}
Using general convergence results on Riemannian TRM (see, e.g., Chapter 7 of \cite{absil2009}), it is not difficult to prove that the gradient sequence $\grad f(q^{(r)}; \wh{\mb Y})$ produced by TRM converges to zero (i.e., global convergence), or the sequence converges (at quadratic rate) to a local minimizer if the initialization is already close a local minimizer (i.e., local convergence). In this section, we show that under our probabilistic assumptions, these results can be substantially strengthened. In particular, the algorithm is guaranteed to produce an accurate approximation to a local minimizer of the objective function, in a number of iterations that is polynomial in the problem size, from arbitrary initializations. The arguments in the companion paper~\cite{sun2015complete_a} showed that w.h.p. every local minimizer of $f$ produces a close approximation to a row of $\mb X_0$. Taken together, this implies that the algorithm efficiently produces a close approximation to one row of $\mb X_0$.

Thorough the analysis, we assume the trust-region subproblem is exactly solved and the step size parameter $\Delta$ is fixed. Our next two theorems summarize the convergence results for orthogonal and complete dictionaries, respectively.
 
\begin{theorem}[TRM convergence - orthogonal dictionary] \label{thm:trm_orth}
Suppose the dictionary $\mb A_0$ is orthogonal. There exists a positive constant $C$, such that for all $\theta \in \paren{0, 1/2}$ and $\mu < c_a \min \set{\theta n^{-1}, n^{-5/4}}$, whenever 
\begin{align*}
p \ge \frac{C}{\mu^2 \theta^2} n^3 \log \frac{n}{\mu \theta},
\end{align*} 
with probability at least
$
1 - c_b p^{-6}, 
$
the Riemannian trust-region algorithm with input data matrix $\widehat{\mb Y} = \mb Y$, any initialization $\mb q^{(0)}$ on the sphere, and a step size satisfying 
\begin{align}
\Delta \le \frac{c_c c_\star^3 \theta^3 \mu^2}{n^{7/2}\log^{7/2}\paren{np}}
\end{align}
returns a solution $\widehat{\mb q}\in \bb S^{n-1} $ which is $\eps$ near to one of the local minimizers $\mb q_\star$ (i.e., $\norm{\widehat{\mb q}-\mb q_\star}{}\leq \eps$) in at most 
\begin{align} \label{eq:cplx_orth}
\max\set{\frac{c_d n^6 \log^3\paren{np}}{c_{\star}^3 \theta^3 \mu^4 }, \frac{c_e n }{c_\star^2 \theta^2 \Delta^2}} f(\mb q^{(0)}) + \log\log \frac{c_f c_\star \theta \mu}{\eps n^{3/2} \log^{3/2}\paren{np}}
\end{align}
iterations. Here $c_\star$ is as defined in Theorem~\ref{thm:geometry_orth}, and $c_a$ through $c_f$ are all positive constants. 
\end{theorem}

\begin{theorem}[TRM convergence - complete dictionary] \label{thm:trm_comp}
Suppose the dictionary $\mb A_0$ is complete with condition number $\kappa\paren{\mb A_0}$. There exists a positive constant $C$, such that for all $\theta \in \paren{0, 1/2}$, and $\mu < c_a\min \set{\theta n^{-1}, n^{-5/4}}$, whenever 
\begin{align*}
p \ge \frac{C}{c_\star^2 \theta^2} \max\set{\frac{n^4}{\mu^4}, \frac{n^5}{\mu^2}} \kappa^8\paren{\mb A_0} \log^4\paren{\frac{\kappa\paren{\mb A_0} n}{\mu \theta}},
\end{align*} 
with probability at least
$
1 - c_b p^{-6}, 
$
the Riemannian trust-region algorithm with input data matrix $\overline{\mb Y} \doteq \sqrt{p\theta}\paren{\mb Y \mb Y^*}^{-1/2} \mb Y$ where $\mb U \mb \Sigma \mb V^* = \mathtt{SVD}\paren{\mb A_0}$, any initialization $\mb q^{(0)}$ on the sphere and a step size satisfying 
\begin{align}
\Delta \le \frac{c_c c_\star^3 \theta^3 \mu^2}{n^{7/2}\log^{7/2}\paren{np}}
\end{align}
returns a solution $\widehat{\mb q}\in \bb S^{n-1} $ which is $\eps$ near to one of the local minimizers $\mb q_\star$ (i.e., $\norm{\widehat{\mb q}-\mb q_\star}{}\leq \eps$) in at most 
\begin{align} \label{eq:cplx_comp}
\max\set{\frac{c_d n^6 \log^3\paren{np}}{c_{\star}^3 \theta^3 \mu^4 }, \frac{c_e n }{c_\star^2 \theta^2 \Delta^2}} f(\mb q^{(0)}) + \log\log \frac{c_f c_\star \theta \mu}{\eps n^{3/2} \log^{3/2}\paren{np}}
\end{align}
iterations. Here $c_\star$ is as in Theorem~\ref{thm:geometry_orth}, and $c_a$ through $c_f$ are all positive constants.
\end{theorem}
Our convergence result shows that for any target accuracy $\eps > 0$ the algorithm terminates within polynomially many steps. \js{Specifically, the first summand in~\eqref{eq:cplx_orth} or~\eqref{eq:cplx_comp} is the number of steps the sequence takes to enter the strongly convex region and be ``reasonably" close to a local minimizer. All subsequent trust-region subproblems are then unconstrained (proved below) -- the constraint is inactive at optimal point, and hence the steps behave like Newton steps. The second summand reflects the typical quadratic local convergence of the Newton steps. }

Our estimate of the number of steps is pessimistic: the running time is a relatively high-degree polynomial in $p$ and $n$. We will discuss practical implementation details that help speed up in Section~\ref{sec:exp}. Our goal in stating the above results is not to provide a tight analysis, but to prove that the Riemannian TRM algorithm finds a local minimizer in polynomial time. For nonconvex problems, this is not entirely trivial -- results of~\cite{murty1987some} show that in general it is NP-hard to find a local minimizer of a nonconvex function. 

\subsection{Sketch of Proof for Orthogonal Dictionaries} \label{sec:alg_orth}
The reason that our algorithm is successful derives from the geometry formalized in Theorem~\ref{thm:geometry_orth}. Basically, the sphere $\bb S^{n-1}$ can be divided into three regions. Near each local minimizer, the function is strongly convex, and the algorithm behaves like a standard (Euclidean) TRM algorithm applied to a strongly convex function -- in particular, it exhibits a quadratic asymptotic rate of convergence. Away from local minimizers, the function always exhibits either a strong gradient, or a direction of negative curvature (i.e., the Hessian has a strictly negative  eigenvalue). The Riemannian TRM algorithm is capable of exploiting these quantities to reduce the objective value by at least a fixed amount in each iteration. The total number of iterations spent away from the vicinity of the local minimizers can be bounded by comparing this amount to the initial objective value. Our proofs follow exactly this line and make the various quantities precise.   

Note that for any orthogonal $\mb A_0$, $f\paren{\mb q; \mb A_0 \mb X_0} = f\paren{\mb A_0^* \mb q; \mb X_0}$. In words, this is the established fact that the function landscape of $f(\mb q; \mb A_0 \mb X_0)$ is a rotated version of that of $f(\mb q; \mb X_0)$. Thus, any local minimizer $\mb q_\star$ of $f(\mb q; \mb X_0)$ is rotated to $\mb A_0 \mb q_\star$, a local minimizer of $f(\mb q; \mb A_0 \mb X_0)$. Also if our algorithm generates iteration sequence $\mb q_0, \mb q_1, \mb q_2, \dots$ for $f(\mb q; \mb X_0)$ upon initialization $\mb q_0$, it will generate the iteration sequence $\mb A_0  \mb q_0, \mb A_0 \mb q_1, \mb A_0 \mb q_2, \dots$ for $f\paren{\mb q; \mb A_0 \mb X_0}$. So w.l.o.g. it is adequate that we prove the convergence results for the case $\mb A_0 = \mb I$. So in this section (Section~\ref{sec:alg_orth}), we write $f(\mb q)$ to mean $f(\mb q; \mb X_0)$. 

We partition the sphere into three regions, for which we label as $\rI$, $\rII$, $\rIII$, corresponding to the strongly convex, nonzero gradient, and negative curvature regions, respectively (see Theorem~\ref{thm:geometry_orth}). That is, $\rI$ consists of a union of $2n$ spherical caps
of radius $\mu/(4\sqrt{2})$, each centered around a signed standard basis vector $\pm \mb e_i$. $\rII$ consist of the set difference of a union of $2n$ spherical caps of radius $1/(20\sqrt{5})$, centered around the standard basis vectors $\pm \mb e_i$, and $\rI$. Finally, $\rIII$ covers the rest of the sphere. We say a trust-region step takes an $\rI$ step if the current iterate is in $\rI$; similarly for $\rII$ and $\rIII$ steps. Since we use the geometric structures derived in Theorem~\ref{thm:geometry_orth} and Corollary II.2 in~\cite{sun2015complete_a}, the conditions 
\begin{align} \label{eq:trm_proof_assumed_cond}
\theta \in (0,1/2), \quad \mu < c\min\Brac{\theta n^{-1}, n^{-5/4}}, \quad p \ge \frac{C}{\mu^2 \theta^2} n^3 \log \frac{n}{\mu \theta}
\end{align}
are always in force. 

At step $r$ of the algorithm, suppose $\mb \delta^{(r)}$ is the minimizer of the trust-region subproblem~\eqref{eqn:subproblem-1}. We call the step ``{\em constrained}'' if $\norm{\mb \delta^{(r)}}{} = \Delta$ (the minimizer lies on the boundary and hence the constraint is active), and call it ``{\em unconstrained}'' if $\|\mb \delta^{(r)}\| < \Delta$ (the minimizer lies in the relative interior and hence the constraint is not in force). Thus, in the unconstrained case the optimality condition is~\eqref{eqn:ts-optimal-solution-1}. 

The next lemma provides some estimates about $\nabla f$ and $\nabla^2 f$ that are useful in various contexts. 
\begin{lemma} \label{lem:mag_lip_fq}
We have the following estimates about $\nabla f$ and $\nabla^2 f$: 
\begin{align*}
\sup_{\mb q \in \bb S^{n-1}}\norm{\nabla f\paren{\mb q}}{} & \doteq M_{\nabla} \le \sqrt{n} \norm{\mb X_0}{\infty} , \\
\sup_{\mb q \in \bb S^{n-1}}\norm{\nabla^2 f\paren{\mb q}}{} & \doteq M_{\nabla^2} \le \frac{n}{\mu}\norm{\mb X_0}{\infty}^2, \\
\sup_{\mb q, \mb q' \in \bb S^{n-1}, \mb q \neq \mb q'} \frac{\norm{\nabla f\paren{\mb q} - \nabla f\paren{\mb q'}}{}}{\norm{\mb q - \mb q'}{}} & \doteq L_{\nabla} \le \frac{n}{\mu} \norm{\mb X_0}{\infty}^2, \\
\sup_{\mb q, \mb q' \in \bb S^{n-1}, \mb q \neq \mb q'} \frac{\norm{\nabla^2 f\paren{\mb q} - \nabla^2 f\paren{\mb q'}}{}}{\norm{\mb q - \mb q'}{}} & \doteq L_{\nabla^2} \le \frac{2}{\mu^2} n^{3/2} \norm{\mb X_0}{\infty}^3. 
\end{align*}
\end{lemma}
\begin{proof}
See Page~\pageref{proof:lem_mag_lip_fq} under Section~\ref{sec:proof_algorithm}. 
\end{proof}
Our next lemma says if the trust-region step size $\Delta$ is small enough, one Riemannian trust-region step reduces the objective value by a certain amount when there is any descent direction.  
\begin{lemma} \label{lem:alg_approx_bd2}
Suppose that the trust region size $\Delta \le 1$, and there exists a tangent vector $\mb \delta \in T_{\mb q} \bb S^{n-1}$ with $\norm{\mb \delta}{} \le \Delta$, such that
\begin{equation*}
f( \exp_{\mb q}(\mb \delta) ) \;\le\; f(\mb q) - s
\end{equation*}
for some positive scalar $s\in \reals $. Then the trust region subproblem produces a point $\mb \delta_\star$ with 
\begin{equation*}
f(\exp_{\mb q}(\mb \delta_\star)) \;\le\; f(\mb q) - s + \frac{1}{3}\eta_f \Delta^3, 
\end{equation*}
where $\eta_f \doteq M_{\nabla} + 2M_{\nabla^2} + L_{\nabla} + L_{\nabla^2}$ and $M_{\nabla}$, $M_{\nabla^2}$, $L_{\nabla}$, $L_{\nabla^2}$ are the quantities defined in Lemma~\ref{lem:mag_lip_fq}. 
\end{lemma}
\begin{proof}
See Page~\pageref{proof:lem_alg_approx_bd2} under Section~\ref{sec:proof_algorithm}. 
\end{proof}

To show decrease in objective value for $\rII$ and $\rIII$, now it is enough to exhibit a descent direction for each point in these regions. The next two lemmas help us almost accomplish the goal. For convenience again we choose to state the results for the ``canonical'' section that is in the vicinity of $\mb e_n$ and the projection map $\mb q\paren{\mb w} = [\mb w; (1-\norm{\mb w}{}^2)^{1/2}]$, with the idea that similar statements hold for other symmetric sections. 

\begin{lemma} \label{lem:alg_gradient_func}
Suppose that the trust region size $\Delta \le 1$, 
$\mb w^* \nabla g(\mb w)/\norm{\mb w}{} \ge \beta_g$ 
 for some scalar $\beta_g$, and that $\mb w^* \nabla g(\mb w)/\norm{\mb w}{}$ is $L_g$-Lipschitz on an open ball $\mc B\left(\mb w, \frac{3\Delta}{2\pi\sqrt{n}}\right)$ centered at $\mb w$. Then there exists a tangent vector $\mb \delta \in T_{\mb q} \bb S^{n-1}$ with $\norm{\mb \delta}{} \le \Delta$, such that
\begin{equation*}
f(\exp_{\mb q}(\mb \delta)) \;\le\; f(\mb q) -  \min \set{ \frac{\beta_g^2}{2 L_g}, \frac{3\beta_g \Delta}{4\pi\sqrt{n}} }. 
\end{equation*}
\end{lemma}
\begin{proof}
See Page~\pageref{proof:lem_alg_gradient_func} under Section~\ref{sec:proof_algorithm}. 
\end{proof}

\begin{lemma}  \label{lem:alg_neg_cuv_func}
Suppose that the trust-region size $\Delta \le 1$, 
$
\mb w^* \nabla^2 g(\mb w) \mb w/\norm{\mb w}{}^2 \le - \betaconcave,
$
for some $\betaconcave$, and that $\mb w^* \nabla^2 g(\mb w) \mb w/\norm{\mb w}{}^2$ is $\Lconcave$ Lipschitz on the open ball $\mc B\left(\mb w, \frac{3\Delta}{2\pi \sqrt{n}} \right)$ centered at $\mb w$. Then there exists a tangent vector $\mb \delta \in T_{\mb q} \bb S^{n-1}$ with $\norm{\mb \delta}{} \le \Delta$, such that
\begin{equation*}
f( \exp_{\mb q}(\mb \delta) ) \;\le\; f( \mb q) - \min \set{ \frac{2 \betaconcave^3}{3 \Lconcave^2}, \frac{3\Delta^2 \betaconcave}{8\pi^2 n} }. 
\end{equation*}
\end{lemma}
\begin{proof}
See Page~\pageref{proof:lem_alg_neg_cuv_func} under Section~\ref{sec:proof_algorithm}. 
\end{proof}
One can take $\beta_g = \betaconcave = c_\star \theta$ as shown in Theorem~\ref{thm:geometry_orth}, and take the Lipschitz results in Proposition~\ref{prop:lip-gradient} and Proposition~\ref{prop:lip-hessian-negative} (note that $\norm{\mb X_0}{\infty} \le 4 \log^{1/2} (np)$ w.h.p. by Lemma~\ref{lem:X-infinty-tail-bound}), repeat the argument for other $2n-1$ symmetric regions, and conclude that w.h.p. the objective value decreases by at least a constant amount. The next proposition summarizes the results. 

\begin{proposition}\label{lem:TRM-lemma-ii}
Assume~\eqref{eq:trm_proof_assumed_cond}.	In regions $\rII$ and $\rIII$, each trust-region step reduces the objective value by at least 
\begin{align} \label{eq:alg_thm_r23_dec}
\dII = \frac{1}{2}  \min\paren{\frac{c_\star^2 c_a \theta^2 \mu}{ n^2 \log\paren{np}}, \frac{3\Delta c_\star \theta}{4\pi\sqrt{n}}},  \quad \text{and} \quad \dIII = \frac{1}{2} \min\paren{\frac{c_\star^3c_b\theta^3 \mu^4}{n^6 \log^3\paren{np}}, \frac{3\Delta^2 c_\star \theta}{8\pi^2 n}}
\end{align}
respectively, provided that
\begin{align} \label{eqn:TRM-lemma-ii-iii-0}
	\Delta < \frac{c_c c_\star \theta \mu^2}{n^{5/2} \log^{3/2}\paren{np}},
\end{align}
where $c_a$ to $c_c$ are positive constants, and $c_\star$ is as defined in Theorem~\ref{thm:geometry_orth}. 
\end{proposition}
\begin{proof}
We only consider the symmetric section in the vicinity of $\mb e_n$ and the claims carry on to others by symmetry. If the current iterate $\mb q^{(r)}$ is in the region $\rII$, by Theorem \ref{thm:geometry_orth}, w.h.p., we have $\mb w^* g\paren{\mb w}/\norm{\mb w}{} \ge c_\star \theta$ for the constant $c_\star$. By Proposition~\ref{prop:lip-gradient} and Lemma~\ref{lem:X-infinty-tail-bound}, w.h.p., $\mb w^* g\paren{\mb w}/\norm{\mb w}{}$ is $C_2n^2\log\paren{np}/\mu$-Lipschitz. Therefore, By Lemma~\ref{lem:alg_approx_bd2} and Lemma~\ref{lem:alg_gradient_func}, a trust-region step decreases the objective value by at least
    \begin{align*}
    \dII \doteq \min\paren{\frac{c_\star^2 \theta^2 \mu}{2C_2 n^2 \log\paren{np}}, \frac{3 c_\star \theta\Delta}{4\pi\sqrt{n}}} - \frac{c_0n^{3/2}\log^{3/2}\paren{np}}{3\mu^2} \Delta^3. 
\end{align*}
	Similarly, if $\mb q^{(r)}$ is in the region $\rIII$, by Proposition \ref{prop:lip-hessian-negative}, Theorem \ref{thm:geometry_orth} and Lemma~\ref{lem:X-infinty-tail-bound}, w.h.p., $\mb w^* \nabla^2 g\paren{\mb w} \mb w/\norm{\mb w}{}^2$ is $C_3n^3\log^{3/2}\paren{np}/\mu^2$-Lipschitz and upper bounded by $-c_\star \theta$. By Lemma~\ref{lem:alg_approx_bd2} and Lemma~\ref{lem:alg_neg_cuv_func}, a trust-region step decreases the objective value by at least
	\begin{align*}
\dIII \doteq \min\paren{\frac{2c_\star^3\theta^3 \mu^4}{3C_3^2n^6 \log^3\paren{np}}, \frac{3\Delta^2 c_\star \theta}{8\pi^2 n}} - \frac{c_0n^{3/2}\log^{3/2}\paren{np}}{3\mu^2} \Delta^3. 
\end{align*}
It can be easily verified that when $\Delta$ obeys~\eqref{eq:alg_thm_r23_dec}, \eqref{eqn:TRM-lemma-ii-iii-0} holds. 
\end{proof}

The analysis for $\rI$ is slightly trickier. In this region, near each local minimizer, the objective function is strongly convex. So we still expect each trust-region step decreases the objective value. On the other hand, it is very unlikely that we can provide a universal lower bound for the amount of decrease - as the iteration sequence approaches a local minimizer, the movement is expected to be diminishing. Nevertheless, close to the minimizer the trust-region algorithm takes ``unconstrained'' steps. For constrained $\rI$ steps, we will again show reduction in objective value by at least a fixed amount; for unconstrained step, we will show the distance between the iterate and the nearest local minimizer drops down rapidly. 

The next lemma concerns the function value reduction for constrained $\rI$ steps. 
\begin{lemma}  \label{lem:alg_strcvx_func}
Suppose the trust-region size $\Delta \le 1$, and that at a given iterate $r$, $\Hess f\paren{\mb q^{(r)}} \succeq m_H \mc P_{T_{\mb q^{(r)}}\bb S^{n-1} }$, and $\norm{\Hess f\paren{\mb q^{(r)}}}{} \le M_H$. Further assume the optimal solution $\mb \delta_\star\in T_{\mb q^{(r)}}\bb S^{n-1}$ to the trust-region subproblem~\eqref{eqn:subproblem-1} satisfies $\norm{\mb \delta_\star}{} = \Delta$, i.e., the norm constraint is active. Then there exists a tangent vector $\mb \delta \in T_{\mb q^{(r)}} \bb S^{n-1}$ with $\norm{\mb \delta}{} \le \Delta$, such that
\begin{equation*}
f( \exp_{\mb q^{(r)}}(\mb \delta) ) \;\le\; f\paren{\mb q^{(r)}} - \frac{m_H^2 \Delta^2}{M_H} + \frac{1}{6} \eta_f\Delta^3,
\end{equation*}
where $\eta_f$ is defined the same as Lemma~\ref{lem:alg_approx_bd2}.
\end{lemma}
\begin{proof}
See Page~\pageref{proof:lem_alg_strcvx_func} under Section~\ref{sec:proof_algorithm}. 
\end{proof}
The next lemma provides an estimate of $m_H$. Again we will only state the result for the ``canonical'' section with the ``canonical'' $\mb q(\mb w)$ mapping.  
\begin{lemma} \label{lem:alg_strcvx_lb}
There exists a positive constant $C$, such that for all $\theta \in \paren{0, 1/2}$ and $\mu < \theta/10$, whenever $p \ge Cn^3 \log \frac{n}{\theta \mu}/(\mu \theta^2)$, it holds with probability at least $1 -  c p^{-7}$ that for all $\mb q$ with $\norm{\mb w\paren{\mb q}}{} \le \mu/(4\sqrt{2})$, 
\begin{align*}
\Hess f\paren{\mb q} \succeq c_\star \frac{\theta}{\mu} \mc P_{T_{\mb q} \bb S^{n-1}}.  
\end{align*}
Here $c_\star$ is as in Theorem~\ref{thm:geometry_orth} and Theorem~\ref{thm:geometry_comp}, and $c > 0$ is another constant. 
\end{lemma}
\begin{proof}
See Page~\pageref{proof:lem_alg_strcvx_lb} under Section~\ref{sec:proof_algorithm}. 
\end{proof}
We know that $\norm{\mb X_0}{\infty} \le 4\log^{1/2} (np)$ w.h.p., and hence by the definition of Riemannian Hessian and Lemma~\ref{lem:mag_lip_fq}, 
\begin{align*}
M_H \doteq \norm{\Hess f(\mb q)}{} & \le \norm{\nabla^2 f(\mb q)}{} + \norm{\nabla f(\mb q)}{} \le M_{\nabla^2} + M_{\nabla} \le \frac{2n}{\mu} \norm{\mb X_0}{\infty} ^2 \le \frac{16n}{\mu} \log (np). 
\end{align*}
Combining this estimate and Lemma~\ref{lem:alg_strcvx_lb} and Lemma~\ref{lem:alg_approx_bd2}, we obtain a concrete lower bound for the reduction of objective value for each constrained $\rI$ step. 
\begin{proposition}\label{lem:TRM-lemma-iii}
Assume~\eqref{eq:trm_proof_assumed_cond}. Each constrained $\rI$ trust-region step (i.e., $\norm{\mb \delta}{} = \Delta$) reduces the objective value by at least 
\begin{align}\label{eqn:TRM-decrease-d-i}
	\dI = \frac{c c_\star^2\theta^2 }{\mu n\log (np) }\Delta^2,
\end{align}
provided 
\begin{align}\label{eqn:TRM-lemma-ii-0}
 \Delta \leq \frac{c'c_\star^2 \theta^2 \mu  }{n^{5/2} \log^{5/2} (np) }. 
\end{align}
Here $c_\star$ is as in Theorem~\ref{thm:geometry_orth} and Theorem~\ref{thm:geometry_comp}, and $c, c'$ are positive constants.
\end{proposition}
\begin{proof}
We only consider the symmetric section in the vicinity of $\mb e_n$ and the claims carry on to others by symmetry. We have that w.h.p. 
\begin{align*}
\norm{\Hess f(\mb q)}{} \le \frac{16n}{\mu} \log(np), \quad \text{and} \quad  \Hess f(\mb q) \succeq c_\star\frac{\theta}{\mu} \mc P_{T_{\mb q} \bb S^{n-1}}. 
\end{align*} 
Combining these estimates with Lemma~\ref{lem:alg_approx_bd2} and Lemma~\ref{lem:alg_strcvx_func}, one trust-region step will find next iterate $\mb q^{(r+1)}$ that decreases the objective value by at least 
\begin{align*}
\dI \doteq \frac{c_\star^2 \theta^2/\mu^2 }{2n\log\paren{np}/\mu} \Delta^2 - \frac{c_0n^{3/2} \log^{3/2}\paren{np}}{\mu^2} \Delta^3.
\end{align*}
Finally, by the condition on $\Delta$ in \eqref{eqn:TRM-lemma-ii-0} and the assumed conditions~\eqref{eq:trm_proof_assumed_cond}, we obtain 
\begin{align*}
	\dI \geq \frac{c_\star^2\theta^2 }{2\mu n\log (np) }\Delta^2 - \frac{c_0n^{3/2} \log^{3/2}\paren{np}}{\mu^2} \Delta^3 \geq \frac{c_\star^2\theta^2 }{4\mu n\log (np) }\Delta^2, 
\end{align*}
as desired.
\end{proof}

By the proof strategy for $\rI$ we sketched before Lemma~\ref{lem:alg_strcvx_func}, we expect the iteration sequence ultimately always takes unconstrained steps when it moves very close to a local minimizer. We will show that the following is true: when $\Delta$ is small enough, once the iteration sequence starts to take unconstrained $\rI$ step, it will take consecutive unconstrained $\rI$ steps afterwards. It takes two steps to show this: (1) upon an unconstrained $\rI$ step, the next iterate will stay in $\rI$. It is obvious we can make $\Delta \in O(1)$ to ensure the next iterate stays in $\rI \cup \rII$. To strengthen the result, we use the gradient information. From Theorem~\ref{thm:geometry_orth}, we expect the magnitudes of the gradients in $\rII$ to be lower bounded; on the other hand, in $\rI$ where points are near local minimizers, continuity argument implies that the magnitudes of gradients should be upper bounded. We will show that when $\Delta$ is small enough, there is a gap between these two bounds, implying the next iterate stays in $\rI$; (2) when $\Delta$ is small enough, the step is in fact unconstrained. Again we will only state the result for the ``canonical'' section with the ``canonical'' $\mb q(\mb w)$ mapping. The next lemma exhibits an absolute lower bound for magnitudes of gradients in $\rII$. 
\begin{lemma} \label{lem:alg_gradient_lb}
For all $\mb q$ satisfying $\mu/(4\sqrt{2}) \le \norm{\mb w\paren{\mb q}}{} \le 1/(20\sqrt{5})$, it holds that 
\begin{align*}
\norm{\grad f\paren{\mb q}}{} \ge \frac{9}{10} \frac{\mb w^* \nabla g\paren{\mb w}}{\norm{\mb w}{}}. 
\end{align*}
\end{lemma}
\begin{proof}
See Page~\pageref{proof:lem_alg_gradient_lb} under Section~\ref{sec:proof_algorithm}. 
\end{proof}
Assuming~\eqref{eq:trm_proof_assumed_cond}, Theorem~\ref{thm:geometry_orth} gives that w.h.p. $\mb w^* \nabla g(\mb w)/\norm{\mb w}{} \ge c_\star \theta$. Thus, w.h.p, $\norm{\grad f(\mb q)}{} \ge 9c_\star\theta/10$ for all $\mb q \in \rII$. The next lemma compares the magnitudes of gradients before and after taking one unconstrained $\rI$ step. This is crucial to providing upper bound for magnitude of gradient for the next iterate, and also to establishing the ultimate (quadratic) sequence convergence.
\begin{lemma}\label{lem:TR-step} 
Suppose the trust-region size $\Delta \le 1$, and at a given iterate $r$, $\Hess f \paren{\mb q^{(r)}} \succeq m_H \mc P_{ T_{\mb q^{(r)}}\bb S^{n-1}}$, and that the unique minimizer $\mb \delta_\star\in T_{\mb q^{(r)}}\bb S^{n-1} $ to the trust region subproblem \eqref{eqn:subproblem-1} satisfies $\norm{\mb \delta_\star}{} < \Delta$ (i.e., the constraint is inactive). Then, for $\mb q^{(r+1)} = \exp_{\mb q^{(r)}}\paren{\mb \delta_\star} $, we have 
\begin{equation*}
\|\grad f (\mb q^{(r+1)})\| \;\le\; \frac{L_H}{2m_H^2} \|\grad f(\mb q^{(r)}) \|^2, 
\end{equation*}
where $L_H \doteq 5n^{3/2}/({2\mu^2}) \norm{\mb X_0}{\infty}^3 + 9n/\mu \norm{\mb X_0}{\infty}^2 + 9\sqrt{n} \norm{\mb X_0}{\infty}$.
\end{lemma}
\begin{proof}
See Page~\pageref{proof:lem_TR-step} under Section~\ref{sec:proof_algorithm}. 
\end{proof}
We can now bound the Riemannian gradient of the next iterate as 
\begin{align*}
\|\grad f (\mb q^{(r+1)})\| 
& \le \frac{L_H}{2m_H^2} \|\grad f(\mb q^{(r)}) \|^2 \\
& \le \frac{L_H}{2m_H^2} \|[\mb U^* \Hess f(\mb q^{(r)}) \mb U] [\mb U^* \Hess f(\mb q^{(r)}) \mb U]^{-1} \grad f(\mb q^{(r)}) \|^2 \\
& \le  \frac{L_H}{2m_H^2}  \norm{\Hess f(\mb q^{(r)})}{}^2 \Delta^2 = \frac{L_H M_H^2}{2m_H^2} \Delta^2. 
\end{align*}
Obviously, one can make the upper bound small by tuning down $\Delta$. Combining the above lower bound for $\norm{\grad f(\mb q)}{}$ for $\mb q \in \rII$, one can conclude that when $\Delta$ is small, the next iterate $\mb q^{(r+1)}$ stays in $\rI$. Another application of the optimality condition~\eqref{eqn:ts-optimal-solution-1} gives conditions on $\Delta$ that guarantees the next trust-region step is also unconstrained. Detailed argument can be found in proof of the following proposition. 

\begin{proposition}\label{lem:TRM-lemma-iv}
Assume~\eqref{eq:trm_proof_assumed_cond}. W.h.p, once the trust-region algorithm takes an unconstrained $\rI$ step (i.e., $\norm{\mb \delta}{}<\Delta$), it always takes unconstrained $\rI$ steps, provided that
\begin{align}
 \Delta \le  \frac{c c_\star^3 \theta^3 \mu}{n^{7/2} \log^{7/2}\paren{np}}, 
\end{align} 	
Here $c_\star$ is as in Theorem~\ref{thm:geometry_orth} and Theorem~\ref{thm:geometry_comp}, and $c > 0$ is another constant. 
\end{proposition}

\begin{proof}
We only consider the symmetric section in the vicinity of $\mb e_n$ and the claims carry on to others by symmetry. Suppose that step $k$ is an unconstrained $\rI$ step. Then 
\begin{align*}
\|\mb w(\mb q^{(r+1)}) - \mb w(\mb q^{(r)})\| 
& \le \|\mb q^{(r+1)} - \mb q^{(r)}\|  = \|\exp_{\mb q^{(r)}(\mb \delta)} - \mb q^{(r)} \| \\ 
& = \sqrt{2-2\cos \norm{\mb \delta}{}} = 2\sin(\norm{\mb \delta}{}/2) \le \norm{\mb \delta}{} < \Delta. 
\end{align*}
Thus, if $\Delta \le \tfrac{1}{20\sqrt{5}} - \tfrac{\mu}{4\sqrt{2}}$, $\mb q^{(r+1)}$ will be in $\rI \cup \rII$. Next, we show that if $\Delta$ is sufficiently small, $\mb q^{(r+1)}$ will be indeed in $\rI$. By Lemma \ref{lem:TR-step}, 
\begin{align}
\norm{\grad f \paren{\mb q^{(r+1)} } }{} 
& \le \frac{L_H}{2m_H^2} \norm{ \grad f \paren{\mb q^{(r)}} }{}^2 \nonumber \\
& \le \frac{L_H M_H^2}{2m_H^2} \norm{ \brac{\mb U^* \Hess f \paren{\mb q^{(r)}} \mb U}^{-1} \mb U^*\grad f \paren{\mb q^{(r)}} }{}^2 
 \le \frac{L_H M_H^2}{2m_H^2} \Delta^2,  \label{eqn:TRM-lemma-iii-3}
 \end{align}
where we have used the fact that 
\begin{align*}
\norm{\mb \delta^{(r)}}{} = \norm{ \brac{\mb U^* \Hess f \paren{\mb q^{(r)}} \mb U}^{-1} \mb U^*\grad f \paren{\mb q^{(r)}} }{} < \Delta,  
\end{align*}
as the step is unconstrained. On the other hand, by Theorem~\ref{thm:geometry_orth} and Lemma~\ref{lem:alg_gradient_lb}, w.h.p. 
\begin{align}
 \norm{\grad f\paren{\mb q}}{} \ge \beta_{\grad} \doteq \frac{9}{10} c_\star \theta,  \quad \forall \; \mb q \in \rII. \label{eqn:TRM-lemma-iii-4}
\end{align}
Hence, provided 
\begin{align}\label{eqn:Delta-bound-1}
	\Delta < \frac{m_H}{M_H}\sqrt{ \frac{2\betagrad}{L_H}},
\end{align}
we have $\mb q^{(r+1)} \in \rI$. 

We next show that when $\Delta$ is small enough, the next step is also unconstrained. Straight forward calculations give 
\begin{align*}
\norm{\mb U \brac{\mb U^* \Hess f \paren{\mb q^{(r+1)}} \mb U}^{-1} \mb U^*\grad f \paren{\mb q^{(r+1)}} }{} \leq \frac{L_H M_H^2}{2m_H^3} \Delta^2. 
\end{align*}
Hence, provided that 
 \begin{equation} \label{eqn:Delta-bound-2}
 \Delta < \frac{2m_H^3}{L_H M_H^2},
 \end{equation}
we will have
\begin{align*}
\norm{\mb U \brac{\mb U^* \Hess f \paren{\mb q^{(r+1)}} \mb U}^{-1} \mb U^*\grad f \paren{\mb q^{(r+1)}} }{} < \Delta; 
\end{align*}
in words, the minimizer to the trust-region subproblem for the next step lies in the relative interior of the trust region - the constraint is inactive. By Lemma~\ref{lem:TR-step} and Lemma \ref{lem:X-infinty-tail-bound}, we have 
\begin{align}
L_H \;=\; C_1 n^{3/2} \log^{3/2}\paren{np}/\mu^2, \label{eqn:lip-value-bound}
\end{align}
w.h.p.. Combining this and our previous estimates of $m_H$, $M_H$, we conclude whenever
\begin{align*}
\Delta 
\leq  \min \set{\frac{1}{20\sqrt{5}} - \frac{\mu}{4\sqrt{2}}, \frac{c_1c_\star^{3/2} \theta^{3/2} }{n^{7/4} \log^{7/4}\paren{np}}, \frac{c_2\mu c_\star^3 \theta^3 }{n^{7/2} \log^{7/2}\paren{np}}}, 
\end{align*}
w.h.p., our next trust-region step is also an unconstrained $\rI$ step. Simplifying the above bound completes the proof. 
\end{proof}

Finally, we want to show that ultimate unconstrained $\rI$ iterates actually converges to one nearby local minimizer rapidly. Lemma~\ref{lem:TR-step} has established the gradient is diminishing. The next lemma shows the magnitude of gradient serves as a good proxy for distance to the local minimizer. 
\begin{lemma}\label{lem:TR-grad-opt} 
Let $\mb q_\star \in \bb S^{n-1}$ such that $\grad f (\mb q_\star) = \mb 0$, and $\mb \delta \in T_{\mb q_\star}\bb S^{n-1}$. Consider a geodesic $\gamma(t) = \exp_{\mb q_\star}(t \mb \delta)$, and suppose that on $[0, \tau]$, $\Hess f (\gamma(t)) \succeq m_H \mc P_{T_{\gamma(t)} \bb S^{n-1}}$. Then 
\begin{equation*}
\norm{ \grad f (\gamma(\tau)) }{} \;\ge\; m_H \tau \norm{\mb \delta}{}. 
\end{equation*}
\end{lemma}
\begin{proof}
See Page~\pageref{proof:lem_TR-grad-opt} under Section~\ref{sec:proof_algorithm}. 
\end{proof}
To see this relates the magnitude of gradient to the distance away from the nearby local minimizer, w.l.o.g., one can assume $\tau = 1$ and consider the point $\mb q = \exp_{\mb q_\star}(\mb \delta)$. Then 
\begin{align*}
\norm{\mb q_\star - \mb q}{} = \norm{\exp_{\mb q_\star}(\mb \delta) - \mb q}{} = \sqrt{2 - 2\cos \norm{\mb \delta}{}} = 2\sin (\norm{\mb \delta}{}/2) \le \norm{\mb \delta}{} \le \norm{\grad f(\mb q)}{}/m_H, 
\end{align*}
where at the last inequality above we have used Lemma~\ref{lem:TR-grad-opt}. Hence, combining this observation with Lemma~\ref{lem:TR-step}, we can derive the asymptotic sequence convergence rate as follows.

\begin{proposition}\label{lem:TRM-lemma-v}
Assume~\eqref{eq:trm_proof_assumed_cond} and the conditions in Lemma \ref{lem:TRM-lemma-iv}. Let $\mb q^{(r_0)} \in \rI$ and the $r_0$-th step the first unconstrained $\rI$ step and $\mb q_\star$ be the unique local minimizer of $f$ over one connected component of $\rI$ that contains $\mb q^{(r_0)}$. Then w.h.p., for any positive integer $r'\geq 1$, 
\begin{align}
	\norm{\mb q^{(r_0 + r')} - \mb q_\star }{} \; \leq\;  \frac{cc_\star \theta \mu}{n^{3/2} \log^{3/2} \paren{np}} 2^{- 2^{r'}},
\end{align}
provided that 
\begin{align}
	\Delta \leq \frac{c'c_\star^2 \theta^2 \mu}{n^{5/2} \log^{5/2} (np) }. 
\end{align}
Here $c_\star$ is as in Theorem~\ref{thm:geometry_orth} and Theorem~\ref{thm:geometry_comp}, and $c$, $c'$ are both positive constants. 
\end{proposition}
\begin{proof}
By the geometric characterization in Theorem~\ref{thm:geometry_orth} and corollary II.2 in~\cite{sun2015complete_a}, $f$ has $2n$ separated local minimizers, each located in $\rI$ and within distance $\sqrt{2}\mu/16$ of one of the $2n$ signed basis vectors $\{\pm \mb e_i\}_{i \in [n]}$. Moreover, it is obvious when $\mu \le 1$, $\rI$ consists of $2n$ disjoint connected components. We only consider the symmetric component in the vicinity of $\mb e_n$ and the claims carry on to others by symmetry. 

Suppose that $r_0$ is the index of the first unconstrained iterate in region $\rI$, i.e., $\mb q^{(r_0)} \in \rI$. By Lemma \ref{lem:TR-step}, for any integer $r'\geq 1$, we have
	 \begin{align}\label{eqn:TRM-lemma-v-1}
	 	\norm{\grad f \paren{\mb q^{(r_0 + r')} } }{} \;\leq\;  \frac{2 m_H^2}{L_H}\left( \frac{L_H}{2 m_H^2}  \norm{\grad f \paren{\mb q^{(r_0)}} }{} \right)^{2^{r'}}.
	 \end{align}
	 where $L_H$ is as defined in Lemma~\ref{lem:TR-step}, $m_H$ as the strong convexity parameter for $\rI$ defined above.   
	 
	 Now suppose $\mb q_\star$ is the unique local minimizer of $f$, lies in the same $\rI$ component that $q^{(r_0)}$ is located. Let $\gamma_{r'}(t) = \exp_{\mb q_\star}\paren{t\mb \delta}$ to be the unique geodesic that connects $\mb q_\star$ and $\mb q^{(r_0+r')}$ with $\gamma_{r'}(0) = \mb q_\star$ and $\gamma_{r'}(1) = \mb q^{(r_0 + r')}$. We have 
	 \begin{align*}
	 \norm{\mb q^{(r_0 + r')} - \mb q_\star}{} 
	 & \le \norm{\exp_{\mb q_\star}(\mb \delta) - \mb q_\star}{} = \sqrt{2- 2\cos \norm{\mb \delta}{}} = 2 \sin(\norm{\mb \delta}{}/2) \\
	 & \le \norm{\mb \delta}{} \le \frac{ 1 }{ m_H }\norm{ \grad f \paren{\mb q^{(r_0 + r')}} }{}
 \leq  \frac{2m_H}{L_H} \left( \frac{L_H}{2m_H^2} \norm{\grad f \paren{\mb q^{(r_0)}} }{} \right)^{2^{r'}}, 
	 \end{align*}
where at the second line we have repeatedly applied Lemma \ref{lem:TR-grad-opt}. 

By the optimality condition~\eqref{eqn:ts-optimal-solution-1} and the fact that $\norm{\mb \delta^{(r_0)}}{} < \Delta$, we have
\begin{align*}
\frac{L_H}{2 m_H^2} \norm{\grad f \paren{\mb q^{(r_0)} } }{} 
& \le \frac{L_H}{2 m_H^2} M_H \norm{ \brac{\mb U^* \Hess f \paren{\mb q^{(r_0)}} \mb U}^{-1} \mb U^* \grad f \paren{\mb q^{(r_0)} } }{} 
\le \frac{L_H M_H}{2m_H^2} \Delta.
\end{align*}
Thus, provided
\begin{align}\label{eqn:trust-region-size-v}
	\Delta < \frac{m_H^2}{L_H M_H},
\end{align} 
we can combine the above results and obtain 
\begin{align*}
	\norm{\mb q^{(r_0 + r')} - \mb q_\star }{} \;\le\; \frac{2m_H}{L_H} 2^{- 2^{r'}}.
\end{align*}
Based on the previous estimates for $m_H$, $M_H$ and $L_H$, we obtain that w.h.p., 
\begin{align*}
	\norm{\mb q^{(r_0 + r')} - \mb q_\star }{} \; \leq\;  \frac{c_1 c_\star \theta \mu}{n^{3/2} \log^{3/2} \paren{np}} 2^{- 2^{r'}}.
\end{align*}
Moreover, by \eqref{eqn:trust-region-size-v}, w.h.p., it is sufficient to have the trust region size
\begin{align*}
	\Delta \leq \frac{c_2 c_\star^2 \theta^2 \mu }{n^{5/2} \log^{5/2} (np) }.
\end{align*}
Thus, we complete the proof.
\end{proof}

Now we are ready to piece together the above technical proposition to prove Theorem~\ref{thm:trm_orth}. 

\js{
\begin{proof}{(\textbf{of Theorem~\ref{thm:trm_orth}})}
Assuming~\eqref{eq:trm_proof_assumed_cond} and in addition that 
	\begin{align*}
	 \Delta <  \frac{c_1 c_\star^3 \theta^3 \mu^2}{n^{7/2}\log^{7/2}\paren{np}}, 
	\end{align*}
it can be verified that the conditions of all the above propositions are satisfied. 

By the preceding four propositions, a step will either be $\rIII$, $\rII$, or constrained $\rI$ step that decreases the objective value by at least a certain fixed amount (we call this \emph{Type A}), or be an unconstrained $\rI$ step (\emph{Type B}), such that all future steps are unconstrained $\rI$ and the sequence converges to a local minimizer quadratically. Hence, regardless the initialization, the whole iteration sequence consists of consecutive Type A steps, followed by consecutive Type B steps. Depending on the initialization, either the Type A phase or the Type B phase can be absent. In any case, from $\mb q^{(0)}$ it takes at most (note $f(\mb q) \ge 0$ always holds)
\begin{align}
\frac{f\paren{\mb q^{(0)}}}{\min\Brac{\dI, \dII, \dIII }}
\end{align}
steps for the iterate sequence to start take consecutive unconstrained $\rI$ steps, or to already terminate. In case the iterate sequence continues to take consecutive unconstrained $\rI$ steps, Proposition~ \ref{lem:TRM-lemma-v} implies that it takes at most
\begin{align}
\log\log \paren{\frac{c_2c_\star \theta \mu}{\eps n^{3/2} \log^{3/2}\paren{np}}}
\end{align}
steps to obtain an $\eps$-near solution to the $\mb q_\star$ that is contained in the connected subset of $\rI$ that the sequence entered. 

Thus, the number of iterations to obtain an $\eps$-near solution to $\mb q_\star$ can be grossly bounded by
\begin{align*}
\#\text{Iter} &\;\leq\;  \frac{f\paren{\mb q^{(0)}}}{\min\Brac{\dI, \dII, \dIII }} \;+\;\log\log \paren{\frac{c_2c_\star \theta \mu}{\eps n^{3/2} \log^{3/2}\paren{np}}} \nonumber\\
 &\;\leq\;  \brac{\min\set{\frac{c_3c_{\star}^3 \theta^3 \mu^4}{n^6 \log^3\paren{np}}, \frac{c_4c_\star^2 \theta^2 }{n}\Delta^2 }}^{-1}f\paren{\mb q^{(0)}} \;+\; \log\log \paren{\frac{c_2c_\star \theta \mu}{\eps n^{3/2} \log^{3/2}\paren{np}}}. 
\end{align*}
Finally, the claimed failure probability comes from a simple union bound with careful bookkeeping. 
\end{proof}
}

\subsection{Extending to Convergence for Complete Dictionaries} \label{sec:alg_comp}
\input{sec/proof_trm_comp}

%% file: sec/proof_trm_comp.tex
Recall that in this case we consider the preconditioned input
\begin{align} \label{eq:precon_def}
\overline{\mb Y} \doteq \sqrt{p \theta} (\mb Y \mb Y^*)^{-1/2} \mb Y. 
\end{align}
Note that for any complete $\mb A_0$ with condition number $\kappa\paren{\mb A_0}$, from Lemma~\ref{lem:pert_key_mag} we know when $p$ is large enough, w.h.p. one can write the preconditioned $\ol{\mb Y}$ as 
\begin{align*}
\ol{\mb Y} = \mb U \mb V^* \mb X_0 + \mb \Xi \mb X_0
\end{align*}
for a certain $\mb \Xi$ with small magnitude, and $\mb U \mb \Sigma \mb V^* = \mathtt{SVD}\paren{\mb A_0}$. Particularly, when $p$ is chosen by Theorem~\ref{thm:geometry_comp}, the perturbation is bounded as 
\begin{align} \label{eq:pert_upper_bound}
\|\widetilde{\mb \Xi}\| \le c c_\star \theta \paren{\max\set{\frac{n^{3/2}}{\mu^2}, \frac{n^2}{\mu}} \log^{3/2}\paren{np}}^{-1}
\end{align}
for a certain constant $c$ which can be made arbitrarily small by making the constant $C$ in $p$ large. Since $\mb U \mb V^*$ is orthogonal, 
\begin{align*}
f\paren{\mb q; \mb U \mb V^* \mb X_0 + \mb \Xi \mb X_0} = f\paren{\mb V \mb U^* \mb q; \mb X_0 + \mb V \mb U^* \mb \Xi \mb X_0}. 
\end{align*}
In words, the function landscape of $f(\mb q; \mb U \mb V^* \mb X_0 + \mb \Xi \mb X_0)$ is a rotated version of that of $f(\mb q; \mb X_0 + \mb V \mb U^* \mb \Xi \mb X_0)$. Thus, any local minimizer $\mb q_\star$ of $f(\mb q; \mb X_0 + \mb V \mb U^* \mb \Xi \mb X_0)$ is rotated to $\mb U \mb V^* \mb q_\star$, one minimizer of $f(\mb q; \mb U \mb V^* \mb X_0 + \mb \Xi \mb X_0)$. Also if our algorithm generates iteration sequence $\mb q_0, \mb q_1, \mb q_2, \dots$ for $f(\mb q; \mb X_0 + \mb V \mb U^* \mb \Xi \mb X_0)$ upon initialization $\mb q_0$, it will generate the iteration sequence $\mb U \mb V^* \mb q_0$, $\mb U \mb V^* \mb q_1$, $\mb U \mb V^* \mb q_2, \dots$ for $f\paren{\mb q;  \mb U \mb V^* \mb X_0 + \mb \Xi \mb X_0}$. So w.l.o.g. it is adequate that we prove the convergence results for the case $f(\mb q; \mb X_0 + \mb V \mb U^* \mb \Xi \mb X_0)$, corresponding to $\bm A_0 = \mb I$ with perturbation $\wt{\mb \Xi} \doteq \mb V \mb U^* \mb \Xi$. So in this section (Section~\ref{sec:alg_comp}), we write $f(\mb q; \wt{\mb X_0})$ to mean $f(\mb q; \mb X_0 + \wt{\mb \Xi} \mb X_0)$. 

Theorem~\ref{thm:geometry_comp} has shown that when 
\begin{align} \label{eq:trm_proof_assumed_cond_comp}
\theta \in \paren{0, \frac{1}{2}}, \; \mu \le c\min\set{\frac{\theta}{n}, \frac{1}{n^{5/4}}}, \; p \ge \frac{C}{c_\star^2 \theta^2} \max\set{\frac{n^4}{\mu^4}, \frac{n^5}{\mu^2}} \kappa^8 \paren{\mb A_0} \log^4\paren{\frac{\kappa\paren{\mb A_0}n}{\mu \theta}}, 
\end{align}
the geometric structure of the landscape is qualitatively unchanged from the orthogonal case, and the parameter $c_\star$ constant can be replaced with $c_\star/2$. Particularly, for this choice of $p$, Lemma~\ref{lem:pert_key_mag} implies 
\begin{align} \label{eq:trm_conv_comp_pert_bound}
\|\wt{\mb \Xi}\| = \| \mb V \mb U^*\mb \Xi\|  \le c c_\star \theta \paren{\max\set{\frac{n^{3/2}}{\mu^2}, \frac{n^2}{\mu}} \log^{3/2}\paren{np}}^{-1} 
\end{align}
for a constant $c$ that can be made arbitrarily small by setting the constant $C$ in $p$ sufficiently large. The whole proof is quite similar to that of orthogonal case in the last section. We will only sketch the major changes below. To distinguish with the corresponding quantities in the last section, we use $\wt{\cdot}$ to denote the corresponding perturbed quantities here. 
\begin{itemize}
\item Lemma~\ref{lem:mag_lip_fq}: Note that  
\begin{align*}
\|\mb X_0 + \wt{\mb \Xi} \mb X_0\|_{\infty} \le \norm{\mb X_0}{\infty} + \|\wt{\mb \Xi}\mb X_0\|_{\infty} \le \|\mb X_0\|_\infty + \sqrt{n} \|\wt{\mb \Xi}\| \|\mb X_0\|_\infty \le 3\|\mb X_0\|_\infty/2, 
\end{align*}
where by~\eqref{eq:trm_conv_comp_pert_bound} we have used $\|\wt{\mb \Xi}\| \le 1/(2\sqrt{n})$ to simplify the above result. So we obtain  
\begin{align*}
\wt{M}_{\nabla} \le \frac{3}{2} M_{\nabla}, \; \wt{M}_{\nabla^2} \le \frac{9}{4} M_{\nabla^2}, \; \wt{L}_{\nabla} \le \frac{9}{4} L_{\nabla}, \; \wt{L}_{\nabla^2} \le \frac{27}{8} L_{\nabla^2}. 
\end{align*}

\item Lemma~\ref{lem:alg_approx_bd2}: Now we have
\begin{align*}
\wt{\eta}_f \doteq \wt{M}_{\nabla} + 2\wt{M}_{\nabla^2} + \wt{L}_{\nabla} + \wt{L}_{\nabla^2} \le 4 \eta_f. 
\end{align*}

\item Lemma~\ref{lem:alg_gradient_func} and Lemma~\ref{lem:alg_neg_cuv_func} are generic and nothing changes. 

\item Proposition~\ref{lem:TRM-lemma-ii}: We have now $\mb w^* \mb g(\mb w; \wt{\mb X_0})/\norm{\mb w}{} \ge c_\star \theta/2$ by Theorem~\ref{thm:geometry_comp}, w.h.p. $\mb w^* \nabla g(\mb w; \wt{\mb X_0})/\norm{\mb w}{}$ is $C_1n^2 \log(np)/\mu$-Lipschitz by Proposition~\ref{prop:lip-gradient}, and $\norm{\mb X_0 + \wt{\mb \Xi} \mb X_0}{\infty} \le 3\norm{\mb X_0}{\infty}/2$ as shown above. Similarly, $\mb w^* \mb g(\mb w; \wt{\mb X_0})/\norm{\mb w}{} \le -c_\star \theta/2$ by Theorem~\ref{thm:geometry_comp}, and $\mb w^* \nabla^2 g(\mb w; \wt{\mb X_0}) \mb w/\norm{\mb w}{}^2$ is $C_2 n^3 \log^{3/2}(np) /\mu^2$-Lipschitz. Moreover, $\wt{\eta}_f \le 4\eta_f$ as shown above. Since there are only multiplicative constant changes to the various quantities, we conclude 
\begin{align}
\wt{\dII} = c_1 \dII, \quad \wt{\dIII} = c_1 \dIII
\end{align}
provided  
\begin{align}
\Delta < \frac{c_2 c_\star \theta \mu^2}{n^{5/2} \log^{3/2}\paren{np}}. 
\end{align}

\item Lemma~\ref{lem:alg_strcvx_func}: $\eta_f$ is changed to $\wt{\eta}_f$ with $\wt{\eta}_f \le 4\eta_f$ as shown above. 

\item Lemma~\ref{lem:alg_strcvx_lb}: By~\eqref{eq:fq_hess}, we have 
\begin{multline*}
\norm{\nabla^2 f(\mb q; \mb X_0) - \nabla^2 f(\mb q; \wt{\mb X_0}) }{} \le \frac{1}{p} \sum_{k=1}^p \Brac{L_{\ddot{h}} \|\wt{\mb \Xi}\| \norm{(\mb x_0)_k}{}^2 + \frac{1}{\mu}\norm{(\mb x_0)_k (\mb x_0)_k^* - \wt{(\mb x_0)}_k \wt{(\mb x_0)}_k^*}{}} \\
 \le \|\wt{\mb \Xi}\| \paren{L_{\ddot{h}} + 2/\mu + \|\wt{\mb \Xi}\|/\mu } \sum_{k=1}^p \norm{(\mb x_0)_k}{}^2 \le \|\wt{\mb \Xi}\| \paren{L_{\ddot{h}} + 3/\mu }n\norm{\mb X_0}{\infty}^2, 
\end{multline*}
where $L_{\ddot{h}}$ is the Lipschitz constant for the function $\ddot{h}_{\mu}\paren{\cdot}$ and we have used the fact that $\|\wt{\mb \Xi}\| \le 1$. Similarly, by~\ref{eq:fq_grad}, 
\begin{align*}
\norm{\nabla f(\mb q; \mb X_0) - \nabla f(\mb q; \wt{\mb X_0}) }{}
\le \frac{1}{p}\sum_{k=1}^p \Brac{L_{\dot{h}_\mu} \|\wt{\mb \Xi}\| \norm{(\mb x_0)_k}{} + \|\wt{\mb \Xi}\| \norm{(\mb x_0)_k}{}  } \le \paren{L_{\dot{h}_\mu} +1} \|\wt{\mb \Xi}\| \sqrt{n} \norm{\mb X_0}{\infty}, 
\end{align*}
where $L_{\dot{h}}$ is the Lipschitz constant for the function $\dot{h}_{\mu}\paren{\cdot}$. Since $L_{\ddot{h}} \le 2/\mu^2$ and $L_{\dot{h}} \le 1/\mu$, and $\norm{\mb X_0}{\infty} \le 4\sqrt{\log(np)}$ w.h.p. (Lemma~\ref{lem:X-infinty-tail-bound}). By~\eqref{eq:trm_conv_comp_pert_bound}, w.h.p. we have 
\begin{align*}
\norm{\nabla f(\mb q; \mb X_0) - \nabla f(\mb q; \wt{\mb X_0}) }{} \le \frac{1}{2} c_\star \theta, 
\quad\text{and}\quad  \norm{\nabla^2 f(\mb q; \mb X_0) - \nabla^2 f(\mb q; \wt{\mb X_0}) }{}  \le \frac{1}{2} c_\star \theta, 
\end{align*}
provided the constant $C$ in~\eqref{eq:trm_proof_assumed_cond_comp} for $p$ is large enough. Thus, by~\eqref{eq:fq_rie_hess} and the above estimates we have 
\begin{align*}
\norm{\Hess f(\mb q; \mb X_0) - \Hess f(\mb q; \wt{\mb X}_0)}{} & \le \norm{\nabla f(\mb q; \mb X_0) - \nabla f(\mb q; \wt{\mb X_0})}{} + \norm{\nabla^2 f(\mb q; \mb X_0) - \nabla^2 f(\mb q; \wt{\mb X_0})}{} \\
& \le c_\star \theta \le \frac{1}{2} c_\star \frac{\theta}{\mu}, 
\end{align*}
provided $\mu \le 1/2$. So we conclude 
\begin{align}
\Hess f(\mb q; \wt{\mb X}_0) \succeq \frac{1}{2}c_\star \frac{\theta}{\mu} \mc P_{T_{\mb q}} \bb S^{n-1} \Longrightarrow \wt{m_H} \ge \frac{1}{2}c_\star \frac{\theta}{\mu}. 
\end{align}

\item Proposition~\ref{lem:TRM-lemma-iii}: From the estimate of $M_H$ above Proposition~\ref{lem:TRM-lemma-iii} and the last point, we have 
\begin{align*}
\norm{\Hess f(\mb q; \wt{\mb X}_0)}{} \le \frac{36}{\mu} \log(np), \quad \text{and} \quad  \Hess f(\mb q; \wt{\mb X}_0) \succeq \frac{1}{2}c_\star \frac{\theta}{\mu} \mc P_{T_{\mb q}} \bb S^{n-1}. 
\end{align*}
Also since $\wt{\eta}_f \le 4\eta_f$ in Lemma~\ref{lem:alg_approx_bd2} and Lemma~\ref{lem:alg_strcvx_func}, there are only multiplicative constant change to the various quantities. We conclude that 
\begin{align}
\wt{\dI} = c_3 \dI
\end{align}
provided that 
\begin{align}
\Delta \le \frac{c_4 c_\star^2 \theta^2 \mu  }{n^{5/2} \log^{5/2} (np) }. 
\end{align}
\item Lemma~\ref{lem:alg_gradient_lb} is generic and nothing changes. 
\item Lemma~\ref{lem:TR-step}: $\wt{L}_H \le 27L_H/8$. 
\item Proposition~\ref{lem:TRM-lemma-iv}: All the quantities involved in determining $\Delta$, $m_H$, $M_H$, and $L_H$, $\beta_{\grad}$ are modified by at most constant multiplicative factors and changed to their respective tilde version, so we conclude that the TRM algorithm always takes unconstrained $\rI$ step after taking one, provided that
\begin{align}
\Delta \le \frac{c_5 c_\star^3 \theta^3 \mu}{n^{7/2} \log^{7/2}\paren{np}}. 
\end{align}
\item Lemma~\ref{lem:TR-grad-opt}:is generic and nothing changes. 

\item Proposition~\ref{lem:TRM-lemma-v}: Again $m_H$, $M_H$, $L_H$ are changed to $\wt{m_H}$, $\wt{M_H}$, and $\wt{L_H}$, respectively, differing by at most constant multiplicative factors. So we conclude for any integer $k' \ge 1$, 
\begin{align}
	\norm{\mb q^{(k_0 + k')} - \mb q_\star }{} \; \leq\;  \frac{c_6 c_\star \theta \mu}{n^{3/2} \log^{3/2} \paren{np}} 2^{- 2^{k'}},
\end{align}
provided 
\begin{align}
	\Delta \leq \frac{c_7 c_\star^2 \theta^2 \mu}{n^{5/2} \log^{5/2} (np) }. 
\end{align}
\end{itemize}
The final proof to Theorem~\ref{thm:geometry_comp} is almost identical to that of Theorem~\ref{thm:geometry_orth}, except that $\dI$, $\dII$, and $\dIII$ are changed to $\wt{\dI}$, $\wt{\dII}$, and $\wt{\dIII}$ as defined above, respectively. The final iteration complexity to each an $\eps$-near solution is hence 
\begin{align*}
\#\text{Iter} 
 & \le \brac{\min\set{\frac{c_{8} c_{\star}^3 \theta^3 \mu^4}{n^6 \log^3\paren{np}}, \frac{c_{9} c_\star^2 \theta^2 }{n}\Delta^2 }}^{-1}\paren{f\paren{\mb q^{(0)}} - f\paren{\mb q_\star}} \;+\; \log\log \paren{\frac{c_{10} c_\star \theta \mu}{\eps n^{3/2} \log^{3/2}\paren{np}}}. 
\end{align*}
Hence overall the qualitative behavior of the algorithm is not changed, as compared to that for the orthogonal case.

%% file: sec/main_result.tex
\section{Complete Algorithm Pipeline and Main Results} \label{sec:main_result} 

For orthogonal dictionaries, from Theorem \ref{thm:geometry_orth} (and Corollary II.2 in~\cite{sun2015complete_a}), we know that all the minimizers $\wh{\mb q}_\star$ are $O(\mu)$ away from their respective nearest ``target'' $\mb q_\star$, with $\mb q_\star^* \wh{\mb Y} = \alpha \mb e_i^* \mb X_0$ for a certain $\alpha \ne 0$ and $i \in [n]$; in Theorem~\ref{thm:trm_orth}, we have shown that w.h.p.\ the Riemannian TRM algorithm produces a solution $\widehat{\mb q}\in \bb S^{n-1}$ that is $\eps$ away to one of the minimizers, say $\wh{\mb q}_\star$. Thus, the $\wh{\mb q}$ returned by the TRM algorithm is $O(\eps + \mu)$ away from $\mb q_\star$. For exact recovery, we use a simple linear programming rounding procedure, which guarantees to produce the target $\mb q_\star$. We then use deflation to sequentially recover other rows of $\mb X_0$. Overall, w.h.p.\ both the dictionary $\mb A_0$ and sparse coefficient $\mb X_0$ are exactly recovered up to sign permutation, when $\theta \in \Omega(1)$, for orthogonal dictionaries. We summarize relevant technical lemmas and main results in Section~\ref{sec:main_orth}. The same procedure can be used to recover complete dictionaries, though the analysis is slightly more complicated; we present the results in Section~\ref{sec:main_comp}. Our overall algorithmic pipeline for recovering orthogonal dictionaries is sketched as follows. 
\begin{leftbar} 
\begin{enumerate}
\item \textbf{Estimating one row of $\mb X_0$ by the Riemannian TRM algorithm.} By Theorem \ref{thm:geometry_orth} (resp. Theorem~\ref{thm:geometry_comp}) and Theorem \ref{thm:trm_orth} (resp. Theorem~\ref{thm:trm_comp}), starting from any $\mb q \in \bb S^{n-1}$, when the relevant parameters are set appropriately (say as $\mu_\star$ and $\Delta_\star$), w.h.p., our Riemannian TRM algorithm finds a local minimizer $\widehat{\mb q}$, with $\mb q_\star$ the nearest target that exactly recovers a row of $\mb X_0$ and $\norm{\wh{\mb q} - \mb q_\star}{} \in O(\mu)$ (by setting the target accuracy of the TRM as, say,  $\eps = \mu$).

\item \textbf{Recovering one row of $\mb X_0$ by rounding.} To obtain the target solution $\mb q_\star$ and hence recover (up to scale) one row of $\mb X_0$, we solve the following linear program:
\begin{align}\label{eqn:LP_rounding}
	\mini_{\mb q} \norm{\mb q^*\wh{\mb Y}}{1},\quad \st \quad \innerprod{\mb r}{\mb q} = 1, 
\end{align}
with $\mb r = \widehat{\mb q}$. We show in Lemma~\ref{lem:alg_rounding_orth} (resp. Lemma~\ref{lem:alg_rounding_comp}) that when $\innerprod{\wh{\mb q}}{\mb q_\star}$ is sufficiently large, implied by $\mu$ being sufficiently small, w.h.p. the minimizer of \eqref{eqn:LP_rounding} is exactly $\mb q_\star$, and hence one row of $\mb X_0$ is recovered by $\mb q_\star^* \wh{\mb Y}$.

\item \textbf{Recovering all rows of $\mb X_0$ by deflation.} Once $\ell$ rows of $\mb X_0$ ($1 \le \ell \le n-2$) have been recovered, say, by unit vectors $\mb q_\star^1, \dots, \mb q_\star^\ell$, one takes an orthonormal basis $\mb U$ for $[\mathrm{span}\paren{\mb q_\star^1, \dots, \mb q_\star^\ell}]^\perp$, and minimizes the new function $h(\mb z) \doteq f(\mb U \mb z; \wh{\mb Y})$ on the sphere $\bb S^{n-\ell-1}$ with the Riemannian TRM algorithm (though conservative, one can again set parameters as $\mu_\star$, $\Delta_\star$, as in Step $1$) to produce a $\wh{\mb z}$. Another row of $\mb X_0$ is then recovered via the LP rounding~\eqref{eqn:LP_rounding} with input $\mb r = \mb U \wh{\mb z}$ (to produce $\mb q_\star^{\ell+1}$). Finally, by repeating the procedure until depletion, one can recover all the rows of $\mb X_0$.

\item \textbf{Reconstructing the dictionary $\mb A_0$.} By solving the linear system $\mb Y = \mb A\mb X_0$, one can obtain the dictionary $\mb A_0 = \mb Y \mb X_0^* \paren{\mb X_0 \mb X_0^*}^{-1}$.  
\end{enumerate}
\end{leftbar}

\subsection{Recovering Orthogonal Dictionaries} \label{sec:main_orth}
\begin{theorem}[Main theorem - recovering orthogonal dictionaries]\label{thm:main_orth}
Assume the dictionary $\mb A_0$ is orthogonal and we take $\wh{\mb Y} = \mb Y$. Suppose $\theta \in \paren{0,1/3}$, $\mu_\star <  c_a\min\Brac{\theta n^{-1}, n^{-5/4}}$, and $p \ge Cn^3 \log \frac{n}{\mu_\star \theta} /\paren{\mu_\star^2\theta^2}$. The above algorithmic pipeline with parameter setting
\begin{align}
\Delta_\star = \frac{c_b c_\star^3 \theta^3 \mu_\star^2}{n^{7/2}\log^{7/2}\paren{np}}, 
\end{align}
recovers the dictionary $\mb A_0$ and $\mb X_0$ in polynomial time, with failure probability bounded by $c_c p^{-6}$. Here $c_\star$ is as defined in Theorem~\ref{thm:geometry_orth}, and $c_a$ through $c_c$, and $C$ are all positive constants. 
\end{theorem}

Towards a proof of the above theorem, it remains to be shown the correctness of the rounding and deflation procedures.

\paragraph{Proof of LP rounding.} The following lemma shows w.h.p.\ the rounding will return the desired $\mb q_\star$, provided the estimated $\wh{\mb q}$ is already near to it. 
\begin{lemma} [LP rounding - orthogonal dictionary]\label{lem:alg_rounding_orth}
For any $\theta \in \paren{0,1/3}$, whenever $p \ge Cn^2\log(n/\theta)/\theta$, with probability at least 
$
	1-cp^{-6},  
$
the rounding procedure~\eqref{eqn:LP_rounding} returns $\mb q_\star$ for any input vector $\mb r$ that satisfies
\begin{align*}
	\innerprod{\mb r}{\mb q_\star} \ge 249/250. 
\end{align*}
Here $C, c$ are both positive constants. 
\end{lemma}
\begin{proof}
See Page~\pageref{proof:lem_alg_rounding_orth} under Section~\ref{sec:proof_main}. 
\end{proof}
Since $\innerprod{\wh{\mb q}}{\mb q_\star} = 1-\|\wh{\mb q} - \mb q_\star\|^2/2$, and $\norm{\wh{\mb q} - \mb q_\star}{} \in O(\mu)$, it is sufficient when $\mu$ is smaller than some small constant. 

\paragraph{Proof sketch of deflation.} We show the deflation works by induction. To understand the deflation procedure, it is important to keep in mind that the ``target'' solutions $\Brac{\mb q_\star^i}_{i=1}^n$ are orthogonal to each other. W.l.o.g., suppose we have found the first $\ell$ unit vectors $\mb q_\star^1, \dots, \mb q_\star^\ell$ which recover the first $\ell$ rows of $\mb X_0$. Correspondingly, we partition the target dictionary $\mb A_0$  and $\mb X_0$ as
\begin{align}\label{eqn:matrices-partition}
	\mb A_0 = [\mb V, \mb V^\perp],\quad \mb X_0 = \brac{\begin{smallmatrix}
\mb X_0^{[\ell]} \\
\mb X_0^{[n-\ell]}
\end{smallmatrix} },
\end{align}
where $\mb V \in \R^{n \times \ell}$, and $\mb X_0^{[\ell]}\in \bb R^{\ell\times n} $ denotes the submatrix with the first $\ell$ rows of $\mb X_0$. Let us define a function: $f_{n -\ell}^{\downarrow}: \R^{n-\ell} \mapsto \R$ by
\begin{align}\label{eqn:func-(n-l)}
f_{n -\ell}^{\downarrow}(\mb z; \mb W) \doteq \frac{1}{p}\sum_{k=1}^p h_{\mu}(\mb z^* \mb w_k), 
\end{align}
for any matrix $\mb W \in \R^{(n -\ell) \times p}$. Then by \eqref{eq:main_obj}, our objective function is equivalent to  
\begin{align*}
	h(\mb z) = f(\mb U \mb z; \mb A_0 \mb X_0) = f_{n-\ell}^{\downarrow}(\mb z;\mb U^* \mb A_0\mb X_0) = f_{n-\ell}^{\downarrow}(\mb z; \mb U^*\mb V\mb X_0^{[\ell]} + \mb U^*\mb V^\perp \mb X_0^{[n-\ell]}).
\end{align*}
Since the columns of the orthogonal matrix $\mb U\in \bb R^{n\times (n-\ell)}$ forms the orthogonal complement of $\text{span}\paren{\mb q_\star^1,\cdots,\mb q_\star^\ell}$, it is obvious that $\mb U^*\mb V=\mb 0$. Therefore, we obtain
\begin{align*}
h(\mb z) = f_{n-\ell}^{\downarrow}(\mb z; \mb U^*\mb V^\perp \mb X_0^{[n-\ell]}).
\end{align*}
Since $\mb U^* \mb V^\perp$ is orthogonal and $\mb X_0^{[n-\ell]} \sim_{i.i.d.} \mathrm{BG}(\theta)$, this is another instance of orthogonal dictionary learning problem with reduced dimension. If we keep the parameter settings $\mu_\star$ and $\Delta_\star$ as Theorem \ref{thm:main_orth}, the conditions of Theorem~\ref{thm:geometry_orth} and Theorem~\ref{thm:trm_orth} for all cases with reduced dimensions are still valid. So w.h.p., the TRM algorithm returns a $\wh{\mb z}$ such that $\norm{\wh{\mb z} - \mb z_\star}{} \in O(\mu_\star)$ where $\mb z_\star$ is a ``target'' solution that recovers a row of $\mb X_0$: 
\begin{align*}
\mb z_\star^*\mb U^* \mb V^\perp \mb X_0^{[n-\ell]} = \mb z_\star^*\mb U^* \mb A_0\mb X_0 = \alpha \mb e_i^* \mb X_0,\quad \text{for some }i \not \in [\ell].
\end{align*}
So pulling everything back in the original space, the effective target is $\mb q_\star^{\ell+1} \doteq \mb U \mb z_\star$, and $\mb U \wh{\mb z}$ is our estimation obtained from the TRM algorithm. Moreover, 
\begin{align*}
\norm{\mb U \wh{\mb z} - \mb U \mb z_\star}{} = \norm{\wh{\mb z} - \mb z_\star}{} \in O(\mu_\star). 
\end{align*}
Thus, by Lemma~\ref{lem:alg_rounding_orth}, one successfully recovers $\mb U \mb z_\star$ from $\mb U \wh{\mb z}$ w.h.p. when $\mu_\star$ is smaller than a constant. The overall failure probability can be obtained via a simple union bound and simplification of the exponential tails with inverse polynomials in $p$. 

\subsection{Recovering Complete Dictionaries} \label{sec:main_comp}
By working with the preconditioned data samples $\wh{\mb Y} = \overline{\mb Y} \doteq \sqrt{\theta p}\paren{\mb Y\mb Y^*}^{-1/2} \mb Y$,\footnote{In practice, the parameter $\theta$ might not be know beforehand. However, because it only scales the problem, it does not affect the overall qualitative aspect of results.} we can use the same procedure as described above to recover complete dictionaries.

\begin{theorem}[Main theorem - recovering complete dictionaries]\label{thm:main_comp}
Assume the dictionary $\mb A_0$ is complete with a condition number $\kappa\paren{\mb A_0}$ and we take $\wh{\mb Y} = \ol{\mb Y}$. Suppose $\theta \in \paren{0,1/3}$, $\mu_\star <  c_a\min\Brac{\theta n^{-1},n^{-5/4}}$, and $p \ge \frac{C}{c_\star^2 \theta^2} \max\set{\frac{n^4}{\mu^4}, \frac{n^5}{\mu^2}} \kappa^8\paren{\mb A_0} \log^4\paren{\frac{\kappa\paren{\mb A_0} n}{\mu \theta}}$. The algorithmic pipeline with parameter setting
\begin{align}
\Delta_\star =  \frac{c_d c_\star^3 \theta^3 \mu_\star^2}{n^{7/2}\log^{7/2}\paren{np}}
\end{align}
recovers the dictionary $\mb A_0$ and $\mb X_0$ in polynomial time, with failure probability bounded by $c_b p^{-6}$. Here $c_\star$ is as defined in Theorem~\ref{thm:geometry_orth}, and $c_a, c_b$ are both positive constants. 
\end{theorem}

Similar to the orthogonal case, we need to show the correctness of the rounding and deflation procedures so that the theorem above holds.

\paragraph{Proof of LP rounding}
The result of the LP rounding is only slightly different from that of the orthogonal case in Lemma \ref{lem:alg_rounding_orth}, so is the proof.
\begin{lemma} [LP rounding - complete dictionary]\label{lem:alg_rounding_comp}
For any $\theta \in \paren{0,1/3}$, whenever 
\begin{align*}
p \ge \frac{C}{c_\star^2 \theta} \max\set{\frac{n^4}{\mu^4}, \frac{n^5}{\mu^2}} \kappa^8\paren{\mb A_0} \log^4\paren{\frac{\kappa\paren{\mb A_0} n}{\mu \theta}},
\end{align*} 
with probability at least 
$
	1- cp^{-6},  
$
the rounding procedure~\eqref{eqn:LP_rounding} returns $\mb q_\star$ for any input vector $\mb r$ that satisfies
\begin{align*}
	\innerprod{\mb r}{\mb q_\star} \ge 249/250.  
\end{align*}
Here $C, c$ are both positive constants. 
\end{lemma}
\begin{proof}
See Page~\pageref{proof:lem_alg_rounding_comp} under Section~\ref{sec:proof_main}. 
\end{proof}

\paragraph{Proof sketch of deflation.} We use a similar induction argument to show the deflation works. Compared to the orthogonal case, the tricky part here is that the target vectors $\Brac{\mb q_\star^i}_{i=1}^n$ are not necessarily orthogonal to each other, but they are almost so. W.l.o.g., let us again assume that $\mb q_\star^1, \dots, \mb q_\star^\ell$ recover the first $\ell$ rows of $\mb X_0$, and similarly partition the matrix $\mb X_0$ as in \eqref{eqn:matrices-partition}.

By Lemma~\ref{lem:pert_key_mag} and~\eqref{eq:pert_upper_bound}, we can write $\ol{\mb Y} = (\mb Q + \mb \Xi) \mb X_0$ for some orthogonal matrix $\mb Q$ and small perturbation $\mb \Xi$ with $\norm{\mb \Xi}{} \le \delta < 1/10$ for some large $p$ as usual. Similar to the orthogonal case, we have
\begin{align*}
	h(\mb z) = f(\mb U \mb z; (\mb Q + \mb \Xi) \mb X_0) = f_{n -\ell}^{\downarrow}(\mb z; \mb U^* (\mb Q + \mb \Xi) \mb X_0),
\end{align*}
where $f_{n -\ell}^{\downarrow}$ is defined the same as in \eqref{eqn:func-(n-l)}. Next, we show that the matrix $\mb U^* (\mb Q + \mb \Xi) \mb X_0$ can be decomposed as $\mb U^*\mb V\mb X_0^{[n-\ell]} + \mb \Delta$, where $\mb V\in \bb R^{(n-\ell)\times n }$ is orthogonal and $\mb \Delta$ is a small perturbation matrix. More specifically, we show that

\begin{lemma}\label{lem:deflation-bound}
Suppose the matrices $\mb U\in \bb R^{n\times (n-\ell)}$, $\mb Q \in \bb R^{n\times n}$ are orthogonal as defined above, $\mb \Xi$ is a perturbation matrix with $\norm{\mb \Xi}{}\leq 1/20$, then
	\begin{align}
		\mb U^*\paren{\mb Q+\mb \Xi}\mb X_0 = \mb U^* \mb V\mb X_0^{[n-\ell]} + \mb \Delta,
	\end{align}
	where $\mb V\in \bb R^{n\times (n-\ell)} $ is an orthogonal matrix spanning the same subspace as that of $\mb U$, and the norms of $\mb \Delta$ is bounded by
	\begin{align}
		\norm{\mb \Delta}{\ell^1\rightarrow \ell^2 } \leq 16\sqrt{n} \norm{\mb \Xi}{} \norm{\mb X_0}{\infty}, \quad \norm{\mb \Delta}{} \leq 16 \norm{\mb \Xi}{} \norm{\mb X_0}{},
	\end{align}
	where $\norm{\mb W}{\ell^1\rightarrow \ell^2} = \sup_{\norm{\mb z}{1}=1} \norm{\mb W\mb z}{} = \max_k \norm {\mb w_k}{} $ denotes the max column $\ell^2$-norm of a matrix $\mb W$. 
\end{lemma}
\begin{proof}
See Page~\pageref{proof:lem_deflation-bound} under Section~\ref{sec:proof_main}. 
\end{proof}

Since $\mb U\mb V$ is orthogonal and $\mb X_0^{[n-\ell]}\sim_{i.i.d.} \text{BG}(\theta)$, we come into another instance of perturbed dictionary learning problem with a reduced dimension
\begin{align*}
	h(\mb z) = f_{n-\ell}^{\downarrow}\paren{\mb z; \mb U^*\mb V \mb X_0^{[n-\ell]} + \mb \Delta}.
\end{align*}
Since our perturbation analysis in proving Theorem~\ref{thm:geometry_comp} and Theorem~\ref{thm:trm_comp} solely relies on the fact that $\norm{\mb \Delta}{\ell^1 \rightarrow \ell^2} \leq C \norm{\mb \Xi}{}\sqrt{n}\norm{\mb X_0}{\infty}$, it is enough to make $p$ large enough so that the theorems are still applicable for the reduced version $f_{n-\ell}^{\downarrow}(\mb z; \mb U^*\mb V \mb X_0^{[n-\ell]} + \mb \Delta)$. Thus, by invoking Theorem~\ref{thm:geometry_comp} and Theorem~\ref{thm:trm_comp}, the TRM algorithm provably returns one $\wh{\mb z}$ such that $\wh{\mb z}$ is near to a perturbed optimal $\wh{\mb z}_\star$ with
\begin{align}\label{eqn:comp-solution}
	\wh{\mb z}_\star^* \mb U^*\mb V\mb X_0^{[n-\ell]} = \mb z_\star^* \mb U^*\mb V\mb X_0^{[n-\ell]} +\mb z_\star^*\mb \Delta =\alpha \mb e_i^* \mb X_0,\quad \text{for some } i\not \in [\ell],
\end{align}
where $\mb z_\star$ with $\norm{\mb z_\star}{}=1$ is the exact solution. More specifically, Corollary II.4 in~\cite{sun2015complete_a} implies that 
\begin{align*}
	\norm{\wh{\mb z} - \wh{\mb z}_\star }{} \leq \sqrt{2}\mu_\star /7.
\end{align*}
Next, we show that $\wh{\mb z}$ is also very near to the exact solution $\mb z_\star$. Indeed, the identity \eqref{eqn:comp-solution} suggests
\begin{align}
	&\paren{\wh{\mb z}_\star -\mb z_\star }^*\mb U^*\mb V\mb X_0^{[n-\ell]} = \mb z_\star^* \mb \Delta \nonumber\\
	\Longrightarrow\;& \wh{\mb z}_\star - \mb z_\star = \brac{(\mb X_0^{[n-\ell]})^* \mb V^*\mb U }^\dagger \mb \Delta^*\mb z_\star = \mb U^*\mb V \brac{(\mb X_0^{[n-\ell]})^* }^\dagger \mb \Delta^*\mb z_\star \label{eqn:comp-z-distance}
\end{align}
where $\mb W^\dagger = (\mb W^*\mb W)^{-1}\mb W^*$ denotes the pseudo inverse of a matrix $\mb W$ with full column rank. Hence, by \eqref{eqn:comp-z-distance} we can bound the distance between $\wh{\mb z}_\star$ and $\mb z_\star$ by
\begin{align*}
	\norm{\wh{\mb z}_\star - \mb z_\star }{}  \leq \norm{\brac{(\mb X_0^{[n-\ell]})^* }^\dagger }{} \norm{\mb \Delta}{} \leq \sigma_{\min}^{-1}(\mb X_0^{[n-\ell]} ) \norm{\mb \Delta}{}
\end{align*}
By Lemma~\ref{lem:bg_identity_diff}, when $p \ge \Omega(n^2 \log n)$, w.h.p., 
\begin{align*}
	 \theta p/2\leq \sigma_{\min}(\mb X_0^{[n-\ell]} (\mb X_0^{[n-\ell]})^* ) \leq \norm{\mb X_0^{[n-\ell]}(\mb X_0^{[n-\ell]})^* }{} \leq \norm{\mb X_0\mb X_0^*}{}\leq 3\theta p/2.
\end{align*}
Hence, combined with Lemma \ref{lem:deflation-bound}, we obtain
\begin{align*}
	\sigma_{\min}^{-1}(\mb X_0^{[n-\ell]}) \leq \sqrt{\frac{2}{\theta p}},\quad  \norm{\mb \Delta }{} \leq  28\sqrt{ \theta p} \norm{\mb \Xi}{}/\sqrt{2},
\end{align*}
which implies that $\norm{\wh{\mb z}_\star -\mb z_\star }{}\leq 28 \norm{\mb \Xi}{}$. Thus, combining the results above, we obtain
\begin{align*}
	\norm{\wh{\mb z} - \mb z_\star}{} \le \norm{\wh{\mb z} - \wh{\mb z}_\star}{} + \norm{\wh{\mb z}_\star - \mb z_\star}{} \le \sqrt{2}\mu_\star/7 + 28\norm{\mb \Xi}{}.
\end{align*}
Lemma~\ref{lem:pert_key_mag}, and in particular~\eqref{eq:pert_upper_bound}, for our choice of $p$ as in Theorem~\ref{thm:geometry_comp}, $\norm{\mb \Xi}{} \le c\mu_\star^2 n^{-3/2}$, where $c$ can be made smaller by making the constant in $p$ larger. For $\mu_\star$ sufficiently small, we conclude that 
\begin{align*}
\norm{\mb U \wh{\mb z} - \mb U\mb z_\star}{} = \norm{\wh{\mb z} - \mb z_\star}{} \le 2\mu_\star/7. 
\end{align*} 
In words, the TRM algorithm returns a $\wh{\mb z}$ such that $\mb U \wh{\mb z}$ is very near to one of the unit vectors $\Brac{\mb q_\star^i}_{i=1}^n$, such that $(\mb q_\star^i)^* \ol{\mb Y} = \alpha \mb e_i^*\mb X_0$ for some $\alpha \ne 0$. For $\mu_\star$ smaller than a fixed constant, one will have 
\begin{align*}
\innerprod{\mb U \wh{\mb z}}{\mb q_\star^i} \ge 249/250, 
\end{align*}
and hence by Lemma~\ref{lem:alg_rounding_comp}, the LP rounding exactly returns the optimal solution $\mb q_\star^i$ upon the input $\mb U \wh{\mb z}$. 

The proof sketch above explains why the recursive TRM plus rounding works. The overall failure probability can be obtained via a simple union bound and simplifications of the exponential tails with inverse polynomials in $p$.

%% file: sec/exp.tex
\section{Simulations} \label{sec:exp}

\js{
\subsection{Practical TRM Implementation}

Fixing a small step size and solving the trust-region subproblem exactly eases the analysis, but also renders the TRM algorithm impractical. In practice, the trust-region subproblem is never exactly solved, and the trust-region step size is adjusted to the local geometry, say by backtracking. It is possible to modify our algorithmic analysis to account for inexact subproblem solvers and adaptive step size; for sake of brevity, we do not pursue it here. Recent theoretical results on the practical version include~\cite{cartis2012complexity,boumal2016global}. 

Here we describe a practical implementation based on the \emph{Manopt} toolbox~\cite{boumal2014manopt}\footnote{Available online: \url{http://www.manopt.org}. }. Manopt is a user-friendly Matlab toolbox that implements several sophisticated solvers for tackling optimization problems over Riemannian manifolds. The most developed solver is based on the TRM. This solver uses the truncated conjugate gradient (tCG; see, e.g., Section 7.5.4 of~\cite{conn2000trust}) method to (approximately) solve the trust-region subproblem (vs.\ the exact solver in our analysis). It also dynamically adjusts the step size using backtracking. However, the original implementation (Manopt 2.0) is not adequate for our purposes. Their tCG solver uses the gradient as the initial search direction, which does not ensure that the TRM solver can escape from saddle points~\cite{absil2007trust,absil2009}. We modify the tCG solver, such that when the current gradient is small and there is a negative curvature direction (i.e., the current point is near a saddle point or a local maximizer of $f(\mb q)$), the tCG solver explicitly uses the negative curvature direction\footnote{...adjusted in sign to ensure positive correlation with the gradient -- if it does not vanish.} as the initial search direction. This modification ensures the TRM solver always escape from saddle points/local maximizers with negative directional curvature. Hence, the modified TRM algorithm based on Manopt is expected to have the same qualitative behavior as the idealized version we analyzed above, with better scalability. We will perform our numerical simulations using the modified TRM algorithm whenever necessary. Algorithm 3 together with Lemmas 9 and 10 and the surrounding discussion in the very recent work~\cite{boumal2016global} provides a detailed description of this practical version. 
}

\subsection{Simulated Data}

To corroborate our theory, we experiment with dictionary recovery on simulated data.\footnote{The code is available online: \url{https://github.com/sunju/dl_focm}} For simplicity, we focus on recovering orthogonal dictionaries and we declare success once a single row of the coefficient matrix is recovered.

Since the problem is invariant to rotations, w.l.o.g.\ we set the dictionary as $\mb A_0 = \mb I \in \R^{n \times n}$. For any fixed sparsity $k$, each column of the coefficient matrix $\mb X_0 \in \R^{n \times p}$ has exactly $k$ nonzero entries, chosen uniformly random from $\binom{[n]}{k}$. These nonzero entries are i.i.d. standard normals. This is slightly different from the Bernoulli-Gaussian model we assumed for analysis. For $n$ reasonably large, these two models have similar behaviors. For our sparsity surrogate, we fix the smoothing parameter as $\mu = 10^{-2}$. Because the target points are the signed basis vector $\pm \mb e_i$'s (to recover rows of $\mb X_0$), for a solution $\wh{\mb q}$ returned by the TRM algorithm, we define the reconstruction error (RE) to be
\begin{align}
\mathtt{RE} =  \min_{i \in [n]} \paren{\norm{\wh{\mb q} - \mb e_i}{},\norm{\wh{\mb q} + \mb e_i}{} }. 
\end{align}
One trial is determined to be a success once $\mathtt{RE} \le \mu$, with the idea that this indicates $\wh{\mb q}$ is already very near the target and the target can likely be recovered via the LP rounding we described (which we do not implement here). 

\js{
We consider two settings: (1) fix $p = 5n^2 \log n$ and vary the dimension $n$ and sparsity $k$; (2) fix the sparsity level as $\lceil 0.2 \cdot n \rceil$ and vary the dimension $n$ and number of samples $p$. For each pair of $(k,n)$ for (1), and each pair of $(p, n)$ for (2), we repeat the simulations independently for $T=5$ times. 
\begin{figure}
\centering
\includegraphics[width=0.4\textwidth]{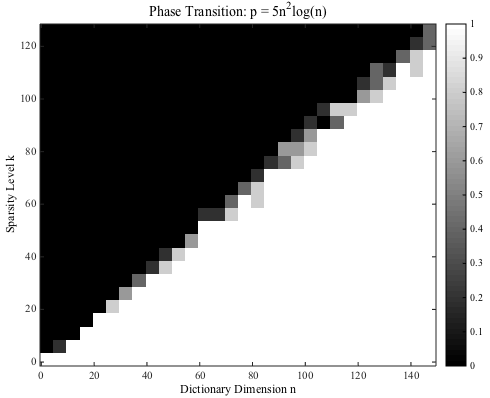}
\includegraphics[width=0.4\textwidth]{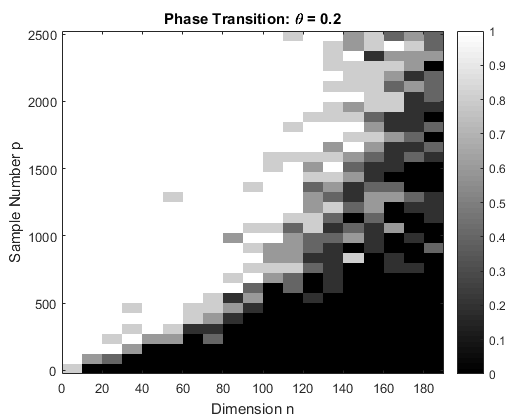}
\caption{Phase transition for recovering a single sparse vector. \textbf{Left:} We fix $p = 5n^2 \log n$ and vary the dimension $n$ and sparsity level $k$; \textbf{Right:} We fix the sparsity level as $\lceil 0.2 \cdot n\rceil$ and vary the dimension $n$ and number of samples $p$. For each configuration, the experiment is independently repeated for five times. White indicates success, and black indicates failure. } 
\label{fig:exp_simulated}
\end{figure}
Fig.~\ref{fig:exp_simulated} shows the phase transition for the two settings. It seems that our TRM algorithm can work well into the linear region whenever $p \in O(n^2 \log n)$ (Fig.~\ref{fig:exp_simulated}-Left), but $p$ should have order greater than $\Omega(n)$ (Fig.~\ref{fig:exp_simulated}-Right). The sample complexity from our theory is significantly suboptimal compared to this. 
}

\js{
\subsection{Image Data Again}

Our algorithmic framework has been derived based on the BG model on the coefficients. Real data may not admit sparse representations w.r.t. complete dictionaries, or even so, the coefficients may not obey the BG model. In this experiment, we explore how our algorithm performs in learning complete dictionaries for image patches, emulating our motivational experiment in the companion paper~\cite{sun2015complete_a} (Section I.B). Thanks to research on image compression, we know patches of natural images tend to admit sparse representation, even w.r.t. simple orthogonal bases, such as Fourier basis or wavelets.

\begin{figure}
\centering
\begin{minipage}{.3\textwidth}
    \centering
    \includegraphics[width = 0.96\textwidth]{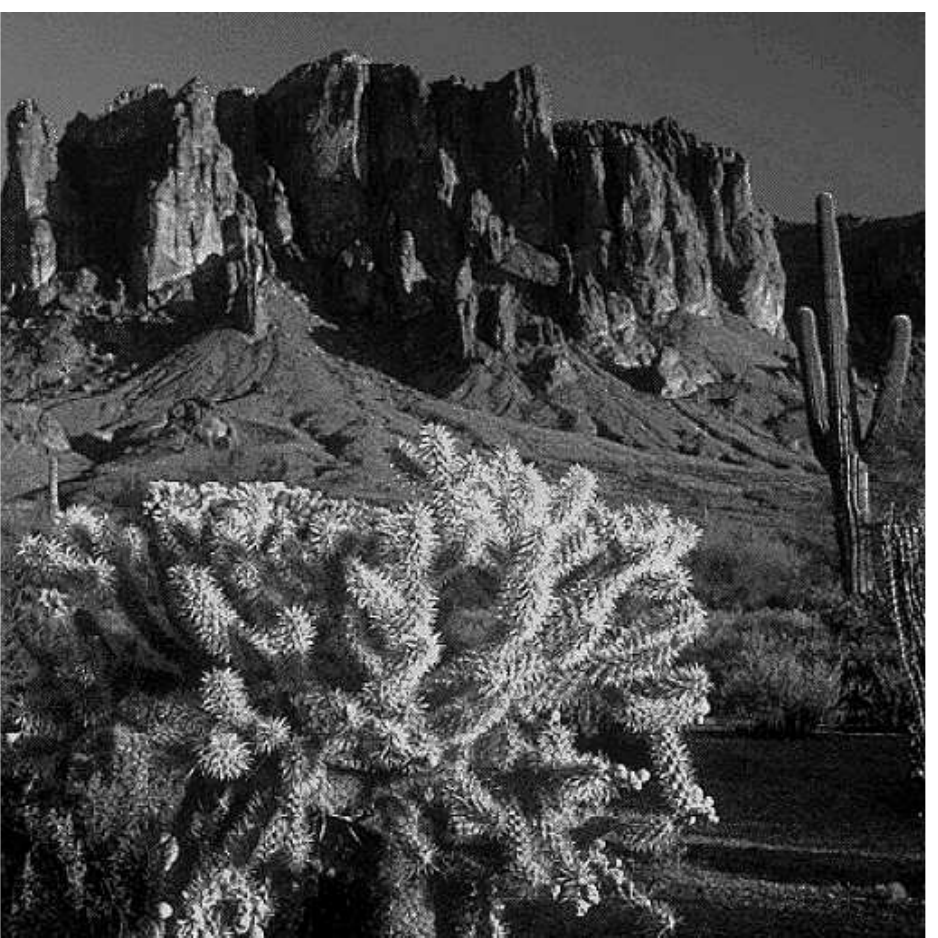}\\
    \vspace{0.05in}
    \includegraphics[width = 0.95\textwidth]{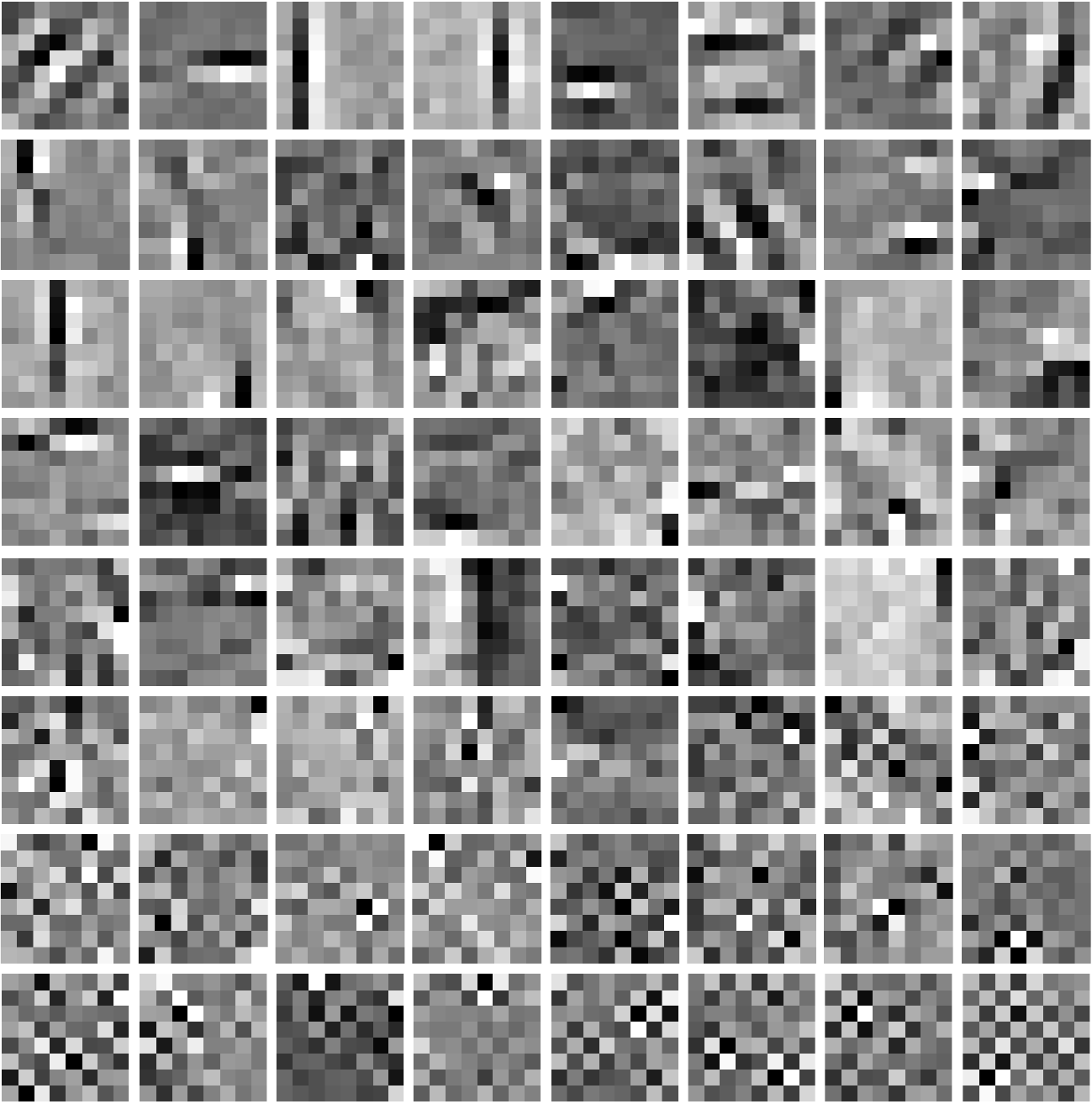}\\
    \vspace{0.05in}
    \includegraphics[width = \textwidth]{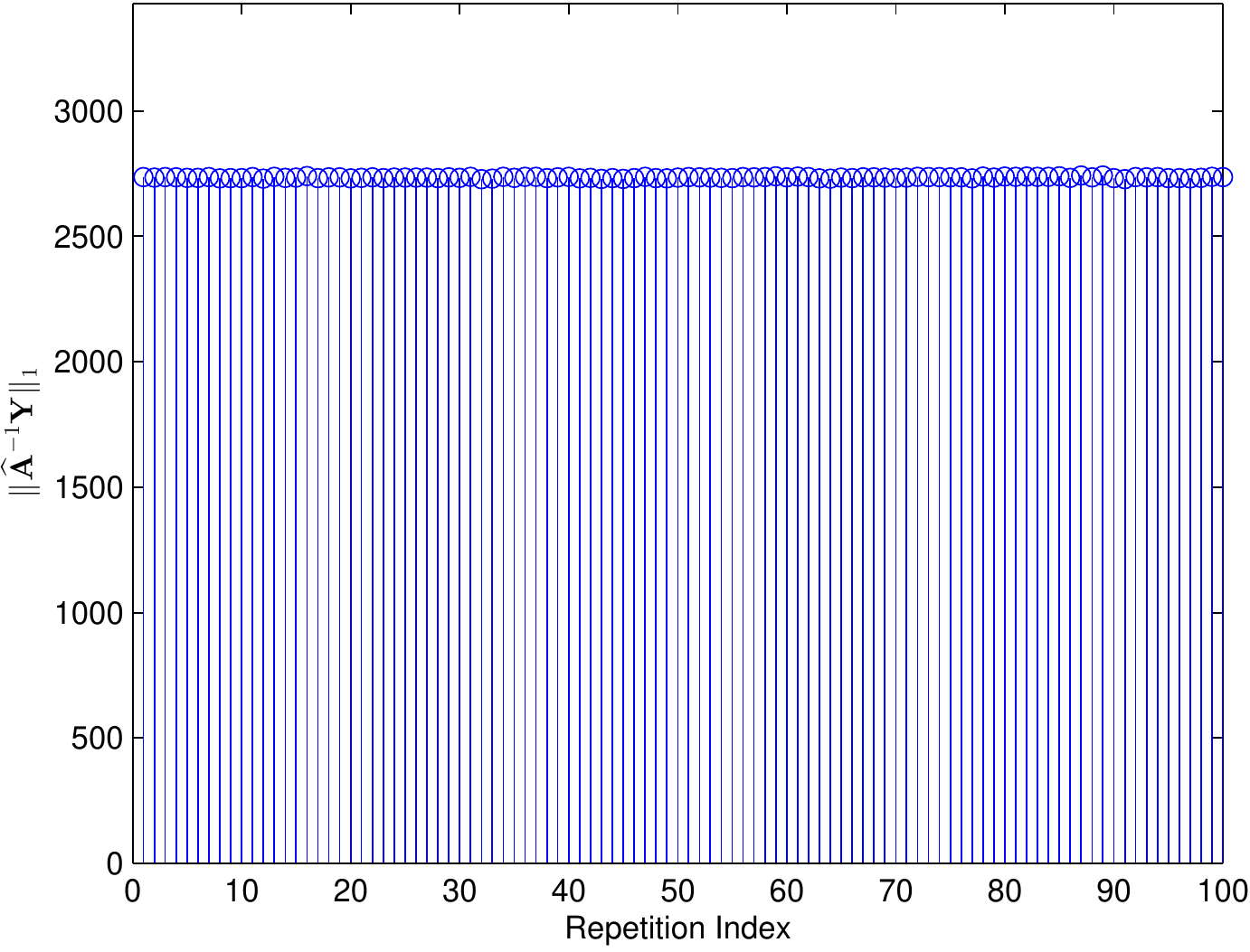}
\end{minipage}%
\begin{minipage}{.3\textwidth}
    \centering
    \includegraphics[width = 0.96\textwidth]{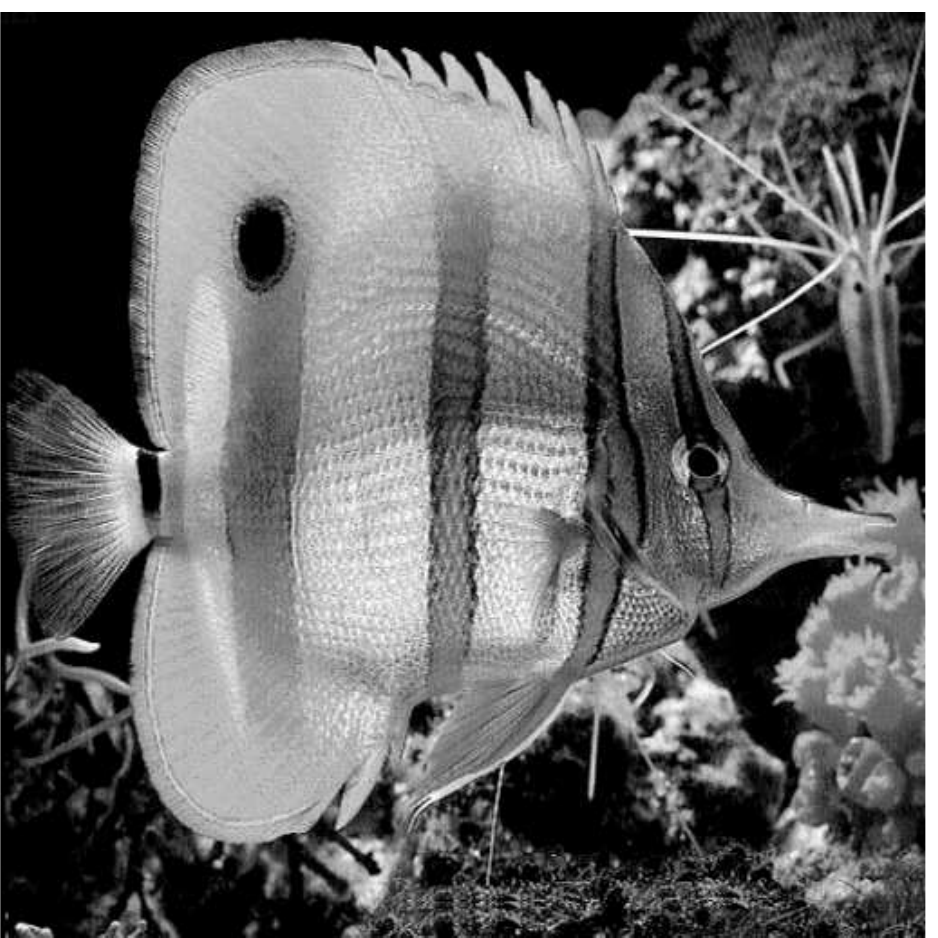}\\
    \vspace{0.05in}
    \includegraphics[width = 0.95\textwidth]{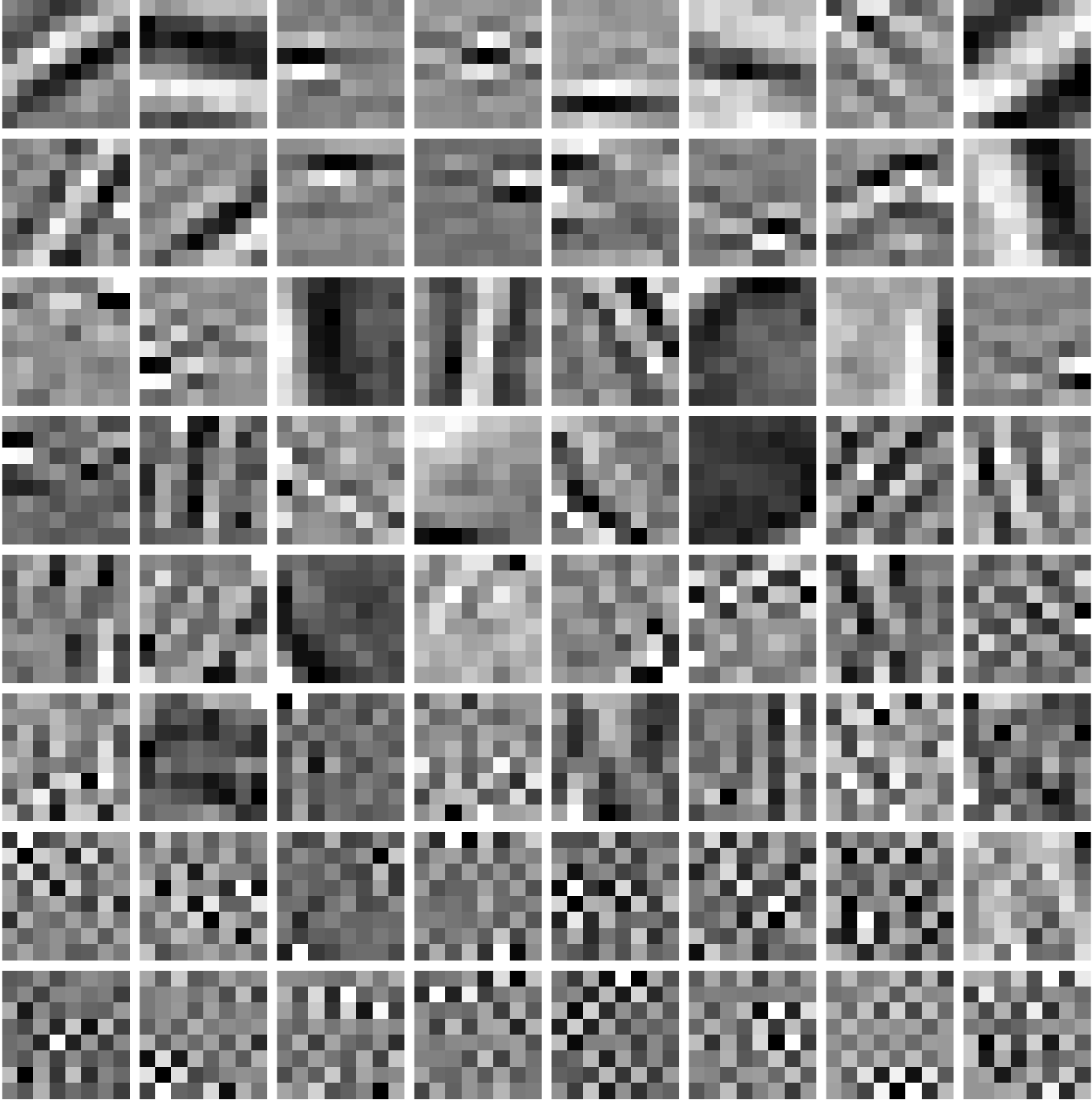}\\
    \vspace{0.05in}
    \includegraphics[width = \textwidth]{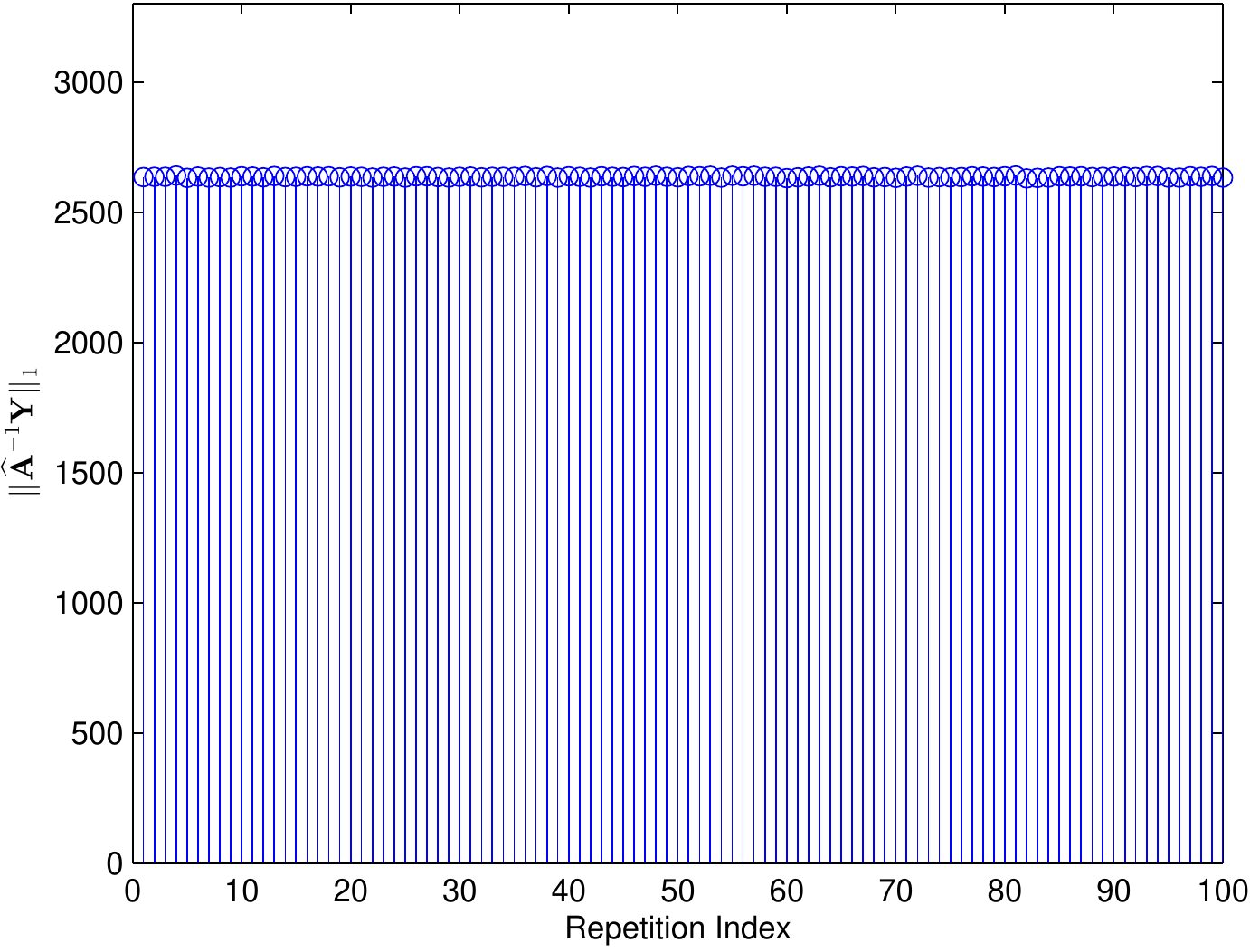}
\end{minipage}%
\begin{minipage}{.3\textwidth}
    \centering
    \includegraphics[width = 0.96\textwidth]{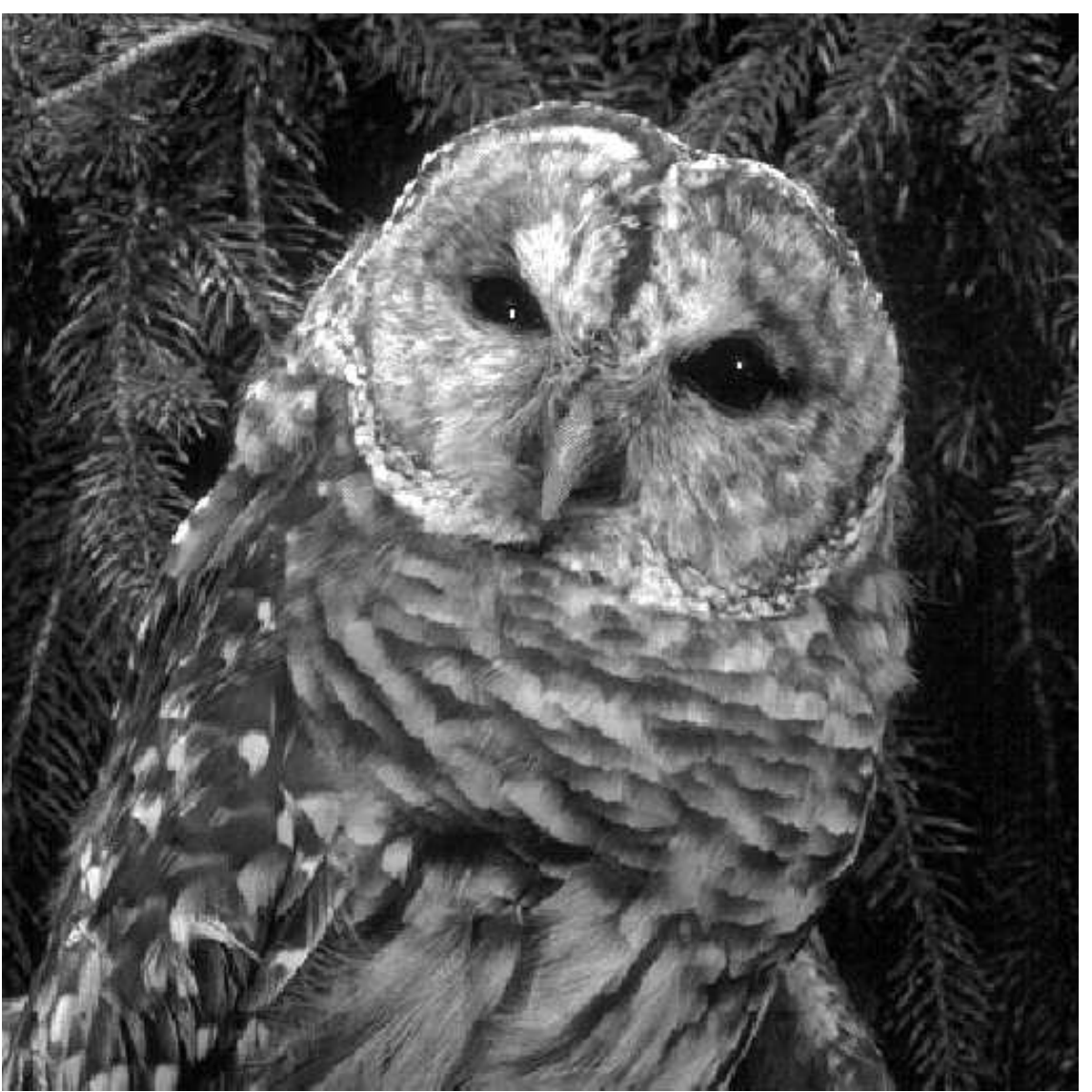}\\
    \vspace{0.05in}
    \includegraphics[width = 0.95\textwidth]{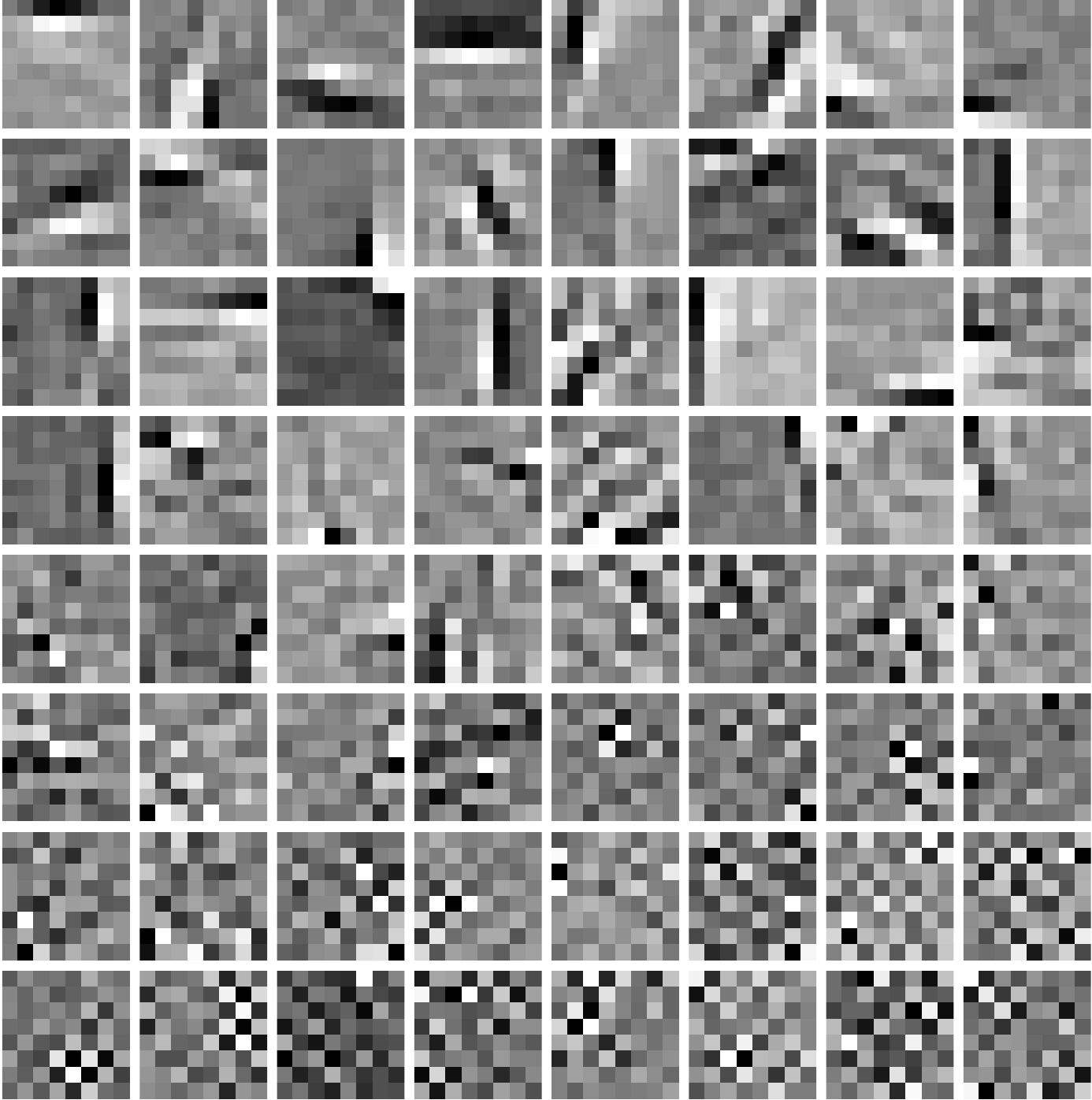} \\
    \vspace{0.05in}
    \includegraphics[width = \textwidth]{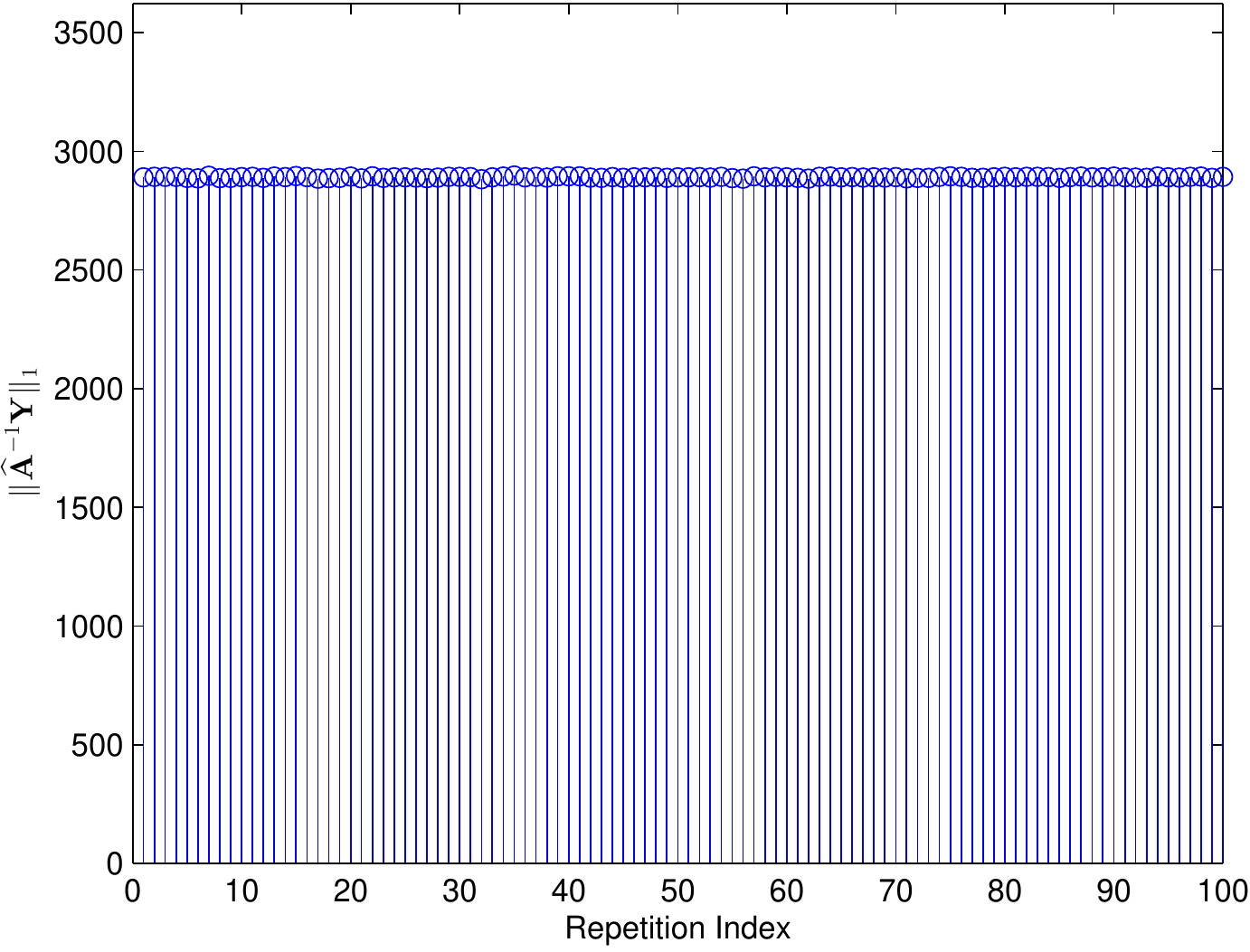}
\end{minipage}
\caption{Results of learning complete dictionaries from image patches, using the algorithmic pipeline in Section~\ref{sec:main_result}. \textbf{Top:} Images we used for the experiment. These are the three images we used in our motivational experiment (Section I.B) in the companion paper~\cite{sun2015complete_a}. The way we formed the data matrix $\mb Y$ is exactly the same as in that experiment. \textbf{Middle:} The $64$ dictionary elements we learned. \textbf{Bottom:} Let $\wh{\mb A}$ be the final dictionary matrix at convergence. This row shows the value $\|\wh{\mb A}^{-1} \mb Y\|_1$ across one hundred independent runs. The values are almost the same, with a relative difference less than $10^{-3}$. }
\label{fig:exp_img_trm}
\end{figure}

We take the three images that we used in the motivational experiment. For each image, we divide it into $8 \times 8$ non-overlapping patches, vectorize the patches, and then stack the vectorized patches into a data matrix $\mb Y$. $\mb Y$ is preconditioned as 
\begin{align*}
\ol{\mb Y} = \paren{\mb Y \mb Y^\top}^{-1/2} \mb Y, 
\end{align*}
and the resulting $\ol{\mb Y}$ is fed to the dictionary learning pipeline described in Section~\ref{sec:main_result}. The smoothing parameter $\mu$ is fixed to $10^{-2}$. Fig.~\ref{fig:exp_img_trm} contains the learned dictionaries: the dictionaries generally contain localized, directional features that resemble subset of wavelets and generalizations. These are very reasonable representing elements for natural images. Thus, the BG coefficient model may be a sensible, simple model for natural images. 

Another piece of strong evidence in support of the above claim is as follows. For each image, we repeat the learning pipeline for one hundred times, with independent initializations across the runs. Let $\wh{\mb A}$ be the final learned dictionary for each run, we plot the value of $\|\wh{\mb A}^{-1} \mb Y\|_1$ across the one hundred independent runs. Strikingly, the values are virtually the same, with a relative difference of $10^{-3}$! This is predicted by our theory, under the BG model. If the model is unreasonable for natural images, the preconditioning, benign function landscape, LP rounding, and the deflation process that hinge on this model would have completely fallen down.  

For this image experiment, $n = 64$ and $p = 4096$. A single run of the learning pipeline, including solving $64$ instances of the optimization over the sphere (with varying dimensions) and solving $64$ instances of the LP rounding (using CVX), lasts about $20$ minutes on a mid-range modern laptop. So with careful implementation we discussed above, the learning pipeline is actually not far from practical. 
}

%% file: sec/discuss.tex
\section{Discussion} \label{sec:discuss}
For recovery of complete dictionaries, the LP program approach in~\cite{spielman2012exact} that works with $\theta \le O(1/\sqrt{n})$ only demands $p \ge \Omega(n^2 \log^2 n)$. \js{The sample complexity has recently been improved to $p \ge \Omega(n \log n)$~\cite{adamczak2016note}, matching the information-theoretic lower bound $\Omega(n \log n)$ (i.e., when $\theta \sim 1/n$; see also~\cite{luh2015dictionary}), albeit still in the $\theta \in O(1/\sqrt{n})$ regime. } The sample complexity of our method working in the $\theta \sim \Theta(1)$ regime as stated in Theorem~\ref{thm:main_comp} is obviously much higher. As already discussed in~\cite{sun2015complete_a}, working with other $\ell^1$ proxies, directly in the $\mb q$ space would likely save the sample complexity both for geometric characterization and algorithm analysis. Another possibility is to analyze the complete case directly, instead of treating it as transformed version of perturbed orthogonal case. 

Our experiments seem to suggest the necessary complexity level lies between $\Omega(n)$ and $O(n^2 \log n)$ even for the orthogonal case. While it is interesting to determine the true complexity requirement for the TRM, there could be other efficient algorithms that demand less. For example, simulations in~\cite{qu2014finding} seem to suggest $O(n\log n)$ samples suffice to guarantee efficient recovery. The simulations run an alternating direction algorithm fed with problem-specific initializations, nevertheless. 

\js{
Our analysis is based on exact trust-region subproblem solver and fixed step size. The convergence result for the practical version from~\cite{boumal2016global}, based on approximate solver and adaptive step size, is general, but pessimistic. It seems not difficult to adapt their analysis according to our objective geometry, and obtain a tight, practical convergence result.  
}

Our motivating experiment on real images in introduction of our companion paper~\cite{sun2015complete_a} remains mysterious. If we were to believe that real image data are ``nice'' and our objective there does not have spurious local minima either, it is surprising ADM would escape all other critical points -- this is not predicted by classic or modern theories. One reasonable place to start is to look at how gradient descent algorithms with generic initializations can escape from ridable saddle points (at least with high probability). The recent work~\cite{ge2015escaping} has showed that randomly perturbing each iterate can help gradient algorithm to achieve this with high probability.

%% file: sec/proof_algorithm.tex
\section{Proof of Convergence for the Trust-Region Algorithm}\label{sec:proof_algorithm}

\subsection{Proof of Lemma~\ref{lem:mag_lip_fq}} \label{proof:lem_mag_lip_fq}
\begin{proof} 
Using the fact $\tanh\paren{\cdot}$ and $1-\tanh^2\paren{\cdot}$ are bounded by one in magnitude, by \eqref{eq:fq_grad} and \eqref{eq:fq_hess} we have 
\begin{align*}
\norm{\nabla f\paren{\mb q}}{} & \le \frac{1}{p}\sum_{k=1}^p \norm{(\mb x_0)_k}{} \le \sqrt{n} \norm{\mb X_0}{\infty}, \\
\norm{\nabla^2 f\paren{\mb q}}{} & \le \frac{1}{p}\sum_{k=1}^p \frac{1}{\mu} \norm{(\mb x_0)_k}{}^2 \le \frac{n}{\mu} \norm{\mb X_0}{\infty}^2,
\end{align*}
for any $\mb q \in \bb S^{n-1}$. Moreover, 
\begin{align*}
\sup_{\mb q, \mb q' \in \bb S^{n-1}, \mb q \neq \mb q'} \frac{\norm{\nabla f\paren{\mb q} - \nabla f\paren{\mb q'}}{}}{\norm{\mb q - \mb q'}{}} 
& \le \frac{1}{p} \sum_{k=1}^p \norm{(\mb x_0)_k}{} \sup_{\mb q, \mb q' \in \bb S^{n-1}, \mb q \neq \mb q'}\frac{\abs{\tanh\paren{\frac{\mb q^* (\mb x_0)_k}{\mu}} - \tanh\paren{\frac{\mb q'^* (\mb x_0)_k}{\mu}}}}{\norm{\mb q - \mb q'}{}} \\
& \le \frac{1}{p} \sum_{k=1}^p \norm{(\mb x_0)_k}{} \frac{\norm{(\mb x_0)_k}{}}{\mu} \le \frac{n}{\mu} \norm{\mb X_0}{\infty}^2, 
\end{align*} 
where at the last line we have used the fact the mapping $\mb q \mapsto \mb q^* (\mb x_0)_k/\mu$ is $\norm{(\mb x_0)_k}{}/\mu$ Lipschitz, and $x \mapsto \tanh\paren{x}$ is $1$-Lipschitz, and the composition rule of Lipschitz functions (i.e., Lemma V.5 of~\cite{sun2015complete_a}). Similar argument yields the final bound. 
\end{proof}

\subsection{Proof of Lemma~\ref{lem:alg_approx_bd2}} \label{proof:lem_alg_approx_bd2}
\begin{proof} 
Suppose we can establish 
\begin{equation*}
\left| f\paren{ \exp_{\mb q}(\mb \delta) } - \widehat{f}\paren{\mb q, \mb \delta} \right| \;\le\; \frac{1}{6}\eta_f \norm{\mb \delta}{}^3. 
\end{equation*}
Applying this twice we obtain 
\begin{align*}
f( \exp_{\mb q}(\mb \delta_\star) ) 
& \le \widehat{f}(\mb q,\mb \delta_\star) + \frac{1}{6}\eta_f \Delta^3 
\le \widehat{f}(\mb q, \mb \delta) + \frac{1}{6}\eta_f \Delta^3
\le f(\exp_{\mb q}(\mb \delta)) + \frac{1}{3}\eta_f\Delta^3 
\le f(\mb q) - s + \frac{1}{3}\eta_f \Delta^3,
\end{align*}
as claimed. Next we establish the first result. Let $\mb \delta_0 = \frac{\mb \delta}{\norm{\mb \delta}{}}$, and $t = \norm{\mb \delta}{}$. Consider the composite function
\begin{align*}
	\zeta(t) \doteq f( \exp_{\mb q}(t \mb \delta_0)) = f( \mb q \cos(t) + \mb \delta_0 \sin(t) ),
\end{align*} 
and also 
\begin{align*}
\dot{\zeta}(t) &= \innerprod{ \nabla f\left( \mb q \cos(t) + \mb \delta_0 \sin(t) \right) }{ - \mb q \sin(t) + \mb \delta_0 \cos(t) } \\
\ddot{\zeta}(t) &= \innerprod{ \nabla^2 f\left( \mb q \cos(t) + \mb \delta_0 \sin(t) \right) ( - \mb q \sin(t) + \mb \delta_0 \cos(t) ) }{ - \mb q \sin(t) + \mb \delta_0 \cos(t) } \nonumber \\
& + \quad \innerprod{ \nabla f\left( \mb q \cos(t) + \mb \delta_0 \sin(t) \right) }{ - \mb q \cos(t) - \mb \delta_0 \sin(t) }. 
\end{align*}
In particular, this gives that 
\begin{align*}
	\zeta(0) &= f( \mb q) \\
\dot{\zeta}(0) &= \innerprod{ \mb \delta_0 }{\nabla f(\mb q) } \\
\ddot{\zeta}(0) &= \mb \delta_0^* \left( \nabla^2 f(\mb q) - \innerprod{ \nabla f(\mb q) }{\mb q} \mb I \right) \mb \delta_0.
\end{align*}

We next develop a bound on $\magnitude{ \ddot{\zeta}(t) - \ddot{\zeta}(0) }$. Using the triangle inequality, we can casually bound this difference as 
\begin{align*}
&\magnitude{ \ddot{\zeta}(t) - \ddot{\zeta}(0) } \\
\le\; & \magnitude{  \innerprod{ \nabla^2 f\left( \mb q \cos(t) + \mb \delta_0 \sin(t) \right) ( - \mb q \sin(t) + \mb \delta_0 \cos(t) ) }{ - \mb q \sin(t) + \mb \delta_0 \cos(t)   }  - \mb \delta_0^* \nabla^2 f(\mb q) \mb \delta_0} \nonumber \\ 
& \qquad + \qquad \magnitude{ \innerprod{ \nabla f\left( \mb q \cos(t) + \mb \delta_0 \sin(t) \right) }{ - \mb q \cos(t) - \mb \delta_0 \sin(t) } + \innerprod{ \nabla f(\mb q) }{\mb q} } \nonumber \\
\le\; & \magnitude{ \innerprod{ \left[ \nabla^2 f( \mb q \cos(t) + \mb \delta_0 \sin(t) ) - \nabla^2 f(\mb q) \right] \left( - \mb q \sin(t) + \mb \delta_0 \cos(t) \right) }{ - \mb q \sin(t) + \mb \delta_0 \cos(t) } }  \nonumber \\
& \qquad + \qquad \magnitude{ \innerprod{ \nabla^2 f(\mb q) \left( - \mb q \sin(t) + \mb \delta_0 \cos(t) - \mb \delta_0 \right) }{ - \mb q \sin(t) + \mb \delta_0 \cos(t) } } \nonumber \\
& \qquad + \qquad \magnitude{ \innerprod{ \nabla^2 f(\mb q) \mb \delta_0 }{ - \mb q \sin(t) + \mb \delta_0 \cos(t) - \mb \delta_0 } } \nonumber \\ 
& \qquad + \qquad \magnitude{ \innerprod{ \nabla f( \mb q \cos(t) + \mb \delta_0 \sin(t) ) }{ - \mb q \cos(t) - \mb \delta_0 \sin(t) } + \innerprod{ \nabla f(\mb q \cos(t) + \mb \delta_0 \sin(t) ) }{ \mb q } } \nonumber \\
& \qquad + \qquad \magnitude{ \innerprod{ \nabla f(\mb q \cos(t) + \mb \delta_0 \sin(t)) }{\mb q } - \innerprod{ \nabla f(\mb q) }{\mb q} } \\
\le\; & L_{\nabla^2} \norm{ \mb q \cos(t) + \mb \delta_0 \sin(t) - \mb q }{} \nonumber  \\
& \qquad + M_{\nabla^2} \norm{ - \mb q \sin(t) + \mb \delta_0 \cos(t) - \mb \delta_0 }{} \nonumber \\
& \qquad + M_{\nabla^2} \norm{ - \mb q \sin(t) + \mb \delta_0 \cos(t) - \mb \delta_0 }{} \nonumber \\
& \qquad + M_\nabla \norm{ - \mb q \cos(t) - \mb \delta_0 \sin(t) + \mb q }{} \nonumber \\
& \qquad + L_\nabla \norm{ \mb q \cos(t) + \mb \delta_0 \sin(t) - \mb q }{} \\
=\; & \left( L_{\nabla^2} + 2 M_{\nabla^2} + M_{\nabla} + L_{\nabla} \right) \sqrt{ (1 - \cos(t))^2 + \sin^2(t) } \\
=\; & \eta_f \sqrt{ 2 - 2\cos t} \le \eta_f \sqrt{ 4 \sin^2\paren{t/2}} \le \eta_f t, 
\end{align*}
where in the final line we have used the fact $1-\cos x = 2\sin^2\paren{x/2}$ and that $\sin x \le x$ for $x\in \brac{0, 1}$, and $M_{\nabla}$, $M_{\nabla^2}$, $L_{\nabla}$ and $L_{\nabla^2}$ are the quantities defined in Lemma~\ref{lem:mag_lip_fq}. By the integral form of Taylor's theorem in Lemma \ref{lem:Taylor-integral-form} and the result above, we have 
\begin{align*}
\abs{f\paren{\exp_{\mb q}(\mb \delta)} - \widehat{f}\paren{\mb q,\mb \delta} } &=
\magnitude{ \zeta(t) - \left( \zeta(0) + t \dot{\zeta}(0) + \tfrac{t^2}{2} \ddot{\zeta}(0) \right) } \\
& = \magnitude{ t^2\int_0^1 \paren{1-s} \ddot{\zeta}\paren{st}\; ds - \tfrac{t^2}{2} \ddot{\zeta}(0) }  \\
& =  t^2\magnitude{ \int_0^1 \paren{1-s} \brac{\ddot{\zeta}\paren{st} - \ddot{\zeta}\paren{0} } \; ds }   \\
& \le t^2 \int_0^1 \paren{1-s}st \eta_f\; ds = \frac{\eta_f t^3}{6},
\end{align*}
with $t = \norm{\mb \delta}{}$ we obtain the desired result. 
\end{proof}

\subsection{Proof of Lemma~\ref{lem:alg_gradient_func}}  \label{proof:lem_alg_gradient_func}
\begin{proof}
By the integral form of Taylor's theorem in Lemma \ref{lem:Taylor-integral-form}, for any $t \in \brac{0, \frac{3\Delta}{2\pi\sqrt{n}}}$, we have 
\begin{align*}
& g\left( \mb w - t \frac{\mb w}{\norm{\mb w}{}} \right) \\
=\; & g(\mb w) - t\int_0^1 \innerprod{\nabla g\paren{\mb w -st\frac{\mb w}{\norm{\mb w}{}}}}{\frac{\mb w}{\norm{\mb w}{}}}\; ds \\
=\; & g\paren{\mb w} - t\frac{\mb w^* \nabla g\paren{\mb w}}{\norm{\mb w}{}} + t\int_0^1 \innerprod{\nabla g\paren{\mb w} - \nabla g\paren{\mb w -st\frac{\mb w}{\norm{\mb w}{}}}}{\frac{\mb w}{\norm{\mb w}{}}}\; ds \\
=\; & g\paren{\mb w} - t\frac{\mb w^* \nabla g\paren{\mb w}}{\norm{\mb w}{}} + t\int_0^1 \paren{\innerprod{\nabla g\paren{\mb w}}{\frac{\mb w}{\norm{\mb w}{}}} - \innerprod{\nabla g\paren{\mb w - st \frac{\mb w}{\norm{\mb w}{}}}}{\frac{\mb w - st\mb w/\norm{\mb w}{}}{\norm{\mb w - st \mb w/\norm{\mb w}{}}{}}}}\; ds \\
\le\; & g\paren{\mb w} - t\frac{\mb w^* \nabla g\paren{\mb w}}{\norm{\mb w}{}} + \frac{L_g}{2} t^2 \le g\paren{\mb w} - t\beta_g + \frac{L_g}{2}t^2. 
\end{align*}
Minimizing this function over $t \in \left[0, \frac{3\Delta }{ 2\pi \sqrt{n}} \right]$, we obtain that there exists a $\mb w' \in \mc B\left(\mb w,\frac{3\Delta }{ 2\pi \sqrt{n} } \right)$ such that 
\begin{align*}
g(\mb w') \;\le\; g(\mb w) - \min \set{ \frac{\beta_g^2}{2 L_g}, \frac{3\beta_g \Delta}{4\pi\sqrt{n}} }.
\end{align*}
Given such a $\mb w'\in \mc B\left(\mb w,\frac{3\Delta }{ 2\pi \sqrt{n} } \right)$, there must exist some $\mb \delta\in T_{\mb q}\bb S^{n-1}$ such that $\mb q(\mb w') = \exp_{\mb q}(\mb \delta)$. 
It remains to show that $\norm{\mb \delta}{}\leq \Delta$. It is easy to verify that $\norm{\mb q(\mb w') - \mb q\paren{\mb w}}{} \le 2 \sqrt{n} \norm{\mb w' - \mb w}{} \le 3\Delta/\pi$. Hence,  
\begin{align*}
\norm{\exp_{\mb q}\paren{\mb \delta} - \mb q}{}^2 = \norm{\mb  q\paren{1-\cos\norm{\mb \delta}{}} + \frac{\mb \delta}{\norm{\mb \delta}{}} \sin\norm{\mb \delta}{}}{}^2 = 2 - 2\cos \norm{\mb \delta}{} = 4\sin^2 \frac{\norm{\mb \delta}{}}{2} \le \frac{9\Delta^2}{\pi^2}, 
\end{align*}
which means that $\sin\paren{\norm{\mb \delta}{}/2} \le 3\Delta/\paren{2\pi}$. Because $\sin x \ge \tfrac{3}{\pi} x$ over $x \in \brac{0, \pi/6}$, it implies that $\norm{\mb \delta}{} \le \Delta$. Since $g(\mb w) = f(\mb q(\mb w))$, by summarizing all the results, we conclude that there exists a $\mb \delta$ with $\norm{\mb \delta}{}\leq \Delta$, such that
\begin{align*} 
f(\exp_{\mb q}(\mb \delta)) \le f(\mb q ) - \min \set{ \frac{\beta_g^2}{2 L_g}, \frac{3\beta_g \Delta}{4\pi\sqrt{n}} },
\end{align*}
as claimed. 
\end{proof}

\subsection{Proof of Lemma~\ref{lem:alg_neg_cuv_func}} \label{proof:lem_alg_neg_cuv_func}
\begin{proof} 
Let $\sigma = \mathrm{sign}\paren{ \mb w^* \nabla g(\mb w) }$. For any $t \in \brac{0, \frac{\Delta}{2\sqrt{n}}}$, by integral form of Taylor's theorem in Lemma \ref{lem:Taylor-integral-form}, we have  
\begin{align*}
& g\left( \mb w - t \sigma \frac{\mb w}{\norm{\mb w}{}} \right) \\
=\; & g(\mb w) - t \sigma \frac{\mb w^* \nabla g(\mb w) }{\norm{\mb w}{}} + t^2 \int_0^1 \paren{1-s} \frac{\mb w^* \nabla^2 g\paren{\mb w - st\sigma \frac{\mb w}{\norm{\mb w}{}}}\mb w}{\norm{\mb w}{}^2}\;ds \\
\le\; & g(\mb w) + \frac{t^2}{2} \frac{\mb w^* \nabla^2 g(\mb w) \mb w }{\norm{\mb w}{}^2} + t^2 \int_0^1 \brac{\paren{1-s} \frac{\mb w^* \nabla^2 g\paren{\mb w - st\sigma \frac{\mb w}{\norm{\mb w}{}}}\mb w}{\norm{\mb w}{}^2} -  \paren{1-s} \frac{\mb w^* \nabla^2 g(\mb w) \mb w }{\norm{\mb w}{}^2}} \; ds \\
=\; & g(\mb w) + \frac{t^2}{2} \frac{\mb w^* \nabla^2 g(\mb w) \mb w }{\norm{\mb w}{}^2} \\
& +  t^2\int_0^1 \paren{1-s}\brac{\frac{\paren{\mb w - st\sigma \frac{\mb w}{\norm{\mb w}{}}}^* \nabla^2 g\paren{\mb w - st\sigma \frac{\mb w}{\norm{\mb w}{}}}\paren{\mb w - st\sigma \frac{\mb w}{\norm{\mb w}{}}}}{\norm{\mb w - st\sigma \frac{\mb w}{\norm{\mb w}{}}}{}^2} -\frac{\mb w^* \nabla^2 g(\mb w) \mb w }{\norm{\mb w}{}^2} }\; ds\\
\le\; & g(\mb w) - \frac{t^2}{2} \betaconcave + t^2 \int_0^1\paren{1-s}s\Lconcave t\; ds \;\leq\; g(\mb w) - \frac{t^2}{2} \betaconcave + \frac{ t^3}{6} \Lconcave.
\end{align*}
Minimizing this function over $t \in \left[ 0, \frac{3\Delta}{2\pi \sqrt{n}} \right]$, we obtain 
\begin{align*}
t_\star = \min \set{ \frac{2\betaconcave}{\Lconcave}, \frac{3\Delta}{2\pi \sqrt{n}} },
\end{align*}
and there exists a $\mb w' = \mb w - t_\star \sigma \frac{\mb w}{\norm{\mb w}{}}$ such that
\begin{align*}
g\left( \mb w - t_\star \sigma \frac{\mb w}{\norm{\mb w}{}} \right) \;\le\; g(\mb w) - \min \set{ \frac{2 \betaconcave^3}{3 \Lconcave^2}, \frac{3\Delta^2 \betaconcave}{8\pi^2 n} }. 
\end{align*}
By arguments identical to those used in Lemma \ref{lem:alg_gradient_func}, there exists a tangent vector $\mb \delta\in T_{\mb q}\bb S^{n-1}$ such that $\mb q(\mb w') = \exp_{\mb q}(\mb \delta)$ and $\norm{\mb \delta}{}\leq \Delta$. This completes the proof. 
\end{proof}

\subsection{Proof of Lemma~\ref{lem:alg_strcvx_func}} \label{proof:lem_alg_strcvx_func}
\begin{proof} 
For any $t \in \brac{0, \frac{\Delta}{\norm{\grad f\paren{\mb q^{(r)}}}{}}}$, it holds that $\norm{t\; \grad f\paren{\mb q^{(r)}}}{} \le \Delta$, and the quadratic approximation 
\begin{align*}
\widehat{f}\paren{\mb q^{(r)}, -t\; \grad f\paren{\mb q^{(r)}}}
& \le f\paren{\mb q^{(r)}} - t\norm{\grad f\paren{\mb q^{(r)}}}{}^2 + \frac{M_H}{2}t^2 \norm{\grad f\paren{\mb q^{(r)}}}{}^2 \\
& = f\paren{\mb q^{(r)}} - t\paren{1-\frac{1}{2}M_H t}\norm{\grad f\paren{\mb q^{(r)}}}{}^2. 
\end{align*}
Taking $t_0 = \min\set{\frac{\Delta}{\norm{\grad f\paren{\mb q^{(r)}}}{}}, \frac{1}{M_H}}$, we obtain 
\begin{align} \label{eq:alg_strcvx_key1}
\widehat{f}\paren{\mb q^{(r)}, -t_0\; \grad f\paren{\mb q^{(r)}}} \le f\paren{\mb q^{(r)}} - \frac{1}{2}\min\set{\frac{\Delta}{\norm{\grad f\paren{\mb q^{(r)}}}{}}, \frac{1}{M_H}} \norm{\grad f\paren{\mb q^{(r)}}}{}^2. 
\end{align}
Now let $\mb U$ be an arbitrary orthonormal basis for $T_{\mb q^{(r)}} \bb S^{n-1}$. Since the norm constraint is active, by the optimality condition in \eqref{eqn:ts-optimal-solution-1}, we have 
\begin{align*}
\Delta & \le \norm{\brac{\mb U^* \Hess f\paren{\mb q^{(r)}} \mb U}^{-1} \mb U^* \grad f\paren{\mb q^{(r)}}}{} \\
& \le \norm{\brac{\mb U^* \Hess f\paren{\mb q^{(r)}} \mb U}^{-1}}{} \norm{\mb U^* \grad f\paren{\mb q^{(r)}}}{} \le \frac{\norm{\grad f \paren{\mb q^{(r)}}}{}}{m_H},  
\end{align*} 
which means that $\norm{\grad f \paren{\mb q^{(r)}}}{} \ge m_H\Delta$. Substituting this into~\eqref{eq:alg_strcvx_key1}, we obtain 
\begin{align*}
\widehat{f}\paren{\mb q^{(r)}, -t_0\; \grad f\paren{\mb q^{(r)}}} \le f\paren{\mb q^{(r)}} - \frac{1}{2} \min \Brac{m_H \Delta^2, \frac{m_H^2}{M_H} \Delta^2  }\leq f\paren{\mb q^{(r)}}- \frac{m_H^2 \Delta^2 }{2M_H}.
\end{align*}
By the key comparison result established in proof of Lemma~\ref{lem:alg_approx_bd2}, we have 
\begin{align*}
f\paren{\exp_{\mb q^{(r)}}\paren{-t_0\grad f\paren{\mb q^{(r)}}}} & \le \widehat{f}\paren{\mb q^{(r)}, -t_0\; \grad f\paren{\mb q^{(r)}}} + \frac{1}{6} \eta_f\Delta^3 \\
& \le  f\paren{\mb q^{(r)}} - \frac{m_H^2 \Delta^2}{M_H} + \frac{1}{6} \eta_f\Delta^3. 
\end{align*}
This completes the proof. 
\end{proof}

\subsection{Proof of Lemma~\ref{lem:alg_strcvx_lb}}\label{proof:lem_alg_strcvx_lb}
It takes certain delicate work to prove Lemma~\ref{lem:alg_strcvx_lb}. Basically to use discretization argument, the degree of continuity of the Hessian is needed. The tricky part is that for continuity, we need to compare the Hessian operators at different points, while these Hessian operators are only defined on the respective tangent planes. This is the place where parallel translation comes into play. The next two lemmas compute spectral bounds for the forward and inverse parallel translation operators. 
\begin{lemma} \label{lem:alg_tsp_op}
For $\tau \in [0,1]$ and $\norm{\mb \delta}{} \le 1/ 2$, we have
\begin{eqnarray}
\norm{ \mc P_{\gamma}^{\tau \leftarrow 0} - \mb I }{} &\le& \frac{5}{4}\tau \norm{\mb \delta}{}, \\
\norm{ \mc P_{\gamma}^{0 \leftarrow \tau} - \mb I }{} &\le& \frac{3}{2}\tau \norm{\mb \delta}{}.
\end{eqnarray}
\end{lemma}
\begin{proof}
By~\eqref{eq:alg_tsp_op}, we have 
\begin{align*}
\norm{ \mc P_{\gamma}^{\tau \leftarrow 0} - \mb I }{}
& = \norm{\left( \cos( \tau \norm{\mb \delta}{} ) - 1 \right) \frac{\mb \delta \mb \delta^*}{\norm{\mb \delta }{}^2} - \sin\left( \tau \norm{\mb \delta }{} \right) \frac{\mb q \mb \delta^*}{\norm{\mb \delta}{}}}{} \\
& \le 1-\cos\paren{\tau \norm{\mb \delta}{}} + \sin\paren{\tau \norm{\mb \delta}{}} \\
& \le 2\sin^2\paren{\frac{\tau\norm{\mb \delta}{}}{2}} + \sin\paren{\tau \norm{\mb \delta}{}} \le \frac{1}{4}\tau \norm{\mb \delta}{} +  \tau \norm{\mb \delta}{} \le \frac{5}{4} \tau \norm{\mb \delta}{}, 
\end{align*}
where we have used the fact $\sin\paren{t} \le t$ and $1-\cos x = 2\sin^2\paren{x/2}$. Moreover, $\mc P_{\gamma}^{0 \leftarrow \tau}$ is in the form of $\paren{\mb I + \mb u \mb v^*}^{-1}$ for some vectors $\mb u$ and $\mb v$. By the Sherman-Morrison matrix inverse formula, i.e., $\paren{\mb I + \mb u \mb v^*}^{-1} = \mb I - \mb u \mb v^*/\paren{1+ \mb v^* \mb u}$ (justified as $\norm{\left( \cos( \tau \norm{\mb \delta}{} ) - 1 \right) \frac{\mb \delta \mb \delta^*}{\norm{\mb \delta }{}^2} - \mb q \sin\left( \tau \norm{\mb \delta }{} \right) \frac{\mb \delta^*}{\norm{\mb \delta}{}}}{} \le 5\tau\norm{\mb \delta}{}/4 \le 5/8 < 1$ as shown above), we have 
\begin{align*}
& \norm{ \mc P_{\gamma}^{0 \leftarrow \tau} - \mb I }{} \\
=\; & \norm{\left( \cos( \tau \norm{\mb \delta}{} ) - 1 \right) \frac{\mb \delta \mb \delta^*}{\norm{\mb \delta }{}^2} - \mb q \sin\left( \tau \norm{\mb \delta }{} \right) \frac{\mb \delta^*}{\norm{\mb \delta}{}}}{} \frac{1}{1+\paren{\cos\paren{\tau \norm{\mb \delta}{}}-1}} \quad (\text{as}\; \mb q^* \mb \delta = 0)\\
\le\; & \frac{5}{4}\tau \norm{\mb \delta}{} \frac{1}{\cos\paren{\tau \norm{\mb \delta}{}} } \le \frac{5}{4}\tau \norm{\mb \delta}{} \frac{1}{\cos\paren{1/2}} \le \frac{3}{2} \tau \norm{\mb \delta}{},  
\end{align*}
completing the proof. 
\end{proof}

The next lemma establishes the ``local-Lipschitz" property of the Riemannian Hessian. 
\begin{lemma} \label{lem:rie_hess_lip}
Let $\gamma(t) = \exp_{\mb q}\paren{t\mb \delta}$ denotes a geodesic curve on $\bb S^{n-1}$. Whenever $\norm{\mb \delta}{} \le 1/2$ and $\tau \in [0,1]$,
\begin{eqnarray}
\norm{ \mc P_{\gamma}^{0 \leftarrow \tau} \Hess f (\gamma(\tau)) \mc P_{\gamma}^{\tau \leftarrow 0} - \Hess f (\mb q) }{} \le L_H\cdot\tau\norm{\mb \delta}{}, 
\end{eqnarray}
where $L_H = 5n^{3/2} \norm{\mb X_0}{\infty}^3/(2\mu^2) + \frac{9}{\mu}n \norm{\mb X_0}{\infty}^2 + 9\sqrt{n} \norm{\mb X_0}{\infty}$. 
\end{lemma}
\begin{proof}
First of all, by \eqref{eq:fq_rie_hess} and using the fact that the operator norm of a projection operator is unitary bounded, we have
\begin{align*}
& \norm{\Hess f (\gamma(\tau)) - \Hess f (\mb q)}{} \\
\le\; & \norm{\mc P_{T_{\gamma\paren{\tau}}\bb S^{n-1}} \brac{\nabla^2 f\paren{\gamma\paren{\tau}} - \nabla^2 f\paren{\mb q} -\paren{ \innerprod{\nabla f\paren{\gamma\paren{\tau}}}{\gamma\paren{\tau}} - \innerprod{\nabla f\paren{\mb q}}{\mb q} }\mb I} \mc P_{T_{\gamma\paren{\tau}}\bb S^{n-1} } }{} \\
& \qquad + \left\|\mc P_{T_{\gamma\paren{\tau}}\bb S^{n-1}}  \paren{\nabla^2 f\paren{\mb q} - \innerprod{\nabla f\paren{\mb q}}{\mb q} \mb I} \mc P_{T_{\gamma\paren{\tau}}\bb S^{n-1}}\right. \\
& \qquad \left.  - \mc P_{T_{\mb q}\bb S^{n-1} } \paren{\nabla^2 f\paren{\mb q} - \innerprod{\nabla f\paren{\mb q}}{\mb q} \mb I}  \mc P_{T_{\mb q}\bb S^{n-1}}\right\| \\
\le\; & \norm{\nabla^2 f\paren{\gamma\paren{\tau}} - \nabla^2 f\paren{\mb q}}{} + \abs{\innerprod{\nabla f\paren{\gamma\paren{\tau}} - \nabla f\paren{\mb q}}{\gamma\paren{\tau}}} + \abs{\innerprod{\nabla f\paren{\mb q}}{\gamma\paren{\tau} - \mb q}} \\
& \qquad + \norm{\mc P_{T_{\gamma\paren{\tau}}\bb S^{n-1} } - \mc P_{T_{\mb q}\bb S^{n-1} }}{} \norm{\mc P_{T_{\gamma\paren{\tau}}\bb S^{n-1}} + \mc P_{T_{\mb q}\bb S^{n-1} }}{}\norm{\nabla^2 f\paren{\mb q} - \innerprod{\nabla f\paren{\mb q}}{\mb q} \mb I}{}. 
\end{align*}
By the estimates in Lemma~\ref{lem:mag_lip_fq}, we obtain
\begin{align}
& \norm{\Hess f (\gamma(\tau)) - \Hess f (\mb q)}{} \nonumber  \\
\le\; &  \frac{2}{\mu^2}n^{3/2} \norm{\mb X_0}{\infty}^3 \norm{\gamma\paren{\tau} - \mb q}{} + \frac{n}{\mu} \norm{\mb X_0}{\infty}^2  \norm{\gamma\paren{\tau} - \mb q}{} + \sqrt{n} \norm{\mb X_0}{\infty} \norm{\gamma\paren{\tau} - \mb q}{} \nonumber \\
& \qquad + 2\norm{\gamma\paren{\tau} \gamma^* \paren{\tau} - \mb q \mb q^*}{} \paren{\frac{n}{\mu}\norm{\mb X_0}{\infty}^2 + \sqrt{n} \norm{\mb X_0}{\infty}} \nonumber \\
\le & \paren{\frac{5}{2\mu^2} n^{3/2} \norm{\mb X_0}{\infty}^3 + \frac{25n}{4\mu} \norm{\mb X_0}{\infty}^2 + \frac{25}{4} \sqrt{n} \norm{\mb X_0}{\infty}} \tau \norm{\mb \delta}{}, \label{eq:alg_lip_key1}
\end{align}
where at the last line we have used the following estimates: 
\begin{align*}
\norm{\gamma\paren{\tau} - \mb q}{} & = \norm{\mb q \paren{\cos\paren{\tau \norm{\mb \delta}{}} -1} + \frac{\mb \delta}{\norm{\mb \delta}{}} \sin\paren{\tau \norm{\mb \delta}{}}}{} \le \frac{5}{4} \tau \norm{\mb \delta}{}, \; (\text{Proof of Lemma~\ref{lem:alg_tsp_op}}) \\
\norm{\gamma\paren{\tau} \gamma^* \paren{\tau} - \mb q \mb q^*}{} & \le \norm{\paren{\frac{\mb \delta \mb \delta^*}{\norm{\mb \delta}{}^2} - \mb q \mb q^*}\sin^2\paren{\tau \norm{\mb \delta}{}}}{} + 2\sin\paren{\tau \norm{\mb \delta}{}}\cos\paren{\tau \norm{\mb \delta}{}} \\
& \le \sin^2\paren{ \tau \norm{\mb \delta}{}} + \sin\paren{2\tau \norm{\mb \delta}{}} \le \frac{5}{2}\tau \norm{\mb \delta}{}. 
\end{align*}
Therefore, by Lemma \ref{lem:alg_tsp_op}, we obtain 
\begin{align*}
& \norm{ \mc P_{\gamma}^{0 \leftarrow \tau} \Hess f (\gamma(\tau)) \mc P_{\gamma}^{\tau \leftarrow 0} - \Hess f (\mb q) }{}  \\
\le\; & \norm{ \mc P_{\gamma}^{0 \leftarrow \tau} \Hess f (\gamma(\tau)) \mc P_{\gamma}^{\tau \leftarrow 0} -  \Hess f (\gamma(\tau)) \mc P_{\gamma}^{\tau \leftarrow 0}  }{} + \norm{ \Hess f (\gamma(\tau)) \mc P_{\gamma}^{\tau \leftarrow 0}  -  \Hess f (\gamma(\tau)) }{}  \\
& \qquad + \norm{\Hess f (\gamma(\tau)) - \Hess f (\mb q)}{} \nonumber \\
\le\; & \norm{\mc P_{\gamma}^{0 \leftarrow \tau} - \mb I}{} \norm{\Hess f (\gamma(\tau))}{} + \norm{\mc P_{\gamma}^{\tau \leftarrow 0} - \mb I}{} \norm{\Hess f (\gamma(t))}{} + \norm{\Hess f (\gamma(t)) - \Hess f (\mb q)}{} \\
\le\; & \frac{11}{4} \tau \norm{\mb \delta}{} \norm{\nabla^2 f\paren{\gamma\paren{\tau}} - \innerprod{\nabla f\paren{\gamma\paren{\tau}}}{\gamma\paren{t}} \mb I}{} + \norm{\Hess f (\gamma(\tau)) - \Hess f (\mb q)}{}. 
\end{align*}
By Lemma~\ref{lem:mag_lip_fq} and substituting the estimate in~\eqref{eq:alg_lip_key1}, we obtain the claimed result. 
\end{proof}

\begin{proof}{\textbf{(of Lemma~\ref{lem:alg_strcvx_lb})}}
For any given $\mb q$ with $\norm{\mb w(\mb q)}{} \le \mu/(4\sqrt{2})$, assume $\mb U$ is an orthonormal basis for its tangent space $T_{\mb q}\bb S^{n-1}$. Again we first work with the ``canonical'' section in the vicinity of $\mb e_n$ with the ``canonical" reparametrization $\mb q(\mb w) = [\mb w; \sqrt{1-\|\mb w\|^2}]$. 

\begin{enumerate}
\item \textbf{Expectation of the operator.}
By definition of the Riemannian Hessian in~\eqref{eq:fq_rie_hess}, expressions of $\nabla^2 f$ and $\nabla f$ in~\eqref{eq:fq_grad} and~\eqref{eq:fq_hess}, and exchange of differential and expectation operators, we obtain 
\begin{align*}
\mb U^* \Hess \expect{f(\mb q)} \mb U 
& = \expect{\mb U^* \Hess f(\mb q) \mb U} \\
& = \expect{\mb U^* \nabla^2 f(\mb q) \mb U - \innerprod{\mb q}{\nabla f(\mb q)} \mb I_{n-1}} \\
& = \mb U^* \expect{\frac{1}{\mu} \Brac{1-\tanh^2\paren{\frac{\mb q^* \mb x}{\mu}}} \mb x \mb x^*} \mb U - \expect{\tanh\paren{\frac{\mb q^* \mb x}{\mu}} \mb q^* \mb x} \mb I_{n-1}. 
\end{align*} 
Let $\ol{\mb x} \in \R^{n-1}$ be the first $n-1$ coordinates of $\mb x$ and $\mb w \in \R^{n-1}$ the similar subvector of $\mb q$ (as used in~\cite{sun2015complete_a}). We have 
\begin{align*}
\mb U^* \expect{\frac{1}{\mu} \Brac{1-\tanh^2\paren{\frac{\mb q^* \mb x}{\mu}}} \mb x \mb x^*} \mb U \succeq \frac{1-\theta}{\mu} \mb U^* \expect{\Brac{1-\tanh^2\paren{\frac{\mb w^* \ol{\mb x}}{\mu}}}
\begin{bmatrix}
\ol{\mb x}\; \ol{\mb x}^* & \mb 0 \\
\mb 0^*   &  0
\end{bmatrix}} \mb U.
\end{align*}
Now consider any vector $\mb z \in T_{\mb q} \bb S^{n-1}$ such that $\mb z = \mb U \mb v$ for some $\mb v \in \R^{n-1}$ and $\|\mb z\| = 1$. Then 
\begin{align*}
\mb z^* \expect{\Brac{1-\tanh^2\paren{\frac{\mb w^* \ol{\mb x}}{\mu}}}
\begin{bmatrix}
\ol{\mb x} \;\ol{\mb x}^* & \mb 0 \\
\mb 0^*   &  0
\end{bmatrix}} \mb z \ge \frac{\theta}{\sqrt{2\pi}} \paren{2 - \frac{3}{4}\sqrt{2}} \|\ol{\mb z}\|^2 
\end{align*}
by proof of Proposition II.7 in~\cite{sun2015complete_a}, where $\ol{\mb z} \in \R^{n-1}$ as above is the first $n-1$ coordinates of $\mb z$. Now we know that $\innerprod{\mb q}{\mb z} = 0$, or 
\begin{align*}
\mb w^* \ol{\mb z} + q_n z_n = 0 \Longrightarrow \frac{\|\ol{\mb z}\|}{\abs{z_n}} = \frac{q_n}{\|\mb w\|} = \frac{\sqrt{1-\|\mb w\|^2}}{\|\mb w\|} \ge 50, 
\end{align*}
where we have used $\|\mb w\| \le \mu/(4\sqrt{2})$ and $\mu \le 1/10$ to obtain the last lower bound. Combining the above with the fact that $\|\mb z\| = 1$, we obtain 
\begin{align}
\mb U^* \expect{\frac{1}{\mu} \Brac{1-\tanh^2\paren{\frac{\mb q^* \mb x}{\mu}}} \mb x \mb x^*} \mb U 
& \succeq \frac{99}{100} \frac{1-\theta}{\mu} \frac{\theta}{\sqrt{2\pi}} (2 - 3\sqrt{2}/4) \mb I_{n-1} \nonumber \\
& \succeq \frac{99}{200\sqrt{2\pi}} (2 - 3\sqrt{2}/4) \frac{\theta}{\mu} \mb I_{n-1}, 
\end{align}
where we have simplified the expression using $\theta \le 1/2$. To bound the second term, 
\begin{align*}
& \expect{\tanh\paren{\frac{\mb q^* \mb x}{\mu}} \mb q^* \mb x} \\
=\; & \bb E_{\mc I} \brac{\bb E_{Z \sim \mc N\paren{0, \norm{\mb q_{\mc I}}{}^2}}\brac{\tanh(Z/\mu) Z}} \\
=\; & \frac{1}{\mu}\bb E_{\mc I} \brac{\|\mb q_{\mc I}\|^2\bb E_{Z \sim \mc N\paren{0, \norm{\mb q_{\mc I}}{}^2}}\brac{1-\tanh^2(Z/\mu)}}   \quad \text{(by Lemma B.1 in~\cite{sun2015complete_a})}  \\
\le\; & \frac{1}{\mu}\bb E_{\mc I} \brac{\bb E_{Z \sim \mc N\paren{0, \norm{\mb q_{\mc I}}{}^2}}\brac{1-\tanh^2(Z/\mu)}}. 
\end{align*}
Now we have the following estimate: 
\begin{align*}
& \bb E_{Z \sim \mc N\paren{0, \norm{\mb w_{\mc J}}{}^2 + q_n^2}}\brac{1-\tanh^2(Z/\mu)} \\
=\; & 2 \bb E_{Z \sim \mc N\paren{0, \norm{\mb w_{\mc J}}{}^2 + q_n^2}}\brac{\paren{1-\tanh^2(Z/\mu)} \indicator{Z > 0}} \\
\le\; & 8 \bb E_{Z \sim \mc N\paren{0, \norm{\mb w_{\mc J}}{}^2 +  q_n^2}}\brac{\exp(-2Z/\mu) \indicator{Z > 0}} \\
=\; & 8 \exp\paren{\frac{2\norm{\mb w_{\mc J}}{}^2 + 2 q_n^2}{\mu^2}} \Phi^c\paren{\frac{2\sqrt{\norm{\mb w_{\mc J}}{}^2 +  q_n^2}}{\mu}}  \quad \text{(by Lemma B.1 in~\cite{sun2015complete_a})} \\
\le\; & \frac{4}{\sqrt{2\pi}} \frac{\mu}{\sqrt{\norm{\mb w_{\mc J}}{}^2 +  q_n^2}}, 
\end{align*}
where at the last inequality we have applied Gaussian tail upper bound of Type II in Lemma~\ref{lem:gaussian_tail_est}.  Since $\norm{\mb w_{\mc J}}{}^2 +  q_n^2 \ge q_n^2 = 1-\norm{\mb w}{}^2 \ge 1 - \mu^2/32 \ge 31/32$ for $\norm{\mb w}{} \le \mu/(4\sqrt{2})$ and $\mu \le 1$, we obtain 
\begin{align} \label{eq:strcvx_qsp_key2}
\bb E_{Z \sim \mc N\paren{0, \norm{\mb w_{\mc J}}{}^2 + q_n^2}}\brac{1-\tanh^2(Z/\mu)} \le \frac{4}{\sqrt{2\pi}} \frac{\mu}{\sqrt{31/32}} \le \frac{4}{\sqrt{2\pi}} \mu. 
\end{align}
Collecting the above estimates, we obtain 
\begin{align} \label{eq:strcvx_qsp_key3}
\mb U^* \Hess \expect{f(\mb q)} \mb U \succeq \frac{99}{200\sqrt{2\pi}} (2 - 3\sqrt{2}/4) \frac{\theta}{\mu} \mb I_{n-1} - \frac{1}{\mu}\frac{4}{\sqrt{2\pi}} \mu \mb I_{n-1} \succeq \frac{1}{4\sqrt{2\pi}} \frac{\theta}{\mu} \mb I_{n-1}, 
\end{align}
where we have used the fact $\mu \le \theta/10$ to obtain the final lower bound. 

\item \textbf{Concentration.}
Next we perform concentration analysis. For any $\mb q$, we can write 
\begin{align*}
\mb U^* \nabla^2 f(\mb q) \mb U = \frac{1}{p} \sum_{k=1}^p \mb W_k, \quad \text{with}\; \mb W_k \doteq \frac{1}{\mu} \brac{1-\tanh^2\paren{\frac{\mb q^* (\mb x_0)_k}{\mu}}} \mb U^* (\mb x_0)_k (\mb x_0)_k^* \mb U. 
\end{align*}
For any integer $m \ge 2$, we have 
\begin{align*}
\mb  0 \preceq \expect{\mb W_k^m} 
& \preceq \frac{1}{\mu^m} \expect{\paren{\mb U^* (\mb x_0)_k (\mb x_0)_k^* \mb U}^m} \\
& \preceq \frac{1}{\mu^m} \expect{\norm{(\mb x_0)_k (\mb x_0)_k^*}{}^m}\mb I = \frac{1}{\mu^m} \expect{\norm{(\mb x_0)_k}{}^{2m}} \mb I \preceq \frac{1}{\mu^m} \bb E_{Z \sim \chi^2\paren{n}} \brac{Z^m} \mb I,  
\end{align*}
where we have used Lemma~\ref{lem:U-moments-bound} to obtain the last inequality. By Lemma~\ref{lem:chi_sq_moment}, we obtain 
\begin{align*}
\mb 0 \preceq \expect{\mb W_k^m} \preceq \frac{1}{\mu^m} \frac{m!}{2} \paren{2n}^m \mb I \preceq \frac{m!}{2} \paren{\frac{2n}{\mu}}^m \mb I. 
\end{align*}
Taking $R_{\mb W} = 2n/\mu$, and $\sigma^2_{\mb W} = 4n^2/\mu^2 \ge \expect{\mb W_k^2}$, by Lemma~\ref{lem:mc_bernstein_matrix}, we obtain 
\begin{align} \label{eq:strcvx_qsp_key4}
\prob{\norm{\frac{1}{p} \sum_{k=1}^p \mb W_k - \frac{1}{p} \sum_{k=1}^p \expect{\mb W_k}}{} > t/2} \le 2n\exp\paren{-\frac{p\mu^2t^2}{32n^2 + 8nt}}
\end{align}
for any $t > 0$. Similarly, we write
\begin{align*}
\innerprod{\nabla f(\mb q)}{\mb q} = \frac{1}{p} \sum_{k=1}^p Z_k, \quad \text{with}\; Z_k \doteq \tanh\paren{\frac{\mb q^* (\mb x_0)_k}{\mu}} \mb q^* (\mb x_0)_k. 
\end{align*}
For any integer $m \ge 2$, we have 
\begin{align*}
\expect{\abs{Z_k}^m} \le \expect{\abs{\mb q^* (\mb x_0)_k}^m} \le \bb E_{Z \sim \mc N\paren{0, 1}}\brac{\abs{Z}^m} \le \frac{m!}{2}, 
\end{align*}
where at the first inequality we used the fact $\abs{\tanh(\cdot)} \le 1$, at the second we invoked Lemma~\ref{lem:U-moments-bound}, and at the third we invoked Lemma~\ref{lem:guassian_moment}. Taking $R_Z = \sigma^2_{Z} = 1$, by Lemma~\ref{lem:mc_bernstein_scalar}, we obtain 
\begin{align} \label{eq:strcvx_qsp_key5}
\prob{\abs{\frac{1}{p}\sum_{k=1}^p Z_k - \frac{1}{p}\sum_{k=1}^p \expect{Z_k}} > t/2} \le 2\exp\paren{-pt^2/16}
\end{align}
for any $t > 0$. Gathering~\eqref{eq:strcvx_qsp_key4} and~\eqref{eq:strcvx_qsp_key5}, we obtain that for any $t > 0$, 
\begin{align} \label{eq:strcvx_qsp_key6}
& \prob{\norm{\mb U^* \Hess \expect{f(\mb q)} \mb U - \mb U^* \Hess f(\mb q) \mb U}{} > t} \nonumber \\
\le\; &  \prob{\norm{\mb U^* \nabla^2 f(\mb q) \mb U - \nabla^2 \expect{f(\mb q)}}{} > t/2} + \prob{\abs{\innerprod{\nabla f(\mb q)}{\mb q} - \innerprod{\nabla \expect{f(\mb q)}}{\mb q} }> t/2} \nonumber \\
\le\; & 2n\exp\paren{-\frac{p\mu^2 t^2}{32n^2 + 8nt}} + 2\exp\paren{-\frac{pt^2}{16}} \le 4n \exp\paren{-\frac{p\mu^2 t^2}{32n^2 + 8nt}}. 
\end{align}

\item \textbf{Uniformizing the bound.}
Now we are ready to pull above results together for a discretization argument. For any $\eps \in (0, \mu/(4\sqrt{2}))$, there is an $\eps$-net $N_{\eps}$ of size at most $(3\mu/(4\sqrt{2} \eps))^n$ that covers the region $\set{\mb q: \norm{\mb w(\mb q)}{} \le \mu/(4\sqrt{2})}$. By Lemma~\ref{lem:rie_hess_lip}, the function $\Hess f(\mb q)$ is locally Lipschitz within each normal ball of radius 
\begin{align*}
\norm{\mb q - \exp_{\mb q} (1/2)}{} = \sqrt{2-2\cos(1/2)} \ge 1/\sqrt{5}
\end{align*}
with Lipschitz constant $L_H$ (as defined in Lemma~\ref{lem:rie_hess_lip}). Note that $\eps < \mu/(4\sqrt{2}) < 1/(4\sqrt{2}) < 1/\sqrt{5}$ for $\mu < 1$, so any choice of $\eps \in (0, \mu/(4\sqrt{2}))$ makes the Lipschitz constant $L_H$ valid within each $\eps$-ball centered around one element of the $\eps$-net. Let 
\begin{align*}
\event_{\infty} \doteq \set{1 \le \norm{\mb X_0}{\infty} \le 4\sqrt{\log(np)}}. 
\end{align*}
From Lemma~\ref{lem:X-infinty-tail-bound},  $\prob{\event_{\infty}^c} \le \theta \paren{np}^{-7} + \exp\paren{-0.3\theta np}$. By Lemma~\ref{lem:rie_hess_lip}, with at least the same probability, 
\begin{align*}
L_H \le C_1 \frac{n^{3/2}}{\mu^2} \log^{3/2} (np). 
\end{align*}
Set $\eps = \frac{\theta}{12\sqrt{2\pi}\mu L_H} <\mu/(4\sqrt{2})$, so 
\begin{align*}
\# N_\eps \le \exp\paren{n\log\frac{C_2 n^{3/2} \log^{3/2}(np)}{\theta}}. 
\end{align*}
Let $\event_H$ denote the event that 
\begin{align*}
\event_H \doteq \set{\max_{\mb q \in N_{\eps}} \norm{\mb U^* \Hess \expect{f(\mb q)} \mb U - \mb U^* \Hess f(\mb q) \mb U}{} \le \frac{\theta}{12\sqrt{2\pi}\mu} }. 
\end{align*}
On $\event_{\infty} \cap \event_H$, 
\begin{align*}
\sup_{\mb q: \norm{\mb w(\mb q)}{} \le \mu/(4\sqrt{2})} \norm{\mb U^* \Hess \expect{f(\mb q)} \mb U - \mb U^* \Hess f(\mb q) \mb U}{} \le \frac{\theta}{6\sqrt{2\pi}\mu}. 
\end{align*}
So on $\event_{\infty} \cap \event_H$, we have 
\begin{align}
\mb U^* \Hess f(\mb q) \mb U \succeq c_\sharp \frac{\theta}{\mu}
\end{align}
for any $c_\sharp \le 1/(12\sqrt{2\pi})$. We take $c_\sharp = c_\star$ for simplicity.  Setting $t = \frac{\theta}{12\sqrt{2\pi}\mu}$ in~\eqref{eq:strcvx_qsp_key6}, we obtain that for any fixed $\mb q$ in this region, 
\begin{align*}
\prob{\norm{\mb U^* \Hess \expect{f(\mb q)} \mb U - \mb U^* \Hess f(\mb q) \mb U}{} > t} \le 4n \exp\paren{-\frac{p\theta^2}{c_3n^2 + c_4n\theta/\mu }}. 
\end{align*}
Taking a union bound, we obtain that 
\begin{align*}
\prob{\event_H^c} \le 4n\exp\paren{-\frac{p\theta^2}{c_3n^2 + c_4n\theta/\mu } + C_5 n\log n + C_6 n\log\log p}. 
\end{align*}
It is enough to make $p \ge C_7 n^3\log (n/(\mu \theta))/(\mu \theta^2)$ to make the failure probability small, completing the proof. 
\end{enumerate}
\end{proof}

\subsection{Proof of Lemma~\ref{lem:alg_gradient_lb}} \label{proof:lem_alg_gradient_lb}
\begin{proof} 
For a given $\mb q$, consider the vector $\mb r \doteq \mb q - \mb e_n/q_n$. It is easy to verify that $\innerprod{\mb q}{\mb r} = 0$, and hence $\mb r \in T_{\mb q} \bb S^{n-1}$. Now, by \eqref{eq:fq_grad} and \eqref{eq:fq_rie_grad}, we have 
\begin{align*}
\innerprod{\grad f\paren{\mb q}}{\mb r} 
& = \innerprod{\paren{\mb I - \mb q \mb q^*} \nabla f\paren{\mb q}}{\mb q - \mb e_n/q_n} \\
& = \innerprod{\paren{\mb I - \mb q \mb q^*} \nabla f\paren{\mb q}}{- \mb e_n/q_n} \\
& = \frac{1}{p}\sum_{k=1}^p \innerprod{\paren{\mb I - \mb q \mb q^*} \tanh\paren{\frac{\mb q^* (\mb x_0)_k}{\mu}} (\mb x_0)_k}{-\mb e_n/q_n}   \\
& = \frac{1}{p} \sum_{k=1}^p \tanh\paren{\frac{\mb q^* (\mb x_0)_k}{\mu}} \paren{-\frac{x_k\paren{n}}{q_n} + \mb q^* (\mb x_0)_k} \\
& = \frac{1}{p} \sum_{k=1}^p \tanh\paren{\frac{\mb q^* (\mb x_0)_k}{\mu}} \paren{\mb w^*\paren{\mb q} \overline{\mb x}_k - \frac{x_k\paren{n}}{q_n}\norm{\mb w\paren{\mb q}}{}^2} \\
& = \mb w^*\paren{\mb q} \nabla g\paren{\mb w}, 
\end{align*}
where an explicit expression for $g(\mb w)$ can be found at the start of  Section IV in~\cite{sun2015complete_a}. Thus, 
\begin{align*}
\frac{\mb w^* \nabla g\paren{\mb w}}{\norm{\mb w}{}} = \frac{\innerprod{\grad f \paren{\mb q}}{\mb r}}{\norm{\mb w}{}} \le \norm{\grad f\paren{\mb q}}{} \frac{\norm{\mb r}{}}{\norm{\mb w}{}}, 
\end{align*}
where 
\begin{align*}
\frac{\norm{\mb r}{}^2}{\norm{\mb w}{}^2} = \frac{\norm{\mb w}{}^2 + \paren{q_n - \frac{1}{q_n}}^2}{\norm{\mb w}{}^2} = \frac{\norm{\mb w}{}^2 + \norm{\mb w}{}^4/q_n^2}{\norm{\mb w}{}^2} = \frac{1}{q_n^2} = \frac{1}{1-\norm{\mb w}{}^2} \le \frac{1}{1-\tfrac{1}{2000}} = \frac{2000}{1999}, 
\end{align*}
where we have invoked our assumption that $\norm{\mb w}{} \le \tfrac{1}{20\sqrt{5}}$. Therefore we obtain 
\begin{align*}
\norm{\grad f\paren{\mb q}}{} \ge \frac{\norm{\mb w}{}}{\norm{\mb r}{}} \frac{\mb w^* \nabla g\paren{\mb w}}{\norm{\mb w}{}} \ge \sqrt{\frac{1999}{2000}} \frac{\mb w^* \nabla g\paren{\mb w}}{\norm{\mb w}{}} \ge \frac{9}{10} \frac{\mb w^* \nabla g\paren{\mb w}}{\norm{\mb w}{}}, 
\end{align*}
completing the proof. 
\end{proof}

\subsection{Proof of Lemma~\ref{lem:TR-step}} \label{proof:lem_TR-step}
Proof of Lemma~\ref{lem:TR-step} combines the local Lipschitz property of $\Hess f(\mb q)$ in Lemma~\ref{lem:rie_hess_lip}, and the Taylor's theorem (manifold version, Lemma 7.4.7 of~\cite{absil2009}).  

\begin{proof}{\textbf{(of Lemma~\ref{lem:TR-step})}} 
Let $\gamma\paren{t}$ be the unique geodesic that satisfies $\gamma\paren{0} = \mb q^{(r)}$, $\gamma\paren{1} = \mb q^{(r+1)}$, and its directional derivative $\dot{\gamma}\paren{0} = \mb \delta_\star$. Since the parallel translation defined by the Riemannian connection is an isometry, then $\norm{\grad f (\mb q^{(r+1)}) }{} = \norm{\mc P_{\gamma}^{0 \leftarrow 1}\grad f(\mb q^{(r+1)}) }{}$. Moreover, since $\norm{\mb \delta_\star}{} \leq \Delta$, the unconstrained optimality condition in \eqref{eqn:ts-optimal-solution-1} implies that $\grad f(\mb q^{(r)}) + \Hess f(\mb q^{(r)}) \mb \delta_\star = \mb 0_{\mb q^{(r)}}$. Thus, by using Taylor's theorem in \cite{absil2009}, we have 
\begin{align*}
\norm{\grad f (\mb q^{(r+1)})}{} 
& = \norm{\mc P_{\gamma}^{0 \leftarrow 1}\grad f \paren{\mb q^{(r+1)}} - \grad f \paren{\mb q^{(r)}} - \Hess f\paren{\mb q^{(r)}} \mb \delta_\star}{} \\
& = \norm{\int_0^1 \brac{\mc P_{\gamma}^{0 \leftarrow t} \Hess f\paren{\gamma\paren{t}}\brac{\dot{\gamma}\paren{t}} -  \Hess f\paren{\mb q^{(r)}} \mb \delta_\star} \; dt}{} \; (\text{Taylor's theorem}) \\
& = \norm{\int_0^1 \paren{\mc P_{\gamma}^{0 \leftarrow t} \Hess f\paren{\gamma\paren{t}}\mc P_{\gamma}^{t \leftarrow 0} \mb \delta_\star -  \Hess f\paren{\mb q^{(r)}} \mb \delta_\star} \; dt}{}\\
& \le \norm{\mb \delta_\star}{} \int_0^1 \norm{\mc P_{\gamma}^{0 \leftarrow t} \Hess f\paren{\gamma\paren{t}}\mc P_{\gamma}^{t \leftarrow 0} -  \Hess f\paren{\mb q^{(r)}}}{} \; dt. 
\end{align*}
From the Lipschitz bound in Lemma~\ref{lem:rie_hess_lip} and the optimality condition in \eqref{eqn:ts-optimal-solution-1}, we obtain 
\begin{align*}
\norm{\grad f \paren{\mb q^{(r+1)}}}{}  
& \le \frac{1}{2}\norm{\mb \delta_\star}{}^2 L_H 
= \frac{L_H}{2m_H^2} \norm{\grad f \paren{\mb q^{(r)}}}{}^2.
\end{align*}
This completes the proof. 
\end{proof}

\subsection{Proof of Lemma~\ref{lem:TR-grad-opt}} \label{proof:lem_TR-grad-opt}
\begin{proof}
By invoking Taylor's theorem in \cite{absil2009}, we have
\begin{align*}
\mc P_{\gamma}^{0 \leftarrow \tau} \grad f \paren{\gamma\paren{\tau}} = \int_{0}^{\tau} \mc P_{\gamma}^{0 \leftarrow t} \Hess f\paren{\gamma\paren{t}} [\dot{\gamma}\paren{t}]\; dt. 
\end{align*}
Hence, we have 
\begin{align*}
\innerprod{\mc P_{\gamma}^{0 \leftarrow \tau} \grad f \paren{\gamma\paren{\tau}}}{\mb \delta} 
& = \int_{0}^{\tau} \innerprod{ \mc P_{\gamma}^{0 \leftarrow t} \Hess f\paren{\gamma\paren{t}} [\dot{\gamma}\paren{t}]}{\mb \delta}\; dt \\
& = \int_{0}^{\tau} \innerprod{ \mc P_{\gamma}^{0 \leftarrow t} \Hess f\paren{\gamma\paren{t}} [\dot{\gamma}\paren{t}]}{\mc P_{\gamma}^{0 \leftarrow t }\dot{\gamma}\paren{t}}\; dt \\
& = \int_{0}^{\tau} \innerprod{\Hess f\paren{\gamma\paren{t}} [\dot{\gamma}\paren{t}]}{\dot{\gamma}\paren{t}}\; dt\\
& \ge m_H \int_{0}^\tau \norm{\dot{\gamma}\paren{t}}{}^2 \; dt \ge m_H\tau \norm{\mb \delta}{}^2, 
\end{align*}
where we have used the fact that the parallel transport $\mc P_{\gamma}^{0 \leftarrow t}$ defined by the Riemannian connection is an isometry. On the other hand, we have 
\begin{align*}
\innerprod{\mc P_{\gamma}^{0 \leftarrow \tau} \grad f \paren{\gamma\paren{\tau}}}{\mb \delta} & \le \norm{\mc P_{\gamma}^{0 \leftarrow \tau} \grad f \paren{\gamma\paren{\tau}}}{} \norm{\mb \delta}{} = \norm{\grad f \paren{\gamma\paren{\tau}}}{} \norm{\mb \delta}{},
\end{align*}
where again used the isometry property of the operator $\mc P_{\gamma}^{0 \leftarrow \tau}$. Combining the two bounds above, we obtain 
\begin{align*}
\norm{\grad f \paren{\gamma\paren{\tau}}}{} \norm{\mb \delta}{} \ge m_H \tau \norm{\mb \delta}{}^2,  
\end{align*}
which implies the claimed result. 
\end{proof}

%% file: sec/proof_main.tex
\section{Proofs of Technical Results for Section~\ref{sec:main_result}} \label{sec:proof_main}

We need one technical lemma to prove Lemma~\ref{lem:alg_rounding_orth} and the relevant lemma for complete dictionaries. 
\begin{lemma} \label{lem:rounding-0}
For all integer $n_1 \in \N$, $\theta \in \paren{0, 1/3}$, and $n_2 \in \N$ with $n_2 \ge Cn_1\log\paren{n_1/\theta}/\theta^2$, any random matrix $\mb M \in \R^{n_1 \times n_2} \sim_{i.i.d.} \mathrm{BG}(\theta)$ obeys the following: for any fixed index set $\mc I \subset [n_2]$ with $\abs{\mc I} \leq \frac{9}{8} \theta n_2$, it holds that 
\begin{align*}
\norm{\mb v^* \mb M_{\mc I^c}}{1} - \norm{\mb v^* \mb M_{\mc I}}{1} \ge \frac{n_2}{6}\sqrt{\frac{2}{\pi}} \theta \norm{\mb v}{}\quad \text{for all}\; \mb v \in \R^{n_1}, 
\end{align*}
with probability at least 
$
1-cp^{-6}. 
$
Here $C, c$ are both positive constants. 
\end{lemma}

\begin{proof}
By homogeneity, it is sufficient to consider all $\mb v \in \bb S^{n_1}$. For any $i \in [n_2]$, let $\mb m_i\in \bb R^{n_1}$ be a column of $\mb M$. For a fixed $\mb v$ such that $\norm{\mb v}{} = 1$, we have
\begin{align*}
T\paren{\mb v} \doteq \norm{\mb v^* \mb M_{\mc I^c}}{1} - \norm{\mb v^* \mb M_{\mc I}}{1} = \sum_{i \in \mc I^c} \abs{\mb v^* \mb m_i} - \sum_{i \in \mc I} \abs{\mb v^* \mb m_i},
\end{align*}
namely as a sum of independent random variables. Since $\abs{\mc I} \le 9n_2\theta/8$, we have 
\begin{align*}
\expect{T\paren{\mb v}} \ge \paren{n_2 - \frac{9}{8} \theta n_2 - \frac{9}{8} \theta n_2} \expect{\abs{\mb v^* \mb m_1}} = \paren{1-\frac{9}{4} \theta } n_2 \expect{\abs{\mb v^* \mb m_1}} \geq \frac{1}{4}n_2 \expect{\abs{\mb v^* \mb m_1}},
\end{align*}
where the expectation $\expect{\abs{\mb v^* \mb m_1}}$ can be lower bounded as 
\begin{align*}
\expect{\abs{\mb v^* \mb m_1}} 
& \;=\; \sum_{k=0}^{n_1} \theta^k \paren{1-\theta}^{n_1 - k} \sum_{\mc J \in \binom{[n_1]}{k}} \bb E_{\mb g \sim \mc N\paren{\mb 0, \mb I}} \brac{\abs{\mb v^*_{\mc J} \mb g}}\\
& \;=\; \sum_{k=0}^{n_1} \theta^k \paren{1-\theta}^{n_1 - k} \sum_{\mc J \in \binom{[n_1]}{k}} \sqrt{\frac{2}{\pi}} \norm{\mb v_{\mc J}}{} 
\geq \sqrt{\frac{2}{\pi}} \norm{\bb E_{\mc J}\brac{\mb v_{\mc J}}}{} = \sqrt{\frac{2}{\pi}} \theta. 
\end{align*}
Moreover, by Lemma~\ref{lem:U-moments-bound} and Lemma \ref{lem:guassian_moment}, for any $i \in [n_2]$ and any integer $m \ge 2$, 
\begin{align*}
\expect{\abs{\mb v^* \mb m_i}^m} \le \bb E_{\mb Z \sim \mc N\paren{0, 1}} \brac{\abs{Z}^m } \le \paren{m-1}!! \le \frac{m!}{2}. 
\end{align*} 
So invoking the moment-control Bernstein's inequality in Lemma~\ref{lem:mc_bernstein_scalar}, we obtain 
\begin{align*}
\prob{T\paren{\mb v} < \frac{n_2}{4}\sqrt{\frac{2}{\pi}} \theta - t} \le \prob{T\paren{\mb v} < \expect{T\paren{\mb v}} - t} \le \exp\paren{-\frac{t^2}{2n_2 + 2t}}. 
\end{align*}
Taking $t = \tfrac{n_2}{20}\sqrt{\tfrac{2}{\pi}} \theta$ and simplifying, we obtain that 
\begin{align}
\prob{T\paren{\mb v} < \frac{n_2}{5}\sqrt{\frac{2}{\pi}} \theta} \le \exp\paren{-c_1 \theta^2 n_2}. 
\end{align}
Fix $\eps = \sqrt{\frac{2}{\pi}}\frac{\theta}{120} \brac{ n_1 \log \paren{n_1 n_2}}^{-1/2} < 1$. The unit sphere $\bb S^{n_1}$ has an $\eps$-net $N_\eps$ of cardinality at most $\paren{3/\eps}^{n_1}$. Consider the event
\begin{align*}
\event_{bg} \doteq \set{T\paren{\mb v} \ge \frac{n_2}{5}\sqrt{\frac{2}{\pi}} \theta\;\; \forall\; \mb v \in N_\eps}. 
\end{align*} 
A simple union bound implies 
\begin{align}\label{eqn:rounding-union}
\prob{\event_{bg}^c} 
& \le \exp\paren{-c_1\theta^2 n_2 + n_1\log\paren{\frac{3}{\eps}}}
\le \exp\paren{-c_1\theta^2 n_2 + c_2n_1\log \frac{n_1 \log n_2}{\theta}}.  
\end{align}
Conditioned on $\event_{bg}$, we have that any $\mb z \in \bb S^{n_1-1}$ can be written as $\mb z = \mb v + \mb e$ for some $\mb v \in N_\eps$ and $\norm{\mb e}{} \le \eps$. Moreover, 
\begin{align*}
T\paren{\mb z} 
&\; =\; \norm{\paren{\mb v + \mb e}^* \mb M_{\mc I^c}}{1} - \norm{\paren{\mb v + \mb e}^* \mb M_{\mc I}}{1} \ge T\paren{\mb v} - \norm{\mb e^* \mb M_{\mc I^c}}{1} - \norm{\mb e^* \mb M_{\mc I}}{1} \nonumber \\
& \;=\; \frac{n_2}{5}\sqrt{\frac{2}{\pi}} \theta - \norm{\mb e^* \mb M}{1} = \frac{n_2}{5}\sqrt{\frac{2}{\pi}} \theta - \sum_{k=1}^{n_2} \abs{\mb e^* \mb m_k} \nonumber \\
& \;\geq\;  \frac{n_2}{5}\sqrt{\frac{2}{\pi}} \theta - \eps \sum_{k=1}^{n_2} \norm{\mb m_k}{}. 
\end{align*} 
By Lemma~\ref{lem:X-infinty-tail-bound}, with probability at least $1-\theta\paren{n_1n_2}^{-7} - \exp\paren{-0.3\theta n_1 n_2}$, $\norm{\mb M}{\infty} \le 4\sqrt{\log\paren{n_1n_2}}$. Thus, 
\begin{align}
T\paren{\mb z} \ge \frac{n_2}{5}\sqrt{\frac{2}{\pi}} \theta - \sqrt{\frac{2}{\pi}}\frac{\theta}{120}\frac{n_2\sqrt{n_1} 4\sqrt{\log\paren{n_1n_2}}}{\sqrt{n_1} \sqrt{\log\paren{n_1 n_2}}} = \frac{n_2}{6}\sqrt{\frac{2}{\pi}} \theta. \label{eqn:rounding-lower-bound}
\end{align}
Thus, by \eqref{eqn:rounding-union}, it is enough to take $n_2 > Cn_1\log\paren{n_1/\theta}/\theta^2$ for sufficiently large $C > 0$ to make the overall failure probability small enough so that the lower bound \eqref{eqn:rounding-lower-bound} holds. 
\end{proof}

\subsection{Proof of Lemma~\ref{lem:alg_rounding_orth}} \label{proof:lem_alg_rounding_orth}
\begin{proof}
The proof is similar to that of \cite{qu2014finding}. First, let us assume the dictionary $\mb A_0 = \mb I$. W.l.o.g., suppose that the Riemannian TRM algorithm returns a solution $\widehat{\mb q}$, to which $\mb e_n$ is the nearest signed basis vector. Thus, the rounding LP~\eqref{eqn:LP_rounding} takes the form: 
\begin{align}\label{eqn:rounding-00}
\mini_{\mb q}\; \norm{\mb q^* \mb X_0}{1}, \quad \st \quad \innerprod{\mb r}{\mb q} = 1. 
\end{align}
where the vector $\mb r = \widehat{\mb q}$. Next, We will show whenever $\widehat{\mb q}$ is close enough to $\mb e_n$, w.h.p., the above linear program returns $\mb e_n$.
Let $\mb X_0 = \brac{\ol{\mb X}; \mb x_n^* }$, where $\ol{\mb X} \in \bb R^{(n-1)\times p}$ and $\mb x_n^*$ is the last row of $\mb X_0$. Set $\mb q= \brac{\overline{\mb q}, q_n}$, where $\overline{\mb q}$ denotes the first $n-1$ coordinates of $\mb q$ and $q_n$ is the last coordinate; similarly for $\mb r$. Let us consider a relaxation of the problem~\eqref{eqn:rounding-00},
\begin{align}
	\mini_{\mb q} \norm{\mb q^* \mb X_0}{1},\quad \st \quad q_nr_n+\innerprod{\overline{\mb q}}{\overline{\mb r}} \geq 1, \label{eqn:rounding-1}
\end{align}
It is obvious that the feasible set of \eqref{eqn:rounding-1} contains that of \eqref{eqn:rounding-00}. So if $\mb e_n$ is the unique optimal solution (UOS) of \eqref{eqn:rounding-1}, it is the UOS of \eqref{eqn:rounding-00}. Suppose $\mc I = \supp (\mb x_n)$ and define an event $\event_0 = \Brac{\abs{\mc I}\leq \frac{9}{8} \theta p }$. By Hoeffding's inequality, we know that
$
\prob{\event_0^c }  \le  \exp\paren{- \theta^2 p/2}. 
$
Now conditioned on $\event_0$ and consider a fixed support $\mc I$. \eqref{eqn:rounding-1} can be further relaxed as 
\begin{align} 
	\mini_{\mb q} \norm{\mb x_n}{1} \abs{q_n} - \norm{ \overline{\mb q}^* \ol{\mb X}_{\mc I} }{1} + \norm{\overline{\mb q}^* \ol{\mb X}_{\mc I^c} }{1},\quad \st \quad q_n r_n + \norm{\overline{\mb q}}{} \norm{\overline{\mb r}}{}\geq 1.\label{eqn:rounding-2}
\end{align}
The objective value of \eqref{eqn:rounding-2} lower bounds that of \eqref{eqn:rounding-1}, and are equal when $\mb q = \mb e_n$. So if $\mb q = \mb e_n$ is UOS of \eqref{eqn:rounding-2}, it is UOS of \eqref{eqn:rounding-00}. By Lemma \ref{lem:rounding-0}, we know that
\begin{align*}
	\norm{\overline{\mb q}^*\ol{\mb X}_{\mc I^c} }{1}- \norm{ \overline{\mb q}^* \ol{\mb X}_{\mc I} }{1}  \;\geq\; \frac{p}{6}\sqrt{\frac{2}{\pi}}\theta  \norm{\overline{\mb q}}{}
\end{align*}
holds w.h.p. when $p \ge C_1 (n-1)\log\paren{(n-1)/\theta}/\theta^2$. Let $\zeta = \frac{p}{6}\sqrt{\frac{2}{\pi}}\theta$, thus we can further lower bound the objective value in \eqref{eqn:rounding-2} by
\begin{align}
	\mini_{\mb q} \norm{\mb x_n}{1}\abs{q_n}+ \zeta \norm{\overline{\mb q} }{},\quad \st \quad q_n r_n + \norm{\overline{\mb q}}{} \norm{\overline{\mb r}}{}\geq 1. \label{eqn:rounding-3}
\end{align}
By similar arguments, if $\mb e_n$ is the UOS of \eqref{eqn:rounding-3}, it is also the UOS of \eqref{eqn:rounding-00}. For the optimal solution of \eqref{eqn:rounding-3}, notice that it is necessary to have $\sign\paren{q_n} = \sign\paren{r_n}$ and $q_n r_n + \norm{\overline{\mb q}}{} \norm{\overline{\mb r}}{}=1$. Therefore, the problem \eqref{eqn:rounding-3} is equivalent to
\begin{align}
	\mini_{q_n} \norm{\mb x_n}{1} \abs{q_n} + \zeta \frac{1-\abs{r_n} \abs{q_n} }{\norm{\overline{\mb r}}{}}, \quad \st \quad \abs{q_n}\leq \frac{1}{\abs{r_n}}. \label{eqn:rounding-4}
\end{align}
Notice that the problem \eqref{eqn:rounding-4} is a linear program in $\abs{q_n}$ with a compact feasible set, which indicates that the optimal solution only occurs at the boundary points $\abs{q_n}=0$ and $\abs{q_n} = 1/\abs{r_n}$. Therefore, $\mb q = \mb e_n$ is the UOS of \eqref{eqn:rounding-4} if and only if
\begin{align}
	\frac{1}{\abs{r_n}} \norm{\mb x_n}{1} < \frac{\zeta }{\norm{\overline{\mb r}}{} }. \label{eqn:rounding-5}
\end{align}
Conditioned on $\event_0$, by using the Gaussian concentration bound, we have
\begin{align*}
	\prob{\norm{\mb x_n}{1} \geq  \frac{9}{8}  \sqrt{\frac{2}{\pi}} \theta p + t } \;\leq\; \prob{\norm{\mb x_n}{1} \geq \expect{\norm{\mb x_n}{1}} + t } \;\leq\; \exp\paren{- \frac{t^2}{2p} },
\end{align*}
which means that
\begin{align}
	\prob{\norm{\mb x_n}{1} \geq \frac{5}{4} \sqrt{\frac{2}{\pi}} \theta p  } \;\leq\;\exp\paren{- \frac{\theta^2 p }{64 \pi } }. \label{eqn:rounding-6}
\end{align}
Therefore, by \eqref{eqn:rounding-5} and \eqref{eqn:rounding-6}, for $\mb q = \mb e_n$ to be the UOS of \eqref{eqn:rounding-00} w.h.p., it is sufficient to have
\begin{align}
	\frac{5}{4 \abs{r_n}} \sqrt{\frac{2}{\pi }} \theta p \;<\; \frac{\theta p}{6\sqrt{1-\abs{r_n}^2}} \sqrt{\frac{2}{\pi }},
\end{align}
which is implied by 
\begin{align*}
	\abs{r_n}\;>\; \frac{249}{250}. 
\end{align*}
The failure probability can be estimated via a simple union bound. Since the above argument holds uniformly for any fixed support set $\mc I$, we obtain the desired result. 

When our dictionary $\mb A_0$ is an arbitrary orthogonal matrix, it only rotates the row subspace of $\mb X_0$. Thus, w.l.o.g., suppose the TRM algorithm returns a solution $\widehat{\mb q}$, to which $\mb A_0 \mb q_\star$ is the nearest ``target'' with $\mb q_\star$ a signed basis vector. By a change of variable $\tilde{\mb q} = \mb A_0^* \mb q$, the problem \eqref{eqn:rounding-00} is of the form
\begin{align*}
	\mini_{\tilde{\mb q}} \norm{\widetilde{\mb q}^*\mb X_0 }{1},\quad \st \quad \innerprod{\mb A_0^* \mb r}{\tilde{\mb q} }=1,
\end{align*}
obviously our target solution for $\tilde{\mb q}$ is again the standard basis $\mb q_\star$. By a similar argument above, we only need $\innerprod{\mb A_0^* \mb r}{\mb e_n}>249/250$ to exactly recover the target, which is equivalent to 
$
	\innerprod{\mb r}{\widehat{\mb q}_\star} > 249/250.
$ 
This implies that our rounding \eqref{eqn:LP_rounding} is invariant to change of basis, completing the proof. 
\end{proof}

\subsection{Proof of Lemma~\ref{lem:alg_rounding_comp}} \label{proof:lem_alg_rounding_comp}
\begin{proof}
Define $\wt{\mb q} \doteq (\mb U \mb V^* + \mb \Xi)^* \mb q$. By Lemma~\ref{lem:pert_key_mag}, and in particular~\eqref{eq:pert_upper_bound}, when 
\begin{align*}
p \ge \frac{C_1}{c_\star^2 \theta} \max\set{\frac{n^4}{\mu^4}, \frac{n^5}{\mu^2}} \kappa^8\paren{\mb A_0} \log^4\paren{\frac{\kappa\paren{\mb A_0} n}{\mu \theta}}, 
\end{align*}
$\norm{\mb \Xi}{} \le 1/2$ so that $\mb U \mb V^* + \mb \Xi$ is invertible. 
Then the LP rounding can be written as 
\begin{align*} 
\mini_{\wt{\mb q}} \norm{\wt{\mb q}^*\mb X_0}{1},\quad \st \quad \innerprod{(\mb U \mb V^*+ \mb \Xi)^{-1}\mb r}{\wt{\mb q}} = 1. 
\end{align*}
By Lemma~\ref{lem:alg_rounding_orth}, to obtain $\wt{\mb q} = \mb e_n$ from this LP, it is enough to have 
\begin{align*}
\innerprod{(\mb U \mb V^*+ \mb \Xi)^{-1}\mb r}{\mb e_n} \ge 249/250, 
\end{align*}
and $p \ge C_2n^2\log(n/\theta)/\theta$ for some large enough $C_2$. This implies that to obtain $\mb q_\star$ for the original LP, such that $(\mb U \mb V^* + \mb \Xi)^* \mb q_\star = \mb e_n$, it is enough that 
\begin{align*}
\innerprod{(\mb U \mb V^*+ \mb \Xi)^{-1}\mb r}{(\mb U \mb V^* + \mb \Xi)^* \mb q_\star} = \innerprod{\mb r}{\mb q_\star} \ge 249/250, 
\end{align*}
completing the proof. 
\end{proof}

\subsection{Proof of Lemma~\ref{lem:deflation-bound}}  \label{proof:lem_deflation-bound}
\begin{proof}
Note that $[\mb q_\star^1, \dots, \mb q_\star^\ell] = (\mb Q^* + \mb \Xi^*)^{-1} [\mb e_1, \dots, \mb e_\ell]$, we have
\begin{align*}
\mb U^* (\mb Q + \mb \Xi) \mb X_0 
& = \mb U^* (\mb Q^* + \mb \Xi^*)^{-1} (\mb Q + \mb \Xi)^* (\mb Q + \mb \Xi) \mb X_0 \\
& = \mb U^* \left[\mb q_\star^1, \dots, \mb q_\star^\ell \;\vert\; \wh{\mb V}\right] (\mb I + \mb \Delta_1) \mb X_0, 
\end{align*}
where $\wh{\mb V} \doteq (\mb Q^* + \mb \Xi^*)^{-1} [\mb e_{\ell+1}, \dots, \mb e_n]$, and the matrix $\mb \Delta_1 = \mb Q^*\mb \Xi+ \mb \Xi^*\mb Q+\mb \Xi^*\mb \Xi$ so that $\norm{\mb \Delta_1}{}\leq 3\norm{\mb \Xi}{}$. Since $\mb U^* \left[\mb q_\star^1, \dots, \mb q_\star^\ell \;\vert\; \wh{\mb V}\right] = \brac{\mb 0 \; \vert\; \mb U^* \wh{\mb V}}$, we have
\begin{align} \label{eqn:deflation-comp-1}
\mb U^* (\mb Q + \mb \Xi) \mb X_0  = \brac{\mb 0 \; \vert\; \mb U^* \wh{\mb V}} \mb X_0 + \brac{\mb 0 \; \vert\; \mb U^* \wh{\mb V}} \mb \Delta_1 \mb X_0 = \mb U^* \wh{\mb V} \mb X_0^{[n-\ell]} + \mb \Delta_2 \mb X_0, 
\end{align}
where $\mb \Delta_2 = \brac{\mb 0 \; \vert\; \mb U^* \wh{\mb V}} \mb \Delta_1$. Let $\delta = \norm{\mb \Xi}{}$, so that
\begin{align}\label{eqn:delta-2-bound}
	\norm{\mb \Delta_2}{} \leq \frac{\norm{\mb \Delta_1}{}}{\sigma_{\min} \paren{\mb Q+\mb \Xi} } \leq \frac{3\norm{\mb \Xi}{} }{\sigma_{\min} \paren{\mb Q+\mb \Xi} }\leq \frac{3\delta}{1-\delta}.
\end{align}
Since the matrix $\wh{\mb V}$ is near orthogonal, it can be decomposed as $\wh{\mb V} = \mb V + \mb \Delta_3$, where $\mb V$ is orthogonal, and $\mb \Delta_3$ is a small perturbation. Obviously, $\mb V = \mb U \mb R$ for some orthogonal matrix $\mb R$, so that $\mb V$ spans the same subspace as that of $\mb U$. Next, we control the spectral norm of $\mb \Delta_3$: 
\begin{align}
	\norm{\mb \Delta_3}{} = \min_{\mb R \in O_{\ell}}\norm{\mb U \mb R - \wh{\mb V} }{} \le \min_{\mb R \in O_{\ell}} \norm{\mb U \mb R - \mb Q_{[n-\ell]} }{} + \norm{ \mb Q_{[n-\ell]} - \wh{\mb V} }{},
\end{align} 
where $\mb Q_{[n-\ell]}$ collects the last $n -\ell$ columns of $\mb Q$, i.e., $\mb Q=[\mb Q_{[\ell]}, \mb Q_{[n-\ell]}] $. 

To bound the second term on the right, we have 
\begin{align*}
\norm{\mb Q_{[n-\ell]} - \wh{\mb V} }{} \le \norm{\mb Q^{-1} - (\mb Q + \mb \Xi)^{-1}}{} \le \frac{\norm{\mb Q^{-1}}{} \norm{\mb Q^{-1} \mb \Xi}{}}{1-\norm{\mb Q^{-1} \mb \Xi}{}} \le \frac{\delta}{1-\delta} , 
\end{align*}
where we have used perturbation bound for matrix inverse (see, e.g., Theorem 2.5 of Chapter III in~\cite{stewart1990matrix}). 

To bound the first term, from Lemma~\ref{lem:sp_angle_norm}, it is enough to upper bound the largest principal angle $\theta_1$ between the subspaces $\text{span}([\mb q_\star^1, \dots, \mb q_\star^\ell])$, and that spanned by $\mb Q [\mb e_1, \dots, \mb e_\ell]$. Write $\mb I_{[\ell]} \doteq [\mb e_1, \dots, \mb e_\ell]$ for short, we bound $\sin \theta_1$ as
\begin{align*}
\sin \theta_1 \le\; & \norm{\mb Q \mb I_{[\ell]} \mb I_{[\ell]}^* \mb Q^* - (\mb Q^* + \mb \Xi^*)^{-1} \mb I_{[\ell]} \paren{\mb I_{[\ell]}^* (\mb Q + \mb \Xi)^{-1} (\mb Q^* + \mb \Xi^*)^{-1} \mb I_{[\ell]} }^{-1} \mb I_{[\ell]}^* (\mb Q + \mb \Xi)^{-1}}{} \\
=\; & \norm{\mb Q \mb I_{[\ell]} \mb I_{[\ell]}^* \mb Q^* - (\mb Q^* + \mb \Xi^*)^{-1} \mb I_{[\ell]} \paren{ \mb I_{[\ell]}^* (\mb I +\mb \Delta_1)^{-1} \mb I_{[\ell]} }^{-1} \mb I_{[\ell]}^* (\mb Q + \mb \Xi)^{-1}}{} \\
\le\; & \norm{\mb Q \mb I_{[\ell]} \mb I_{[\ell]}^* \mb Q^* - (\mb Q^* + \mb \Xi^*)^{-1} \mb I_{[\ell]}\mb I_{[\ell]}^* (\mb Q + \mb \Xi)^{-1}}{} \\
& \qquad + \norm{(\mb Q^* + \mb \Xi^*)^{-1} \mb I_{[\ell]} \brac{\mb I - \paren{ \mb I_{[\ell]}^* (\mb I +\mb \Delta_1)^{-1} \mb I_{[\ell]} }^{-1} }\mb I_{[\ell]}^* (\mb Q + \mb \Xi)^{-1} }{} \\
\le\; & \paren{1+\frac{1}{\sigma_{\min}(\mb Q + \mb \Xi)}}\norm{\mb Q^{-1} - (\mb Q + \mb \Xi)^{-1}}{} + \frac{1}{\sigma^2_{\min}(\mb Q + \mb \Xi)} \norm{\mb I - \paren{ \mb I_{[\ell]}^* (\mb I +\mb \Delta_1)^{-1} \mb I_{[\ell]} }^{-1} }{} \\
\le\; & \paren{1 + \frac{1}{1-\delta}} \frac{\delta}{1-\delta}  + \frac{1}{(1-\delta)^2} \frac{\norm{\mb I_{[\ell]}^*(\mb I+\mb \Delta_1)^{-1} \mb I_{[\ell]}-\mb I  }{}}{1-\norm{\mb I_{[\ell]}^*(\mb I+\mb \Delta_1)^{-1} \mb I_{[\ell]}-\mb I  }{}  } \\
\le\;& \paren{1 + \frac{1}{1-\delta}} \frac{\delta}{1-\delta}  + \frac{1}{(1-\delta)^2} \frac{ \norm{\mb \Delta_1}{} }{1-2\norm{\mb \Delta_1}{} },
\end{align*}
where to obtain the first line we used that for any full column rank matrix $\mb M$, $\mb M (\mb M^* \mb M)^{-1} \mb M^*$ is the orthogonal projector onto the its column span, and to obtain the fifth and six lines we invoked the matrix inverse perturbation bound again. Using $\delta < 1/20$ and $\norm{\mb\Delta_1}{} \le 3\delta < 1/2$, we have
\begin{align*}
\sin \theta_1 \leq \frac{(2-\delta)\delta}{(1-\delta)^2}	 + \frac{3\delta}{(1-\delta)^2(1-6\delta)} = \frac{5\delta - 13\delta^2+6\delta^3}{(1-\delta)^2(1-6\delta)}\leq 8\delta.
\end{align*}
For $\delta < 1/20$, the upper bound is nontrivial. By Lemma~\ref{lem:sp_angle_norm}, 
\begin{align*}
\min_{\mb R \in O_\ell}\norm{\mb U \mb R - \mb Q_{[n-\ell]} }{} \le \sqrt{2-2\cos\theta_1} \le \sqrt{2-2\cos^2\theta_1} = \sqrt{2} \sin \theta_1 \le 8\sqrt{2} \delta. 
\end{align*}
Put the estimates above, there exists an orthogonal matrix $\mb R \in O_\ell$ such that $\mb V = \mb U\mb R$ and $\wh{\mb V} = \mb V + \mb \Delta_3$ with 
\begin{align}\label{eqn:delta-3-bound}
	\norm{\mb \Delta_3}{} \leq \delta/(1-\delta) +8\sqrt{2}\delta\leq 25/(2\delta).
\end{align}
Therefore, by \eqref{eqn:deflation-comp-1}, we obtain
\begin{align}
	\mb U^*(\mb Q+\mb \Xi)\mb X_0 = \mb U^*\mb V\mb X_0^{[n-\ell]} +\mb \Delta,\quad \text{with}\quad\mb \Delta \doteq \mb U^* \mb \Delta_3 \mb X_0^{[n-\ell]} + \mb \Delta_2 \mb X_0.
\end{align}
By using the results in \eqref{eqn:delta-2-bound} and \eqref{eqn:delta-3-bound}, we get the desired result.
\end{proof}

%% file: sec/appendix.tex
\begin{appendices}
\section{Technical Tools and Basic Facts Used in Proofs}
\input{sec/app_tools}
\section{Auxillary Results for Proofs}

\input{sec/app_aux_geometry}
\end{appendices}

%% file: sec/app_tools.tex
In this section, we summarize some basic calculations that are useful throughout, and also record major technical tools we use in proofs. 
\begin{lemma}[Derivates and Lipschitz Properties of $h_{\mu}\paren{z}$] \label{lem:derivatives_basic_surrogate}
For the sparsity surrogate 
\begin{align*}
h_{\mu}\left(z\right)= \mu \log \cosh\paren{z/\mu}, 
\end{align*}
the first two derivatives are
\begin{align*}
\dot{ h}_\mu (z) = \tanh(z/\mu),\;  \ddot{h}_\mu (z) = \brac{1-\tanh^2(z/\mu)}/\mu. 
\end{align*}
Also, for any $z>0$, we have
\begin{align*}
(1-\exp(-2z/\mu)/2\;& \le \;\tanh(z/\mu) \;\le\; 1-\exp(-2z/\mu), \\
\exp(-2z/\mu)\;&\le \;1-\tanh^2(z/\mu) \;\le\; 4\exp(-2z/\mu). 
\end{align*}
Moreover, for any $z,~z^\prime\in \reals$, we have
\begin{align*}
|\dot{h}_{\mu}(z) - \dot{h}_{\mu}(z^\prime) | \le |z - z^\prime|/\mu,\; |\ddot{h}_{\mu}(z) - \ddot{h}_{\mu}(z^\prime) | \le 2|z - z^\prime|/\mu^2. 
\end{align*}
\end{lemma}

\begin{lemma}[Gaussian Tail Estimates] \label{lem:gaussian_tail_est}
Let $X \sim \mc N\paren{0, 1}$ and $\Phi\paren{x}$ be CDF of $X$. For any $x \ge 0$, we have the following estimates for $\Phi^c\paren{x} \doteq 1 - \Phi\paren{x}$: 
\begin{align*}
\paren{\frac{1}{x} - \frac{1}{x^3}}\frac{\exp\paren{-x^2/2}}{\sqrt{2\pi}} & \le \Phi^c\paren{x} \le \paren{\frac{1}{x} - \frac{1}{x^3} + \frac{3}{x^5}}\frac{\exp\paren{-x^2/2}}{\sqrt{2\pi}}, \quad (\text{Type I}) \\
\frac{x}{x^2 + 1} \frac{\exp\paren{-x^2/2}}{\sqrt{2\pi}}& \le \Phi^c\paren{x} \le \frac{1}{x} \frac{\exp\paren{-x^2/2}}{\sqrt{2\pi}},  \quad (\text{Type II}) \\
\frac{\sqrt{x^2 + 4} - x}{2} \frac{\exp\paren{-x^2/2}}{\sqrt{2\pi}} & \le  \Phi^c\paren{x} \le \paren{\sqrt{2 + x^2} - x}  \frac{\exp\paren{-x^2/2}}{\sqrt{2\pi}} \quad (\text{Type III}). 
\end{align*}
\end{lemma}
\begin{proof}
See proof of Lemma A.5 in the technical report~\cite{sun2015complete_tr}. 
\end{proof}

\begin{lemma}[Moments of the Gaussian RV] \label{lem:guassian_moment}
If $X \sim \mc N\left(0, \sigma^2\right)$, then it holds for all integer $p \geq 1$ that
\begin{align*}
\expect{\abs{X}^p} \leq \sigma^p \paren{p -1}!!. 
\end{align*}
\end{lemma}

\begin{lemma}[Moments of the $\chi^2$ RV] \label{lem:chi_sq_moment}
If $X \sim \mc \chi^2\paren{n}$, then it holds for all integer $p \geq 1$ that
\begin{align*}
\expect{X^p} = 2^p \frac{\Gamma\paren{p + n/2}}{\Gamma\paren{n/2}} =  \prod_{k=1}^p (n+2k-2) \le p!\paren{2n}^p/2. 
\end{align*}
\end{lemma}


\begin{lemma}[Moment-Control Bernstein's Inequality for Scalar RVs, Theorem 2.10 of~\cite{foucart2013mathematical}] \label{lem:mc_bernstein_scalar}
Let $X_1, \dots, X_p$ be i.i.d. real-valued random variables. Suppose that there exist some positive numbers $R$ and $\sigma^2$ such that
\begin{align*}
\expect{\abs{X_k}^m} \leq m! \sigma^2 R^{m-2}/2, \; \; \text{for all integers}\; m \ge 2. 
\end{align*} 
Let $S \doteq \frac{1}{p}\sum_{k=1}^p X_k$, then for all $t > 0$, it holds  that 
\begin{align*}
\prob{\abs{S - \expect{S}} \ge t} \leq 2\exp\left(-\frac{pt^2}{2\sigma^2 + 2Rt}\right).   
\end{align*}
\end{lemma}

\begin{lemma}[Moment-Control Bernstein's Inequality for Matrix RVs, Theorem 6.2 of~\cite{tropp2012user}] \label{lem:mc_bernstein_matrix}
Let $\mb X_1, \dots, \mb X_p\in \R^{d \times d}$ be i.i.d. random, symmetric matrices. Suppose there exist some positive number $R$ and $\sigma^2$ such that
\begin{align*}
\expect{\mb X_k^m} \preceq m! \sigma^2 R^{m-2}/2 \cdot \mb I \; \text{and} -\expect{\mb X_k^m} \preceq m! \sigma^2 R^{m-2}/2 \cdot \mb I\;,  \; \text{for all integers $m \ge 2$}. 
\end{align*}
Let $\mb S \doteq \frac{1}{p} \sum_{k = 1}^p \mb X_k$, then for all $t > 0$, it holds that 
\begin{align*}
\prob{\norm{\mb S - \expect{\mb S}}{} \ge t} \le 2d\exp\paren{-\frac{pt^2}{2\sigma^2 + 2Rt}}.
\end{align*}
\end{lemma}
\begin{proof}
See proof of Lemma A.10 in the technical report~\cite{sun2015complete_tr}. 
\end{proof}

\begin{lemma}[Integral Form of Taylor's Theorem]\label{lem:Taylor-integral-form}
	Let $f(\mb x): \bb R^n \mapsto \bb R$ be a twice continuously differentiable function, then for any direction $\mb y\in \bb R^n$, we have
	\begin{align*}
		f(\mb x+t\mb y) &= f(\mb x) + t \int_0^1\innerprod{\nabla f(\mb x+st\mb y)}{\mb y} \; ds, \\
		f(\mb x+t\mb y) &= f(\mb x) + t \innerprod{\nabla f(\mb x)}{\mb y} + t^2 \int_0^1 (1-s)\innerprod{\nabla^2 f(\mb x+st\mb y) \mb y}{\mb y}\; ds.
	\end{align*}
\end{lemma}

%% file: sec/app_aux_geometry.tex
\begin{lemma} \label{lem:bg_identity_diff}
There exists a positive constant $C$ such that for any $\theta \in \paren{0, 1/2}$ and $n_2 > C n_1^2 \log n_1$, the random matrix $\mb X \in \R^{n_1 \times n_2}$ with $\mb X\sim_{i.i.d.} \mathrm{BG}\paren{\theta}$ obeys 
\begin{align}
\norm{\frac{1}{n_2 \theta} \mb X \mb X^* - \mb I}{} \le 10\sqrt{\frac{\theta n_1 \log n_2}{n_2}}
\end{align}
with probability at least $1 - n_2^{-8}$. 
\end{lemma}
\begin{proof}
See proof of Lemma B.3 in~\cite{sun2015complete_a}. 
\end{proof}

\begin{lemma} \label{lem:sp_angle_norm}
Consider two linear subspaces $\mc U$, $\mc V$ of dimension $k$ in $\R^n$ ($k \in [n]$) spanned by orthonormal bases $\mb U$ and $\mb V$, respectively. Suppose $\pi/2 \ge \theta_1 \ge \theta_2 \dots \ge \theta_k \ge 0$ are the principal angles between $\mc U$ and $\mc V$. Then it holds that \\
i) $\min_{\mb Q \in O_k} \norm{\mb U - \mb V \mb Q}{} \le \sqrt{2-2\cos \theta_1}$; \\
ii) $\sin \theta_1 = \norm{\mb U\mb U^* - \mb V\mb V^*}{}$;\\
iii) Let $\mc U^\perp$ and $\mc V^\perp$ be the orthogonal complement of $\mc U$ and $\mc V$, respectively. Then $\theta_1(\mc U, \mc V) = \theta_1(\mc U^\perp, \mc V^\perp)$. 
\end{lemma}
\begin{proof}
See proof of Lemma B.4 in~\cite{sun2015complete_tr}. 
\end{proof}
Below are restatements of several technical results from~\cite{sun2015complete_a} that are important for proofs in this paper. 
\begin{proposition}[Hessian Lipschitz]\label{prop:lip-hessian-negative}
Fix any $\rconcave \in \paren{0, 1}$. Over the set $\Gamma \cap \set{\mb w: \norm{\mb w}{} \ge \rconcave}$, $\mb w^* \nabla^2 g(\mb w; \mb X_0) \mb w/\norm{\mb w}{}^2$ is $\Lconcave$-Lipschitz with 
\begin{align*}
\Lconcave \le \frac{16n^3}{\mu^2} \norm{\mb X_0}{\infty}^3 + \frac{8n^{3/2}}{\mu \rconcave} \norm{\mb X_0}{\infty}^2 + \frac{48 n^{5/2} }{\mu} \norm{\mb X_0}{\infty}^2 + 96 n^{5/2} \norm{\mb X_0}{\infty}.
\end{align*}
\end{proposition} 

\begin{proposition}[Gradient Lipschitz]\label{prop:lip-gradient}
Fix any $r_g \in \paren{0, 1}$. Over the set $\Gamma \cap \set{\mb w: \norm{\mb w}{} \ge r_g}$, $\mb w^* \nabla g( \mb w; \mb X_0 )/\norm{\mb w}{}$ is $L_g$-Lipschitz with 
\begin{align*}
L_g \le \frac{2 \sqrt{n} \norm{\mb X_0}{\infty}}{r_g} + 8 n^{3/2} \norm{\mb X_0}{\infty} + \frac{4 n^2}{\mu} \norm{\mb X_0}{\infty}^2.
\end{align*}
\end{proposition} 

\begin{proposition}[Lipschitz for Hessian around zero]\label{prop:lip-hessian-zero}
Fix any $\rconvex \in (0, 1/2)$. Over the set $\Gamma \cap \set{\mb w: \norm{\mb w}{} \le \rconvex}$, $\nabla^2 g( \mb w; \mb X_0)$ is $\Lconvex$-Lipschitz with  
\begin{align*}
\Lconvex \;\le\;\frac{4n^2}{\mu^2} \norm{\mb X_0}{\infty}^3+\frac{4n}{\mu}\norm{\mb X_0}{\infty}^2 +  \frac{8\sqrt{2}\sqrt{n}}{\mu}\norm{\mb X_0}{\infty}^2 + 8\norm{\mb X_0}{\infty}. 
\end{align*}
\end{proposition}

\begin{lemma}\label{lem:X-infinty-tail-bound}
For any $\theta \in \paren{0, 1}$, consider the random matrix $\mb X \in \R^{n_1 \times n_2}$ with $\mb X \sim_{i.i.d.} \mathrm{BG}\paren{\theta}$. Define the event $\event_\infty \doteq \Brac{ 1 \le \norm{\mb X}{\infty}\leq 4\sqrt{\log\paren{n p}}}$. It holds that 
\begin{align*}
\prob{\event_\infty^c} \leq \theta \paren{np}^{-7} + \exp\paren{-0.3\theta np}. 
\end{align*}
\end{lemma}

\begin{lemma} \label{lem:pert_key_mag}
For any $\theta \in \paren{0, 1/2}$, suppose $\mb A_0$ is complete with condition number $\kappa\paren{\mb A_0}$ and $\mb X_0 \sim_{i.i.d.} \mathrm{BG}\paren{\theta}$. Provided $p \ge C\kappa^4\paren{\mb A_0} \theta n^2 \log (n \theta \kappa\paren{\mb A_0})$, one can write $\overline{\mb Y}$ as defined in~\eqref{eq:precon_def} as 
\begin{align*}
\overline{\mb Y} = \mb U \mb V^* \mb X_0 + \mb \Xi \mb X_0, 
\end{align*}
for a certain $\mb \Xi$ obeying $\norm{\mb \Xi}{} \le 20\kappa^4\paren{\mb A} \sqrt{\frac{\theta n \log p}{p}}$, with probability at least $1-p^{-8}$. Here $\mb U \mb \Sigma \mb V^* = \mathtt{SVD}\paren{\mb A_0}$, and $C$ is a positive numerical constant. 
\end{lemma}

\begin{lemma}\label{lem:U-moments-bound}
Suppose $\mb z, \mb z' \in \R^n$ are independent and obey $\mb z \sim_{i.i.d.} \mathrm{BG}\paren{\theta}$ and $\mb z' \sim_{i.i.d.} \mc N\paren{0, 1}$. Then, for any fixed vector $\mb v \in \R^n$, it holds that 
\begin{align*}
\expect{\abs{\mb v^* \mb z}^m} & \le \expect{\abs{\mb v^* \mb z'}^m } = \bb E_{Z \sim \mc N\paren{0, \norm{\mb v}{}^2}}\brac{\abs{Z}^m}, \\
\expect{\norm{\mb z}{}^m} & \le \expect{\norm{\mb z'}{}^m}, 
\end{align*}
for all integers $m \ge 1$. 
\end{lemma}